\documentclass[sigconf,nonacm]{acmart}

\usepackage{graphicx} 
\usepackage[utf8]{inputenc}
\usepackage{lipsum}
\usepackage{caption}
\usepackage{subcaption}
\usepackage[english]{babel}
\usepackage[dvipsnames]{xcolor}
\usepackage{tabularx}
\usepackage{tabularray}
 
\usepackage{amsmath,amsfonts,amsthm, amssymb}
\usepackage{physics}
\usepackage{enumitem}
\UseTblrLibrary{diagbox}
\usepackage{thmtools}
\usepackage{thm-restate}
\usepackage{algorithm}
\usepackage{algpseudocode}
\usepackage{appendix}
\usepackage{multirow} 
\usepackage{subfiles} 
\usepackage[operators,sets,asymptotics, probability, advantage]{cryptocode} 
\usepackage{changepage}
\usepackage{geometry}
\usepackage{adjustbox}
\usepackage[svgnames]{xcolor}

\newtheorem{theorem}{Theorem}[section]
\newtheorem{corollary}[theorem]{Corollary}
\newtheorem{lemma}[theorem]{Lemma}

\newtheorem{remark}[theorem]{Remark}

\usepackage[capitalise]{cleveref}

\crefname{lemma}{Lemma}{Lemmas}
\Crefname{lemma}{Lemma}{Lemmas}
\crefname{claim}{Claim}{Claims}
\Crefname{claim}{Claim}{Claims}

\usepackage[table]{xcolor}
\definecolor{lightgray}{gray}{0.9}

\let\oldforall\forall
\renewcommand{\forall}{\, \oldforall \, }
\let\oldexist\exists
\renewcommand{\exists}{\: \oldexist \: }

\newcounter{protocol}
\makeatletter

\newcounter{construction}
\makeatletter

\theoremstyle{definition}
\newtheorem{definition}{Definition}[section]

\newcommand{\bfB}{\mathbf{B}}
\newcommand{\bfR}{\mathbf{R}}
\newcommand{\bfT}{\mathbf{T}}

\newcommand{\NTT}{\mathrm{NTT}}
\newcommand{\bfA}{\mathbf{A}}

\newcommand{\com}{\mathbf{com}}
\newcommand{\ch}{\mathbf{ch}}
\newcommand{\resp}{\mathbf{resp}}

\newcommand{\pk}{\mathbf{pk}}
\newcommand{\sk}{\mathbf{sk}}

\newcommand{\calP}{\mathcal{P}}
\newcommand{\calV}{\mathcal{V}}
\newcommand{\Cset}{\mathcal{C}}

\newcommand{\adv}{\mathcal{A}}
\newcommand{\CalRq}{\mathcal{R}_q}
\newcommand{\msis}{\mathrm{MSIS}}
\newcommand{\mlwe}{\mathrm{MLWE}}
\newcommand{\rej}{\texttt{Rej}}

\newcommand{\negll}{\mathsf{negl}}
\newcommand{\dotprod}[2]{\langle #1, #2 \rangle}
\newcommand{\simulator}[0]{\mathcal{S}_\Pi}
\newcommand{\distinguisher}[0]{\mathcal{D}_\Pi}
\newcommand{\extractor}[0]{\mathcal{E}_\Pi}
\newcommand{\QPT}{\mathsf{QPT}}

\newcommand{\FS}{\mathsf{FS}}
\newcommand{\Rq}{\mathcal{R}_q}
\newcommand{\Prb}{\mathbb{P}}
\newcommand{\E}{\mathbb{E}}

\usepackage{xcolor}

\createprocedureblock{procb}{center, boxed}{}{}{}
\createpseudocodeblock{pcb}{center, boxed}{}{}{}

\createpseudocodeblock{pcbprotocol}{center, boxed}{}{}{}
\createpseudocodeblock{pcnobox}{}{}{}{}
\createpseudocodeblock{pcnocenter}{boxed}{}{}{}

\newcommand{\name}{\mathsf{PQ-TaDL}}
\newcommand{\nameextension}{\mathsf{PQ-TaDL}_{Compact}}

\newenvironment{rev}{\begingroup\color{black}}{\endgroup}

\newcommand{\codecomment}[1]{\textcolor{gray}{// \text{#1}}}

\newcommand{\auditscheme}{\Pi_\mathsf{Audit}}
\newcommand{\auditprove}{\mathsf{AuditProve}}
\newcommand{\auditverify}{\mathsf{AuditVerify}}
\newcommand{\publicparameter}{\mathsf{pp}}

\newcommand{\etl}{\mathsf{ETL}}
\newcommand{\setup}{\mathsf{Setup}}
\newcommand{\publicparam}{\mathsf{pp}}
\newcommand{\keygen}{\mathsf{KeyGen}}
\newcommand{\checkbalance}{\mathsf{CheckBalance}}
\newcommand{\extract}{\mathsf{Extract}}
\newcommand{\createtx}{\mathsf{CreateTx}}
\newcommand{\verifytx}{\mathsf{VerifyTx}}
\newcommand{\ledger}{\mathcal{L}}
\newcommand{\tx}{\mathsf{tx}}
\newcommand{\txlist}{\mathcal{TX}}
\newcommand{\asset}{\mathsf{a}}
\newcommand{\valuelist}{\mathsf{V}}
\newcommand{\assetlist}{\mathcal{AS}}
\newcommand{\valueset}{{\mathcal{V}^{\dagger}_{\ledger}}}
\newcommand{\sklist}{\mathsf{SK}}
\newcommand{\pklist}{\mathcal{{PT}}}
\newcommand{\bank}{\mathsf{{BN}}}

\newcommand{\expbalance}{\mathsf{{Balance}}}
\newcommand{\expprivacy}{\mathsf{{Privacy}}}

\newcommand{\nizksetup}{\sf{ZKSetup}}
\newcommand{\nizkprove}{\sf{ZKProve}} 
\newcommand{\nizkverify}{\sf{ZKVerify}} 
\newcommand{\nizkproof}{\pi_{\sf{ZKP}}} 
\newcommand{\crs}{\mathsf{crs}} 
\newcommand{\stmt}{\mathsf{stmt}} 
\newcommand{\wit}{\mathsf{wit}} 
 
\newcommand{\lang}{L}
\newcommand{\relation}{R} 
\newcommand{\zksimulator}{\mathcal{SIM}} 

\newcommand{\commitgen}{{\sf CKeyGen}}
\newcommand{\commitkey}{{\sf ctk}}

\newcommand{\commit}{{\sf Com}}
\newcommand{\open}{{\sf Open}}
\newcommand{\commitment}{\mathbf{com}}
\newcommand{\hidingadv}{{\sf Hid_{}^{\adv}}(\secpar)}

\newcommand{\bindingadv}{{\sf Bind_{}^{\adv}}(\secpar)}

\newcommand{\GAME}{\mbox{G}}
\providecommand{\gamedes}[1]{\noindent\underline{\textbf{Game $\GAME_{#1}$:}}}

\newcommand{\tab}{\ \ \ \ }
\newcommand{\rangebound}{N}
\newcommand{\ntt}{\mathsf{NTT}}

\newcommand{\z}{\mathbf{z}}
\newcommand{\zb}{\mathcal{Z}_{c_i}^{\pi^B}}
\newcommand{\e}{\textbf{e}}

\title{A Practical Post-Quantum Distributed Ledger Protocol for Financial Institutions}

\author{Yeoh Wei Zhu, Naresh Goud Boddu, Yao Ma, Shaltiel Eloul, Giulio Golinelli, Yash Satsangi, \\ Rob Otter, Kaushik Chakraborty \\Global Technology Applied Research, JPMorganChase}

\renewcommand{\shortauthors}{Zhu et al.}

\begin{document}


\begin{abstract}
Traditional financial institutions face inefficiencies that can be addressed by distributed ledger technology. However, a primary barrier to adoption is the privacy concerns surrounding publicly available transaction data. Existing private protocols for distributed ledger that focus on the RingCT model are not suitable for adoption for financial institutions. We propose a post-quantum, lattice-based transaction scheme for encrypted ledgers which better aligns with institutions' requirements for confidentiality and audit-ability. The construction leverages various zero-knowledge proof techniques, and introduces a new method for equating two commitment messages, without the capability to open one of the commitment during the re-commitment. Subsequently, we build a publicly verifiable transaction scheme that is efficient for single or multi-assets, by introducing a new compact range-proof. The techniques used and the proofs constructed could be of independent interest.
\end{abstract}


\maketitle

\section{Introduction}
\label{sec:intro}

Distributed ledger technology (DLT) introduces the possibility to transact without the involvement of centralized or intermediary bodies. However, it does not necessarily address the privacy and auditing needs of financial institutions.
Numerous works \cite{narula2018zkLedger,PADL} have explored the adoption of DLT in this context of traditional financial systems. The key requirements of financial institutions are hiding the transaction value, input/output account anonymity, asset anonymity, native support for atomic exchange, and a customized auditability, as discussed in \cite{PADL}. 
The distributed ledgers described in \cite{PADL, narula2018zkLedger, chatzigiannis2021miniledger} can be interpreted as encrypted table-based distributed ledgers (ETL) \cite{Zhu2026} that are based on the account model rather than the unspent transaction output (UTXO) model \cite{utxo1, utxo2}. 
However, these works have so far been built upon the foundation of the discrete log assumption, which is vulnerable to the threat of quantum computing \cite{Shor1995-mn}.
Meanwhile, there also exists a type of distributed ledger based on the ring confidential transaction (RingCT) scheme \cite{ringct2, ringct3}. RingCT relies on stealth addresses \cite{stealthaddress2} and ring signatures \cite{ringsignature} for anonymity, while relying on unique serial numbers to prevent double-spending. Despite the existence of lattice-based RingCT schemes \cite{matrict_2_2022, smile_tx_2021}, RingCT lacks the properties required for traditional financial institution use cases. 
\begin{rev} 
In the discrete log setting, the prominent cryptocurrency example for RingCT is the Monero scheme \cite{monero} and for ETL is the Solana CT scheme \cite{solana_confidential_transfer} where the Solana CT can be interpreted as a sparse 1-to-1 single asset ETL.
\end{rev}

\textbf{Advantage of ETL over RingCT.} 
This work shows that rather than RingCT scheme, the encrypted table-based ledger (ETL) is a practical privacy-preserving scheme for financial institutions in the post-quantum setting.
Let the transaction be a row entry and the participant's account be a column in the ledger, an account state in ETL is a simple sum over the transaction entries in the account's column. The account balance can be succinctly represented as shown in \cite{chatzigiannis2021miniledger} via pruning. 
On the other hand, the use of stealth addresses in RingCT and the need of scanning, complicate zero-knowledge-based auditing, as the history of transactions must be included to avoid missing any potential token with respect to the account viewing key. 
Although the number of assets held and traded by financial institution is in the range of hundreds/thousands (different currencies, digitized bonds, collateral assets, etc.), current RingCT solutions mainly focus on a typical blockchain setting where there is only one type of asset being transfer.
Further advantages of ETL over RingCT for financial institutions are discussed in \cite{Zhu2026}. 

\textbf{Multi-Asset Uses Case: Repurchase Agreement (Repo) Tokenization.}
We highlight the advantages by demonstrating concrete use cases where the advantages of ETL can be observed.
In the context of typical cryptocurrencies such as Ethereum, a multi-asset token standard has been proposed in ERC-1155 \cite{ethereum_erc1155} for Ethereum which is a multi-asset extension of ERC-20 token that enables one logical token unit to hold multiple asset values.  In the financial institution context, multi-asset use cases can be demonstrated through FX, Bond/Debit market, and the repurchase agreement use case \cite{baklanova2015repo} which is commonly known as repo. In repo, two parties agree on the fact that one of the parties will trade securities collateral to the other party at the agreed price. The selling party is committed to buy back the sold securities at a usually higher agreed price where the price difference is the interest. The exchange of multiple different types of securities with other asset(s) is exactly captured by the multi-asset transaction ETL. Since ETL natively supports multi-asset transaction with asset anonymity, we are able to perform the repo transaction for large assets portfolio, while still hiding the exact type of securities being traded in a single transaction. Specifically, our work is able to support thousand(s) type of securities being traded/anonymized in a single transaction efficiently.  To achieve atomic exchange, one can extend ETL with atomic swap by using the technique from \cite{PADL}. 
A smart contract layer can be used to manage the lifecycle and contractual conditions, whereas the ETL validates and records the resulting confidential asset transfers while preserving auditability and ledger consistency.




\subsection{Contribution}

Our main contribution is the design and analysis of a lattice-based ETL-like transaction scheme. Our construction builds upon lattice-based range proofs \cite{LNP22, LNPS21}, verifiable extraction for commitment \cite{LNPS21}, and lattice-based zero-knowledge proof (ZKP) \cite{ALS20, BDLOP16}.
Our new results enable proving commitments' equivalence without the knowledge of randomness used in the original commitment. 
The proof of equivalence is made possible through a combination of two techniques. First technique shows that if the committed value $v$ and $\sqrt q v$ (where $q$ is the prime modulus) are committed under two valid MLWE samples respectively which constituted the commitment key, then it is possible to show that any difference would be amplified by the $\sqrt q$ factor (or the lack-thereof) and be detectable by the norm bound check assuming the secret $s$ and random value $r$ are short. Second technique is used to overcome the deficit of weak opening that does not guarantee that the extracted random value $r^*$ is short by providing a proof that the random value used in the commitment is short through an approximate range proof. 
Subsequently, it can be conclusively shown that the committed values are equal if the verification passes.
Compared to other schemes, our construction supports multi-asset transaction, efficient zero-knowledge auditing of account state history due to a table-like structure, verifiable spendable-token through a proof of commitment extractability. Beyond showing security in the ROM model which is typically done in lattice-based constructions \cite{matrict_2_2022,smile_tx_2021}, we take a step further to provide analysis in the QROM model. 

The contributions are summarized as follows:
{
\begin{itemize}
    \item To the best our knowledge, proposed $\name$ is the first lattice-based transaction scheme for the encrypted table-based ledger (ETL). The proposed transaction scheme supports multi-asset transactions, verifiable spendable-token, and efficient ZK auditing over the account state history.
    \item We construct various ZKPs for transaction relations. Importantly, we show a new technique for asserting two commitments shares the same committed value in zero-knowledge but without the knowledge of random value of the original commitment. The technique used can be of independent interest. This is an essential proof to obtain a publicly verifiable transaction protocol in the ETL setting.
    \item We further propose $\nameextension$, a \textit{compact} multi-asset transaction scheme in the ETL setting that scales with the polynomial degree. $\nameextension$ achieve $4.27$ms proving and $0.37$ms verification time with $1.87$KB transaction size when amortized across $1024$ assets.
    \item A security analysis is provided for the proposed scheme, notably also analysis in the QROM settings. 
\end{itemize}
}
\subsection{Related Works}
\textbf{Privacy-preserving Transaction Scheme.} 
Privacy preserving transaction schemes may be grouped into two categories, namely RingCT scheme \cite{ringct2,ringct3} and Encrypted Table-based Ledger (ETL) scheme \cite{PADL,chatzigiannis2021miniledger, narula2018zkLedger}. 
Transaction schemes based on RingCT such as Monero \cite{monero} offers attractive privacy features when executing transactions, mainly attributed to the use of ring signature \cite{ringsignature} and stealth address \cite{stealthaddress2}. However, the use of stealth address complicates zero-knowledge based auditing for financial institutions as to prove a statement in zero-knowledge about the account state with respect to a public key, one needs to include the entire transaction history for completeness, such that no transaction made to the wallet public key is excluded from the audit. Moreover, most of the RingCT schemes do not provide a proof that the token key (which is needed to spend the token) for the stealth address is well-formed and rely on a gentlemen agreement for the communication of the token key. 
In the realm of ETL,
zkLedger \cite{narula2018zkLedger} and MiniLedger \cite{chatzigiannis2021miniledger} offer efficient zero-knowledge account state auditing. 
More recently, the lack of support for native atomic-exchange for multi-asset transaction and the lack of support for active transfer with both the receiver and the spender approvals was addressed in \cite{PADL} for the ETL setting. While these ETL proposals drop the usage of ring signature and stealth address to solve some of the issues of completeness for zero-knowledge disclosure, this increases the computational and communication complexity of the resulting schemes. Despite this, the scheme is reasonable for deployment for financial institution's use cases as the number of participants is a few orders of magnitude smaller than the traditional blockchain or distributed ledger. 
Despite their usefulness, these works are based on discrete-log setting which is vulnerable to the threat of quantum computing \cite{Shor1995-mn}. 

\textbf{Lattice-based Transaction Scheme.} Transitioning into a post-quantum setting, most of the existing work focuses on the instantiation of a transaction scheme for the RingCT like setting. MatRiCT \cite{matrict_2_2022}, SMILE \cite{smile_tx_2021}, and others \cite{Gao2025-se} are all examples of a lattice-based RingCT transaction scheme. Similar to the RingCT scheme based on discrete-log setting, these lattice-based RingCT suffers from inefficiency for proving statements about an account history during auditing due to the use of stealth address and its scanning requirement. 
Moreover, these schemes do not provide proof that a sent token can always be spent by the receiver. The token (spending) key is encrypted and given to receiver in an ad-hoc manner without the ledger verifying the well-formness of this token key with respect to the token. The always-spendable token property is crucial for an institutional use case as the token held by the institution would be reflected on its account balance. 
Furthermore, none of the existing works on lattice-based transaction scheme shows security in the quantum random oracle (QROM) model \cite{qrom} and only relies on the lattice-based assumptions to argue for their post-quantum security.  

In this work, we focus on constructing a post-quantum secure ETL-based transaction scheme called $\name$ (Post-Quantum Table based Distributed Ledger). Compared to other existing lattice-based transaction schemes, $\name$ supports multi-asset transaction, efficient zero-knowledge auditing for the account state, achieving comparable performance under realistic setting, verifiable spendable-token transaction, privacy-preserving transaction, while offering security analyzed in the QROM model.

\section{Technical Overview}
In this section, we present a high-level overview of the design and construction of the transaction scheme for a post-quantum private distributed ledger.
\begin{rev}
Section \ref{sec:etl_transaction} gives a high-level overview of encrypted table-based ledger (ETL), a high-level abstraction of the table-like encrypted transaction schemes used in MiniLedger \cite{chatzigiannis2021miniledger}, zkLedger \cite{narula2018zkLedger}, and PADL \cite{PADL}. The encrypted table-based ledger (ETL) abstraction notion is first explored in \cite{Zhu2026}, which ties various encrypted table-based transaction schemes together under one notion.
Sections \ref{sec:overview_pob}-\ref{sec:overview_rangeproof} introduce various ZKPs which are mainly a combination of existing techniques to construct ZKPs for the necessary relations. Section \ref{sec:overview_boostrap_random} introduces a new technique to bootstrap knowledge of randomness despite the commitment weak-opening deficit. Section \ref{sec:compact_overview} introduces a new technique for compacting and proving consistency for multiple assets to be fitted into the coefficients of the polynomial ring to achieve compactness.  
\end{rev} 

\subsection{Encrypted Table-based Ledger (ETL)}
\label{sec:etl_transaction}

{
    \setlength{\floatsep}{5pt}          
    \setlength{\intextsep}{5pt}
    \begin{figure}[h]
    \centering
    \includegraphics[width=0.45\textwidth]{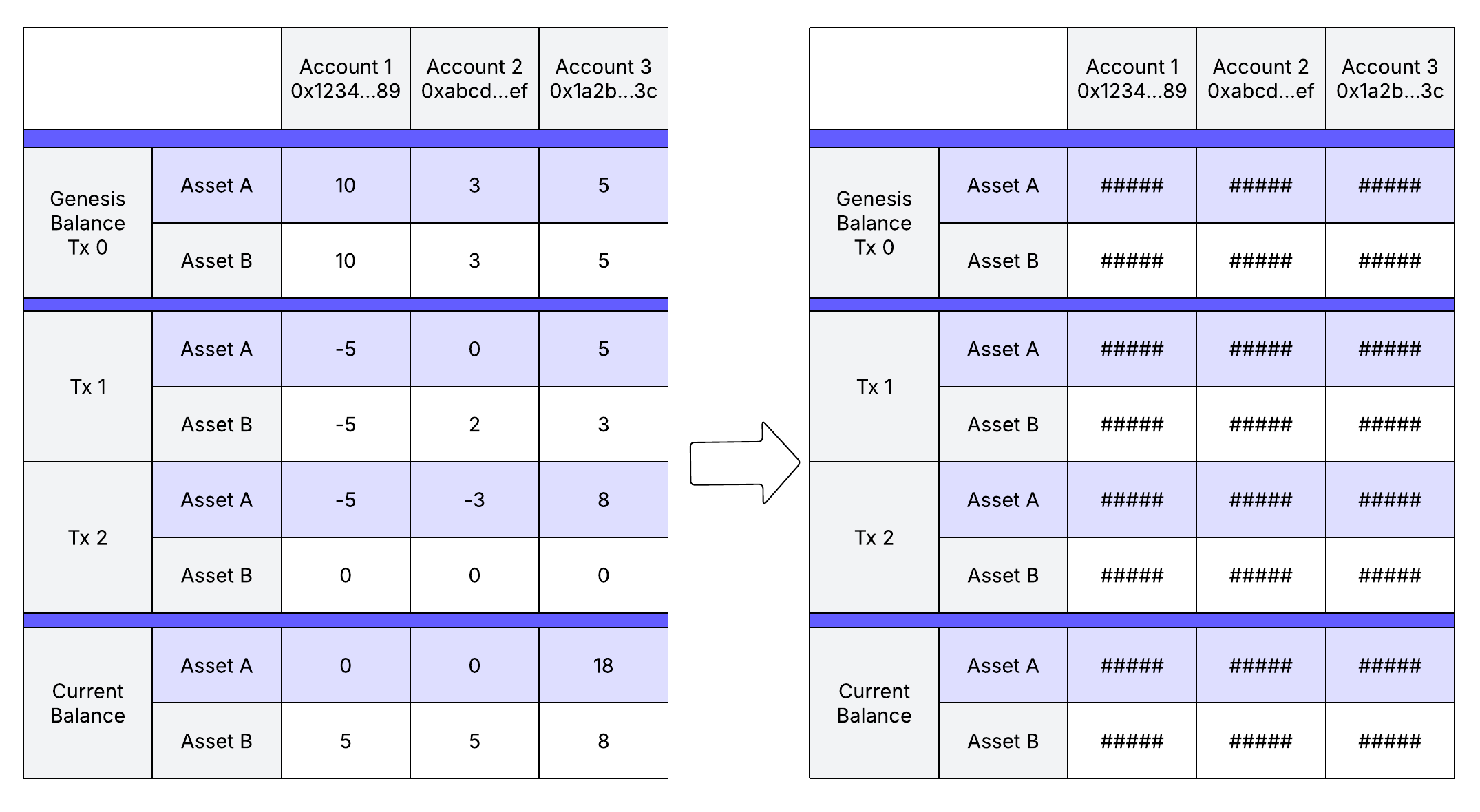} 
    \caption{Logical Representation of ETL}
    \label{fig:etl_table}
    \end{figure}
}

To facilitate and motivate the construction later, we introduce the notion for ETL. 
A (multi-asset) table-based ledger illustrated in Table \ref{fig:etl_table}, is a structured approach to managing transactional data where the data is arranged into a table format. On the left of the figure is the logical representation of the ledger, while the ledger on the right of the figure is the encrypted version of the same ledger. 
The ledger consists of $n$ participants where each column represents a participant and each row in the ledger represents a transaction. Each transaction may consist of sub-rows where each sub-row corresponds to an asset transacted on the ledger.  
Let $t$ be the index identifier for a transaction, $a$ be the index identifier for an asset, and $i$ be the index identifier of a participant's account. 
We define $\pklist$ be a list of participant, $\assetlist$ be a list of asset that can be transacted in the system and $\txlist$ be a list of transaction.

\emph{Commitments}: To hide a value $v$ that belongs to a transaction $t$, we use lattice-based commitment schemes. A transaction $\tx_{t}$ is composed of the commitments $\com_{t,a,i}$ and the proofs $\pi_{t,a,i}$ as $\tx_{t} := \{ \com_{t,a,i}, \pi_{t,a,i}\}_{a \in \assetlist, i \in \pklist}$ where a value  $v_{t,a,i}$ is committed using a random value $\mathbf{r}_{t,a,i}$ using $\com_{t,a,i} \gets \commit(v_{t,a,i}, \mathbf{r}_{t,a,i})$.

We first describe the transaction flow in detail using notation from \cref{sec:model} and  \cref{sec:qpadl_con} before delving into the details of the various proofs. The transaction flow is illustrated in \cref{fig:PADL_transaction_flow}.

{
    \setlength{\floatsep}{5pt}          
    \setlength{\intextsep}{5pt}
    \begin{figure}[ht] 
        \centering
        \includegraphics[width=0.5\textwidth]{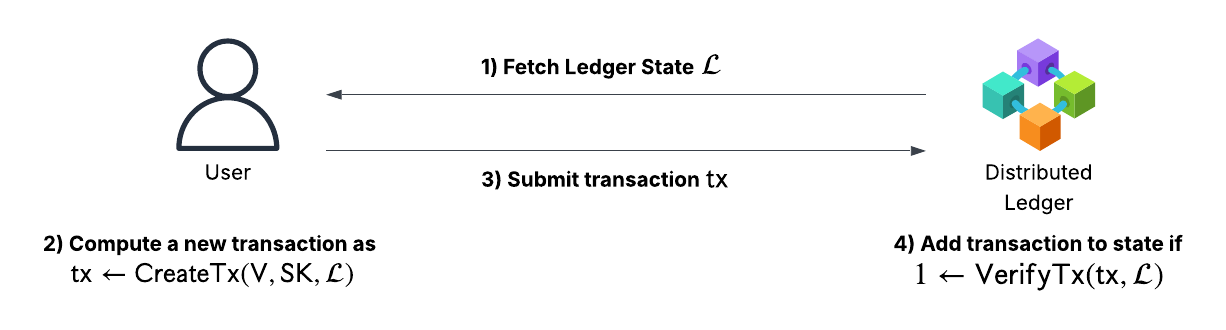}
        \caption{A high-level overview of transaction.}
        \label{fig:PADL_transaction_flow}
    \end{figure}
}

\begin{itemize}
    \item To perform a transaction in $\name$, the sender first commits to the transaction amounts for both input and output accounts. The sender also commits null values to accounts for decoy purpose. 
    Let $\valuelist$ be the transaction values chosen by the sender, $\ledger$ be the state of the ledger, and $\sklist$ be the secret key of the input/spending accounts.
    The sender then computes the needed proofs $\pi$. 
    The transaction is now composed of both the commitments and the proofs which is generated using $\tx \gets \createtx(\valuelist, \sklist, \ledger)$.
    \item The sender then proceeds to submit the transaction $\tx$ publicly to the distributed ledger. 
    \item The distributed ledger verifies the transaction $\tx$ using $\verifytx(\allowbreak\tx,\ledger)$. If the verification returns true, the transaction is appended to the ledger $\tx \rightarrow \ledger$.
\end{itemize}

{
    \setlength{\floatsep}{5pt}          
    \setlength{\intextsep}{5pt}
    \begin{figure}[ht!]
    \centering
    \includegraphics[width=0.35\textwidth]{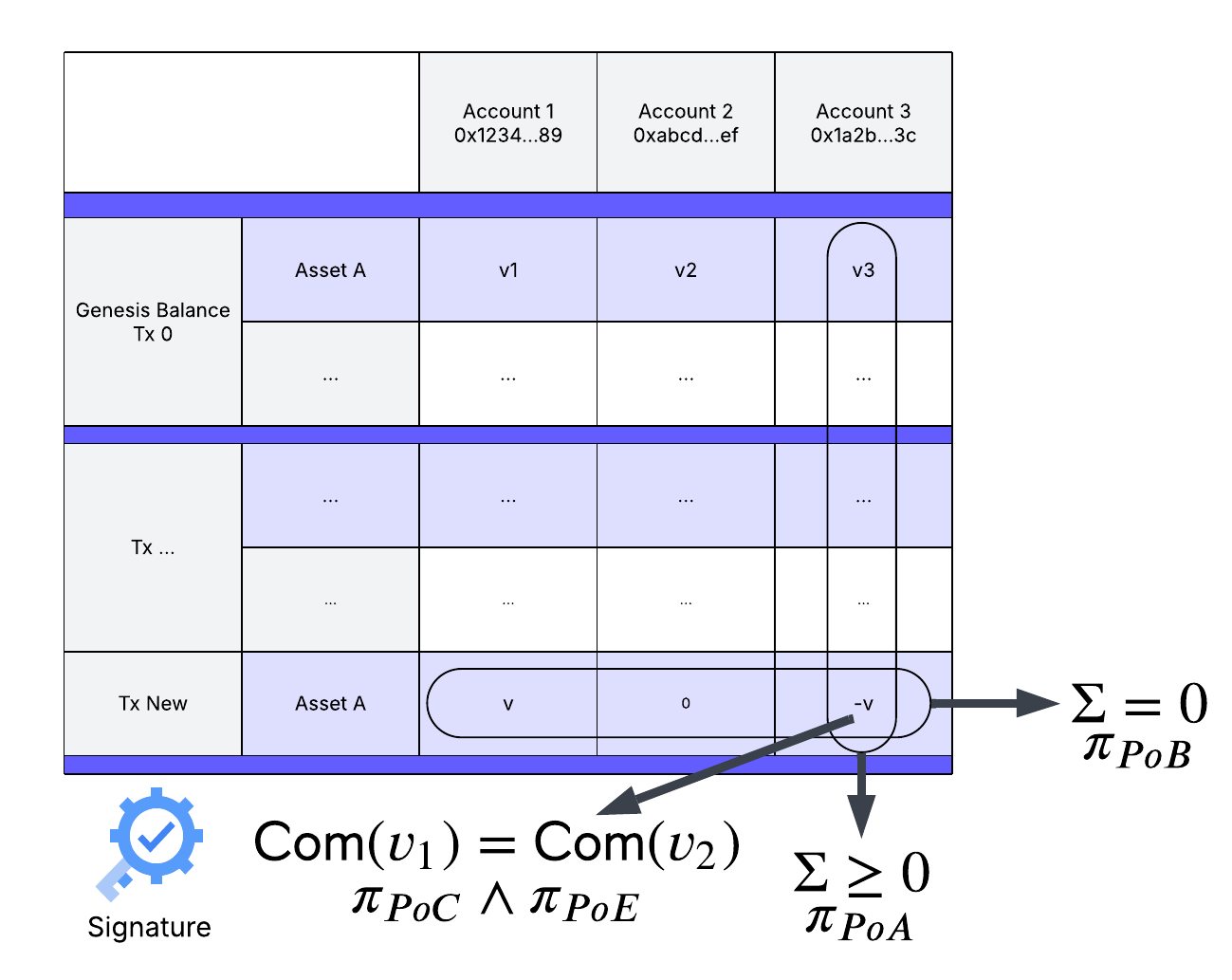} 
    \caption{ETL Proofs Illustration}
    \label{fig:etl_property_table}
    \end{figure}
}

Since the ledger is encrypted, the sender of the transaction is required to prove that the new transaction is well-formed in zero-knowledge to preserve the integrity of the ledger which is checked in $\verifytx(\cdot)$. We informally introduce various relations that preserve the integrity of the ledger while deferring their formal definitions to \cref{sec:qpadl_con}. The relations enforced by the proofs are illustrated in \cref{fig:etl_property_table}. In essence, we have the following:

\begin{itemize}
    \item Proof of Balance ($\pi^{PoB}$): ensures that the values in the new transaction summed up to zero which ensures that no new asset value is created, thus conserving the total asset value.
    \item Proof of Consistency ($\pi^{PoC}$): In PoC, the commitment is proven to satisfy the conditions that the message contains $v,\sqrt{q} \cdot v$ 
    while the random value $\mathbf{r}$ is proven to be short. 
    The $\sqrt q$ relation ensures that the secret key holder can perform direct value extraction from the commitment. The short random value $\mathbf{r}$ is later used to assist in proving commitments' values equivalence statement.
    \item Proof of Equivalence ($\pi^{PoE}$): Since the random value $\mathbf{r}$ (required for opening) used in the commitment might not be known by the address owner since it is generated by the sender, the address owner needs to first equate the original commitment with a new commitment that shares the same value using proof of equivalence such that $\relation(\com_1,\com_2)=1 \Leftrightarrow v_1=v_2$ where $ \com_1 \gets \commit(v_1,\mathbf{r}_1),\com_2 \gets \commit(v_2, \allowbreak \mathbf{r}_2)$. 
    In addition, this proof also checks the knowledge of secret key, thus serving as a signature to authorize a transaction for the spending account.
    \item Proof of Asset ($\pi^{PoA}$): Given a desired spending amount $v_s:=-v$ for some $v\in \mathbb{Z^+}$, PoA asserts that the transaction only goes through if the account has sufficient funds to spends such that $(\sum_{t\in\txlist} v_t) + v_s \geq 0$.
\end{itemize}

Intuitively, combining all the previous proofs allow the ledger to enforce that any new transaction appended to its state is correctly formed and preserves all the important integrity requirements for the ledger state. The stated requirements are the main integrity requirements for a ledger as shown in other ETL-like schemes \cite{chatzigiannis2021miniledger, PADL}. In this scheme and in contrast to previous ETL-like schemes, we add another requirement where the public key is generated with a proof of well-formness $\pi^{KW}$ during the setup. The transaction scheme construction can be found in \cref{sec:qpadl_con} while the techniques used for the ZKP are detailed in the next section.

\subsection{Constructing Zero-Knowledge Proofs for various Transaction-related Relations}
\label{sec:technical_overview_zkp}

In the following, we start with the commitment structure of a transaction in $\name$, while gradually introducing the zero-knowledge proofs to obtain the full scheme. This scheme utilities the polynomial ring $\mathcal{R}_q = \mathbb{Z}_q[X]/(X^d + 1)$. 

\subsubsection*{\textbf{Initial Commitment Structure}}
Each participant $i$ in a transaction $t$ for an asset $a$ has an associated BDLOP commitment $\commitment_{t,a,i}$ \cite{BDLOP16} as given as $\com_{0,t,a,i} = \mathbf{Ar_{t,a,i}}, \com_{3,t,ai}= \mathbf{B}^T\mathbf{r}_{t,a,i}+v_{t,a,i}$, where $\mathbf{A}$ and $\mathbf{B}$ are public random matrix commitment keys, $\mathbf{r}_{t,a,i}$ is a secret random value, $v_{t,a,i}$ is the committed value.

\subsubsection{Zero-Knowledge Proof of Balance}
\label{sec:overview_pob}

The Proof of Balance protocol (ZKPoB) in $\name$ enables a prover (the sender) to demonstrate, in zero-knowledge, that the sum of all committed values for a given transaction and asset is zero.
This ensures that no value is created or destroyed in the transaction, i.e., the transaction is balanced, while preserving the confidentiality of individual values.

\textbf{PoB - Protocol Overview}.
Given the commitment structure above is a standard BDLOP \cite{BDLOP16} commitment scheme, the proof of balance is a simple proof of opening for BDLOP commitment combined with a linear check over the sum of commitment that results in a commitment of zero. The \textbf{aggregate commitment} for asset $a$ in transaction $t$ is $\commitment^{sum}_{t,a} = \sum_i \commitment_{t,a,i} $ with total random value $\mathbf{r}_{t,a} = \sum_i \mathbf{r}_{t,a,i}$ and total value should be $\sum_{i} v_{t,a,i} = 0$.

\subsubsection{Extracting Commitment with Proof of Consistency}
\label{sec:overview_poc}
In ETL, the commitments in a transaction are computed by the sender/spender. To know the new wallet values for the account, the receiver would receive the information about the commitment. In RingCT schemes, such as \cite{matrict_2_2022}, it is assumed that the information (token spending key) is communicated to the receiver by the spender through an encrypted channel. However, ETL further requires that any transaction received can be used or explained by the receiver which is not possible in RingCT without further employing provable encryption. This is because the correct information may never be communicated. Therefore, we further amend the commitment introduced earlier to enable provable extraction for the commitment. The updated commitment is firstly described as follows:
\begin{align}
    \com_{t,a,i} =
    \begin{bmatrix}
    \com_{0,t,a,i} \\
    \com_{1,t,a,i} \\
    \com_{2,t,a,i} \\
    \com_{3,t,a,i}
    \end{bmatrix}
    =
    \begin{bmatrix}
    \mathbf{A} \\
    \pk_{1,i}^T \\
    \pk_{2,i}^T \\
    \mathbf{B}^T
    \end{bmatrix}
    \mathbf{r}_{t,a,i}
    +
    \begin{bmatrix}
    \mathbf{0} \\
    v_{t,a,i} \\
    \sqrt{q} v_{t,a,i} \\
    v_{t,a,i}
    \end{bmatrix}
    \label{eq:commitment_extractable}
\end{align}
where $\pk_{1,i}^T$ and $\pk_{2,i}^T$ are the public keys introduced below.

The encryption or extraction functions similarly to \cite{LNPS21} that is based on Regev-type encryption scheme \cite{regev2024latticeslearningerrorsrandom} but for arbitrary ring element, whereby it commits to both the message, $v_{t,a,i}$, and the message multiplied with a $\sqrt{q}$ factor, $\sqrt{q} v_{t,a,i}$. In particular, the committed message is committed under the commitment key $\pk^T =\mathbf{s}^T\mathbf{A}+\mathbf{e}^T$ for some secret key ($\mathbf{s}$, $\mathbf{e}$). By using the knowledge of secret key $\mathbf{s}$, the account owner can compute the message shifted by the error term $\mathbf{e}^T\mathbf{r}$. 
Finally, message can be recovered by eliminating the error term $\mathbf{e}^T\mathbf{r}$. The error term is computed using $\sqrt{q}$ modulo arithmetic assuming the error terms are small enough and no reduction modulo $q$ has occurred. 
Therefore, given a well-formed commitment, proved using proof of consistency which relates both the messages with the $\sqrt {q}$ factor, the participant (account owner) can extract the value using algorithm from Figure \ref{proto:exctractability}.

\textbf{PoC - \textit{Partial} Protocol Overview}. A proof of $\sqrt{q}$ linear relation in the committed message can be done using a simple linear proof over the standard BDLOP proof of opening plus a linear relation check. However, we are not done yet. Looking ahead, we further require checking that the $\mathbf{r}$ used in the commitment is somewhat short which is not guaranteed by the weak (or relaxed) opening. 

In the (weak) opening of BDLOP commitment, we obtain two transcripts $(\mathbf{w},c,\mathbf{z})$ and $(\mathbf{w},c',\mathbf{z}')$ where $\mathbf{w}$ is the initial commitment, $c$ is the challenge and $\mathbf{z}$ is the masked response, then we set the extracted random value $\mathbf{r}^*$ to be $\mathbf{r}^* =  \mathbf{\bar z}/\bar c$ where $\mathbf{\bar z} = \mathbf{z}-\mathbf{z}'$ and $\bar c=c-c'$. Then, the extracted message is $\mathbf{m}^* = \com_1 - \mathbf{Br^*}$.
Note that the weak opening (see \cref{def:opening_bdlop}) only says that extracted $\mathbf{r}^*$ multiplied by the challenge $\bar c$ (i.e. $\bar c \mathbf{r}^*$) is short because $\mathbf{r}^* = \mathbf{\bar z}/\bar c$ and only $\mathbf{\bar z}$ is bounded. 
The equivalent opening with a different statement where $\mathbf{r}^*$ is set to be $\mathbf{\bar z}$ but the commitment is scaled by the challenge as $\bar c\com_0$ and $\bar c \com_1$ results in a message $\bar cm^*$, is not suitable to our case later where we make use of the fact that difference of two commitment of the same message result in a message of 0 but we are unable to ensure $\bar c= 1$ in general.

\subsubsection{Proof of Asset: Range Proof in the NTT }
\label{sec:overview_rangeproof}

To prove that the spender does not overspend its asset, ETL requires the spender to prove that its sum of transactions (positive and negative) together with the pending transaction is still positive. Using the additive homomorphic property of the commitment, the value $v^{sum}$ balance committed: $\commitment^{sum}_{t,a} = (\sum_i \commitment_{t,a,i}) + \com_{\widehat t}$
is proved to be positive (i.e. $[0,(q-1)/2]$ and $v^{sum} \in \mathbb{Z}$) where $\widehat t$ is the new pending transaction.
Range proof for the range $[0,2^{n})$ for some integer $n$ via bit decomposition \cite{bulletproof} can be achieved by proving the statements 
$$
\sum_{i=0}^{n-1} 2^{i}\mathsf{binary}(v)_i =v\ \wedge\ \mathsf{binary}(v)\odot(\mathsf{binary}(v)-\mathbf{1}) = 0
$$
where $\mathsf{binary}(v)$ is the binary decomposition of $v$ and $\odot$ is the hadamard product. 

\textbf{PoA - Protocol Overview}.
We describe a range proof that performs the check in the NTT domain \cite{LNPS21,LNS20, ENS20} which uses the BDLOP commitment scheme. 
The former bit-packing statement can be checked in the NTT domain using 
\begin{align*}
\ntt(v) =  \mathbf{Q}\ntt(\mathbf{v}_{bin}) = 
\begin{bmatrix}
   1 & 2 & \cdots & 2^{n-1}\\
   \vdots & \vdots & & \vdots \\
   1 & 2 & \cdots & 2^{n-1}
\end{bmatrix}\ntt(\mathbf{v}_{bin})
\end{align*}
where $\mathbf{v}_{bin}=\ntt^{-1}(\mathsf{binary}(v))$ and $\mathbf{Q}$ is a Vandermonde matrix evaluated at 2. The Vandermonde matrix also checks that the value $v$ is an integer since all the NTT slots share the same value $v$. 
We use the unstructured linear proof of \cite{ENS20} that proves linear relation on the NTT representation. The latter binary relation can be checked in the NTT domain using $\mathsf{v}_{bin}(\mathsf{v}_{bin}-1) = 0$ which ensures the message is binary. It can be proved via a standard product proof of the form $m_1(m_1-1)=0$ and we use product proof from \cite{ALS20}.

Note that generating proof of asset requires the knowledge of $\mathbf{r}$ of the commitment. 
However, we conveniently assume $\mathbf{r}$ is known by both sender and receiver. It can be the case that $\mathbf{r}$ is never communicated to the receiver by a malicious sender since the sender computed the commitment for the receiver. We show how we bootstrap this knowledge of $\mathbf{r}$ through a new approach next.

\subsubsection{Bootstrapping the Knowledge of Random Values by Recommitment}
\label{sec:overview_boostrap_random}

One of the ways to verifiably communicate the random value $\mathbf{r}$ is to perform a verifiable encryption of $\mathbf{r}$ and prove a linear relation such that the same encrypted $\mathbf{r}$ satisfies $\com_0 = \mathbf{Ar}$ for BDLOP commitment key matrix $\mathbf{A}$ and commitment $\com_0$. However, this approach is not efficient because it requires the encryption to encrypt messages that consists of $(\kappa+ \lambda+3)$ numbered $\mathcal{R}_q$ elements and provide verifiable encryption proof with the linear relation check.  Therefore, in the following approach, once the sender submits a valid commitment with a proof of opening, the receiver can always re-commit the same commitment under a different commitment, which allows the receiver to "bootstrap" the knowledge of random value $\mathbf{r}$, and subsequently, open the message efficiently. 

We hereby provide a new result on how to equate two commitment's message/value without the knowledge of random value $\mathbf{r}$ needed to open one of the commitment. Assume that both the commitments are of the form specified in \cref{eq:commitment_extractable} and proofs of consistency are provided to prove the well-formness of both the commitments. Then, the difference of the commitments is: 
\begin{align*}
\Delta \com_{0,t,a,i} &= \com_{0,t,a,i} - \com_{0,t,a,i}'= \mathbf{A}(\mathbf{r}_{t,a,i}-\mathbf{r}_{t,a,i}'),\\
\Delta \com_{1,t,a,i} &= \com_{1,t,a,i} - \com_{1,t,a,i}'\\
&= \mathbf{\pk^T_{1,i}}(\mathbf{r}_{t,a,i}-\mathbf{r}_{t,a,i}')+ (v_{t,a,i}-v_{t,a,i}'),\\
\Delta \com_{2,t,a,i} &= \com_{2,t,a,i} - \com_{2,t,a,i}' \\ &= \mathbf{\pk^T_{2,i}}(\mathbf{r}_{t,a,i}-\mathbf{r}_{t,a,i}')+ \sqrt{q}(v_{t,a,i}-v_{t,a,i}').
\end{align*}
Notice that we can perform the message difference opening by 
computing 
$\Delta\com_{1,t,a,i} - s_1^T\Delta \com_{0,t,a,i} = \dotprod{\mathbf{e}_{1,i}}{\mathbf{r}_{t,a,i} - \mathbf{r}_{t,a,i}'} + (v_{t,a,i}-v_{t,a,i}')$ and similarly $\Delta\com_{2,t,a,i} - s_2^T\Delta \com_{0,t,a,i} = \dotprod{\mathbf{e}_{2,i}}{\mathbf{r}_{t,a,i} - \mathbf{r}_{t,a,i}'} + \sqrt{q}(v_{t,a,i}-v_{t,a,i}')$. If the decryption is performed using a masked opening and linear relation check, further assuming $\mathbf{e}$ and $\mathbf{r}$ is short, then any small difference in $v$ and $v'$ would be amplified by $\sqrt{q}$ factor, otherwise any large difference would be detected without the $\sqrt{q}$ factor. Therefore, the value is conditioned to be the same only if the norm of the difference is small for both of the message difference opening. The formal result is stated in \cref{lemma:comp_equil}. Note that the adversary is bounded to a short $\mathbf{s,e}$, 
given a public key $\pk$, otherwise we can extract a MSIS solution in the reduction because of proof of equivalence and proof of key well-formness relations. 

\textbf{PoE - Protocol Overview}. 
We use the approximate range proof technique described in \cite{LNP22} to check for the norm bound requirement of commitment difference mentioned in the previous subsection and then perform a simple linear check that asserts that the secret key used is the correct one.  

The approximate range proof uses ABDLOP commitment scheme as part of the commit-and-prove ZKP scheme. This augments BDLOP commitment scheme with Ajtai \cite{Ajtai96} commitment, so that additional commitment of message with short norm but large dimension is "free".
Let $\mathbf{s_1}$ and $\mathbf{m}$ be two messages where $\mathbf{s_1}$ is a message vector with short norm and $\mathbf{m}$ is a message vector with arbitrary norm. The approximate range proof proves that given two matrices $\mathbf{D}_\mathbf{s_1}$ and $\mathbf{D_m}$, the norm of the computed message with respect to the relation is bounded, $\norm{\mathbf{D}_\mathbf{s_1}\mathbf{s_1}+\mathbf{D}_\mathbf{m}\mathbf{m}+ \mathbf{u} }\leq B$, for some bound $B$ and some public vector $\mathbf{u}$. Let $\mathbf{s} := \mathbf{D}_\mathbf{s_1}\mathbf{s_1}+\mathbf{D}_\mathbf{m}\mathbf{m}+ \mathbf{u}$, then by \cref{lemma:infty_bound}, checking that $\mathbf{z} := \mathbf{y}+R\mathbf{s}$ is short implies that $\mathbf{s}$ is short. Further, let $\mathbf{e_i}$ be a polynomial vector such that all its coefficient vector consists of all zeroes but one 1 in the $i$-th position and $\mathbf{r}_i$ be a polynomial vector such that its coefficient vector is the $i$-th row of $R$. To prove the well-formness of $\mathbf{z}$, we need to check that for all $i$, we have $\langle \mathbf{e}_i, \mathbf{z} \rangle = \langle \mathbf{e}_i, \mathbf{y} \rangle + \langle \mathbf{r}_i, \mathbf{D}_\mathbf{s_1}\mathbf{s_1}+\mathbf{D}_\mathbf{m}\mathbf{m}+ \mathbf{u}  \rangle$.

Using property of $\langle \mathbf{r}_i, \mathbf{D}_{\mathbf{s_1}}\mathbf{s_1}\rangle = \langle \sigma_{-1}(\mathbf{D}_{\mathbf{s_1}})^T\mathbf{r}_i, \mathbf{s_1}\rangle$ which follows from \cref{lemma:constant_coeff} and applying \cref{lemma:constant_coeff} itself, we just need to check the constant coefficient of the following polynomial is zero: 
$$
\begin{bmatrix}
    \sigma_{-1}(\mathbf{r}_i)^T \mathbf{D_{s_1}}  & \sigma_{-1}(\mathbf{r}_i)^T\mathbf{D}_{\mathbf{m}} & \sigma_{-1}(\mathbf{e}_i)^T 
\end{bmatrix}
\begin{bmatrix}
   \mathbf{s_1} \\ \mathbf{m} \\ \mathbf{y}
\end{bmatrix}
+ \langle \mathbf{r}_i, \mathbf{u} \rangle - \langle \mathbf{e}_i, \mathbf{z} \rangle.
$$

\textbf{Updated PoC and Proof of Key Well-formness.} We reuse the framework described in PoE to prove that $\mathbf{r}$ is short in the commitment while checking the same linear relation described in proof of consistency. 
During setup, we also check that $\mathbf{s}$ and $\mathbf{e}$ are short in the public key. These proofs allows us to then invoke \cref{lemma:comp_equil} to claim equivalent of values between two commitments even though proof of opening is given in two separate proofs.

\subsection{Compact Multi-Asset Transaction}
\label{sec:compact_overview}
We further optimize/extend the transaction such that multi-asset transaction (up to size $d$) can be compactly represented by a single ring element committed in a transaction. Without loss of generality\footnote{For  $|\assetlist|>d$, we can increase the degree or simply reuse the 3D structure described previously.}, we now consider a 2D table whereby the assets are condensed into the coefficient of the polynomial. Instead of using an integer coefficient to encode the asset value, we encode asset value $v_1,\cdots,v_d\in \mathbb{Z^+}$ for asset $a_1,\cdots,a_d \in \assetlist$ using a polynomial ring element message $m \in \mathcal{R}_q$ of degree $d$. Notice that proof of balance still requires that all asset values sum to zero. In addition, we only make use of the additive homomorphic property to perform coefficient-wise message value summation which does not affect other asset values in other coefficient slots. We now only require that the proof of asset checks that each of the coefficient slots encodes a positive integer. A naive approach would repeatedly apply the described range proof strategy for each slot (with appropriate masking and an updated Q) to enforce the positive range check. This incurs a linear increase in proof size for $\pi^A$. 

We hereby describe an approach to verify the encoded asset value in each coefficients more efficiently where the proof size only increase due to the need to commit to the binary coefficients and the size for the rest of transcript stay relatively constant. 
Let the asset values be encoded as $v=(v_1,\cdots,v_d) \in R_q$ and let $v_i < 2^\beta $. We can decompose each value into $v_{bin_i} := bin(v_i)$ and $v_{bin} := (v_{bin_1},\cdots,v_{bin_d})$ where $\langle v_{bin_i}, pow(2^\beta-1)\rangle = v_i$ and $pow(x):= \sum_{i=0}^{\lfloor \log_2 x \rfloor} 2^i \cdot X^i \in R_q$. Then, we need to check the following linear equation that enforces $v_i < 2^\beta< q$ if the coefficients are binary:
$$\langle v_{bin_i}, pow(2^\beta-1)\rangle - v_i = 0,$$
and the following quadratic equation that enforces $v_{bin}$ is composed of binary coefficients over the integers:
$$\langle v_{bin},v_{bin} -\mathbf{1}\rangle = 0.$$
Finally, we perform check on $\norm{v_{bin}}$ by using approximate range proof where $\norm{v_{bin}} < \psi \cdot \sqrt{\beta d}$. Then to ensure the inner product check holds over the integers, we pick $q$ such that $\abs{\langle v_{bin},v_{bin} -\mathbf{1}\rangle} \leq \abs{\sum_{i=1}^{\beta d}a_i(a_i-1)} \leq \psi^2\beta d+\psi \beta d < q$. The approximate range proof checks the same relation as before but we use \cref{lemma:2_bound_rwy} to bound the $\ell_2$ norm instead.
We can further compress the binary representation of the value so that fewer ring elements are used for commitments if $\beta < d$ 
by shifting $pow(2^\beta-1)$ appropriately such that they act on the correct coefficients.
The three relations can be combined using random linear combination and be checked using a single quadratic relation for well-formedness and the inner product relation can be checked by verifying that the  constant coefficient vanishes (\cref{lemma:constant_coeff}). 



\subsection{Collapsing the Protocol: Security in QROM}

The collapsing property asserts that, it is computationally infeasible to distinguish whether a superposition over valid response has been measured (collapsed) or not. In other words, measuring the register is ``as destructive'' as in the classical setting, even for quantum adversaries. This property is strictly stronger than classical binding or collision-resistance, and is necessary for quantum rewinding techniques in the QROM. We consider the following variants:
\begin{itemize}
    \item \textbf{Weakly Collapsing:} In some works (e.g., Liu–Zhandry \cite{LZ19}), a \emph{weakly collapsing} property is considered, where the measured experiment's acceptance probability is only required to be a non-negligible fraction of the unmeasured experiment's, up to negligible additive loss.
    \item \textbf{Collapsing for Sigma Protocols:} For sigma protocols, the definition is applied to the response (masked witness) register after the commitment and challenge are fixed, and the set $\mathcal{W}$ is the set of valid responses for $(x, a, c)$.
\end{itemize}

\paragraph{Relation to Quantum Proofs of Knowledge.}
The collapsing property is a key technical tool for enabling quantum rewinding and extraction in the QROM, as shown in~\cite{Unr16b,Unr17,LZ19}. It is used to prove that certain sigma protocols and commitment schemes admit quantum proofs of knowledge, and is a necessary assumption for the security of Fiat–Shamir-type transformations in the quantum setting. In our case, we show that the ZKP sigma protocols used in the scheme are weakly collapsing. Then, by applying known results, we show that our scheme is NIZKPoK in the QROM setting.

\section{Preliminaries}
\label{sec:preliminaries}

Let $q$ be an odd prime and we use $\sqrt q $ to denote $\lfloor \sqrt q \rceil$ for brevity. $\mathbb{Z}_q$ denotes the set of integers modulo $q$. We use [n] to denote the set $\{1,\cdots,n\}$. We write $x \sample \mathcal{S}$ when $x$ is sampled uniformly at random from the finite set $\mathcal{S}$, and similarly $x \sample \mathbf{D}$ when $x$ is sampled according to the distribution $\mathbf{D}$. For $r \in \mathbb{Z}$, we define $r \bmod^{\pm} q$ to be the unique element in the interval $[-(q-1)/2, (q-1)/2]$ that is congruent to $r$ modulo $q$. 

Let $d$ be a power of 2, denote $\mathcal{R}$ and $\mathcal{R}_q$ to be rings $\mathcal{R} = \mathbb{Z}[X]/(X^d + 1)$ and $\mathcal{R}_q = \mathbb{Z}_q[X]/(X^d + 1)$.
Suppose $q \equiv 2l + 1 \  (\ \bmod\  4l) $ for some $l \in \mathbb{N}$ then by [\cite{LS18}, Theorem 2.3], the polynomial $X^d + 1$ factors into $l$ prime ideals and let $\zeta \in \mathbb{Z}_q$ be a primitive $2l$-th root of unity, that is 
$
    X^d + 1 = \prod_{j=0}^{l-1} (X^{\frac{d}{l}} - \zeta^{2j+1}) = \prod_{i \in \mathbb{Z}_{2l}^{\cross}} (X^{\frac{d}{l}} - \zeta^i) \quad (\bmod \ q).
 $

 The ring $\mathcal{R}_q$ has a group of automorphisms $Aut(\mathcal{R}_q)$ that is isomorphic to $\mathbb{Z}_{2d}^{\cross}$, 
$i \mapsto \sigma_i: \mathbb{Z}_{2d}^{\cross} \to Aut(\mathcal{R}_q),$
where $\sigma_i$ is defined by $\sigma_i(X) = X^i$.
As described in \cite{ENS20}, for $i \in Z_{2d}^{\cross}$,
    $\sigma_i(X^{\frac{d}{l}} - \zeta) = (X^{\frac{id}{l}} - \zeta) = (X^{\frac{d}{l}} - \zeta^{-i}), $
since the roots of $(X^{\frac{id}{l}} - \zeta)$ are also roots of $(X^{\frac{d}{l}} - \zeta^{-i})$. We will make use of the following lemmas.

\begin{lemma}[{\cite{ENS20}}]
\label{lemma:NTT_as_original}
Let $p = p_0 + p_1 X + \cdots + p_{d-1} X^{d-1}$. Then, $
    \frac{1}{l} \sum_{i=0}^{l-1} \NTT(p)_i = \sum_{i=0}^{d/l-1} p_i X^i.$
\end{lemma}

\begin{lemma}[{Function Evaluations with Constant Coefficient}\cite{Nguyen22}]
    \label{lemma:constant_coeff}
    Let $\mathbf{x}, \mathbf{y} \in \mathbb{Z}_q^{kd}$ and define the polynomial $f = \sigma_{-1}(\mathbf{x})^T \mathbf{y} \in \mathcal{R}_q$. Then, the constant coefficient of $f$ is equal to $\langle \mathbf{x}, \mathbf{y} \rangle$.
\end{lemma}

We further recall norms definition, various properties of cyclotomic rings used and NTT properties in the Appendix \ref{app:additional_background}.

\subsubsection{Challenge Space}
Let $\mathcal{C} := \{-1, 0, 1\}^d \subset \mathcal{R}_q$, with $|\mathcal{C}| = 3^d$, be the challenge set of ternary polynomials with coefficients in $\{-1, 0, 1\}$. We define the probability distribution $C_{p}$ as a binomial distribution centered at $0$ with central value probability $p = 1/2$. The coefficients of a challenge $c \leftarrow \mathcal{C}$ are independently and identically distributed with $\Pr(0) = p \ \text{and} \ \Pr(1) = \Pr(-1) = \frac{1-p}{2}$.
We write $\omega$ such that $\Pr_{c \leftarrow \mathcal{C}}(\|c\|_1 \leq \omega)$ holds with overwhelming probability.
We require the challenge space for our ZKP to consist of short elements such that the difference between two elements is invertible in $\mathcal{R}_q$ with overwhelming probability. We use the strategy from \cite{ALS20} to show that our challenge space is invertible except with negligible probability. The invertibility is discussed in Appendix \ref{app:additional_background}.

\subsubsection{Probability Distributions}
\label{sec:prelim_distribution}
We sample the random polynomials' coefficients in the commitment scheme using the centered binomial distribution $\chi$ on $\{-1, 0, 1\}$, where $\pm 1$ each have probability $5/16$ and $0$ has probability $6/16$, as in~\cite{BLS19, ALS20, ENS20}. We denote by $S_\mu$ the uniform distribution over the set $\{x \in \mathcal{R}_q \mid \|x\|_\infty \leq \mu\}$. We recall the discrete gaussian distribution in Appendix \ref{app:additional_background}.

\subsubsection{Approximate Range Proof}
\label{sec:approx_range_proof}

In this section, we provide techniques for proving the norm bound of a polynomial vector. Let $\texttt{Bin}_k$ be the distribution $\sum_{i=1}^k(a_i - b_i)$, where $a_i, b_i \leftarrow \{0,1\}$ for a positive integer parameter $k$.
In~\cite{BL17}, the authors showed that for a vector $\mathbf{w}$ over $\CalRq$ and a matrix $\mathbf{R}\sample \texttt{Bin}_1$, if $\mathbf{v} = \mathbf{R}\mathbf{w}$, the projection with respect to $\mathbf{R}$ has small coefficients (i.e., small infinity norm), then $\mathbf{w}$ also has small coefficients with high probability. Furthermore, the projection can be generalized as $\mathbf{R}\mathbf{w} + \mathbf{y}$, where $\mathbf{y}$ is an arbitrary vector over $\CalRq$. We refer to~\cite{LNS21} and give the formal description in Lemma~\ref{lemma:infty_bound} and recall all other required lemma in Appendix \ref{app:additional_background}.

\subsection{Cryptographic Definitions}

\subsubsection{Module-SIS, Module-LWE Problems and Rejection Sampling}
The security of our constructions is based on the well-established computational hardness of the Module-LWE ($\mlwe$) and Module-SIS ($\msis$) lattice problems. Here, both problems are defined over the ring $\CalRq$. We recall the $\msis$ and $\mlwe$ in the Appendix \ref{app:additional_background}. In the ZKP, the prover usually compute a masked response $z=y+cr$ that depends on some secret $r$. Since $y$ does not span the whole space, rejection sampling can be used to remove the dependency of $z$ on $r$ which enables the simulatability of ZKP. We recall the rejection sampling algorithm $\rej$ from \cite{Lyu12} in the Appendix \ref{app:additional_background}.

\subsubsection{Commit-and-Prove System, Commitment and NIZK}
We recall the definition for standard cryptographic primitives commitment scheme $ :=(\commitgen, \commit, \open)$, commit-and-prove system and non-interactive zero-knowledge proof, NIZK $:=(\nizksetup, \nizkprove, \allowbreak \nizkverify)$ in the Appendix \ref{app:additional_background}.

\subsubsection{(A)BDLOP Commitment and Zero-Knowledge Proof of Opening}
We briefly recall BDLOP commitment \cite{BDLOP16} and ABDLOP commitment \cite{LNP22} in terms of weak opening and provide a slightly more complete background in Appendix \ref{app:commitment_bdlop_abdlop}. 

\begin{definition}[Weak Opening of BDLOP Commitment~\cite{ALS20}]
\label{def:opening_bdlop}
    Let $\bfA \in \CalRq^{\kappa \times (\kappa+\lambda +n)}$ and vectors $\mathbf{b}_i \in \CalRq^{(\kappa+\lambda +n)}$. 
    A weak opening for the commitment $\com$ consists of a polynomial $\bar{c}\in\mathcal{R}_q$, randomness $\mathbf{r}^\ast$ over $\mathcal{R}_q$, and message $m^\ast\in\mathcal{R}_q$ such that 
    \begin{itemize}
        \item $\|\bar{c}\|_1\leq2\omega$ and $\bar{c}$ is invertible over $\mathcal{R}_q$. 
        \item $\|\bar{c}\mathbf{r}^\ast\| \leq 2 \beta = 2\mathfrak{s}\sqrt{2(\kappa+\lambda+n)d}$. 
        \item $\bfA\mathbf{r}^\ast=\com_0$.
        \item $\mathbf{b}_i^T\mathbf{r}^\ast+m_i^\ast =\com_i$ for $i\in [n]$.
    \end{itemize}
\end{definition}

\begin{definition}[Weak Opening of ABDLOP Commitment~\cite{LNP22}]
\label{def:opening_abdlop}
      Let $\bfA_1 \in \CalRq^{\kappa \times (\mu_1)}, \bfA_2 \in \CalRq^{\kappa \times (\mu_2)}, \mathbf{b}_i \in \CalRq^{(\mu_2)}$. 
      A weak opening for the commitment $\com$ consists of a polynomial $\bar{c}\in\mathcal{R}_q$, two random elements $\mathbf{r}^\ast$ and $\mathbf{s}^\ast$ over $\mathcal{R}_q$, and message $m^\ast\in\mathcal{R}_q$ such that 
     \begin{itemize}
        \item $\|\bar{c}\|_1\leq2\omega$ and $\bar{c}$ is invertible over $\mathcal{R}_q$.
        \item $\|\bar{c}\mathbf{r}^\ast\|\leq2\mathfrak{s}_1\sqrt{2(\mu_1)d}$ and $\|\bar{c}\mathbf{s}^\ast\|\leq2\mathfrak{s}_2\sqrt{2(\mu_2)d}$.
        \item $\bfA_1\mathbf{r}^\ast+\bfA_2\mathbf{s}^\ast=\com_0$.
        \item $\mathbf{b}_i^T\mathbf{s}^\ast+m_i^\ast =\com_i$ for $i\in [n]$.
    \end{itemize}
\end{definition}

\section{A Suite of ZKP Constructions for ETL}
\label{sec:zkp_construction}
\subsection{Ledger Setup and Transaction}

The setup algorithm is formally given in \cref{fig:algorithm_qpadl_setup} and we give a high-level overview for the setup as below.
A participant $i$ starts by generating two secret keys $(\mathbf{sk}_{1,i}, \mathbf{sk}_{2, i}) := ((\mathbf{s}_{1, i}, \mathbf{e}_{1, i}), (\mathbf{s}_{2, i}, \mathbf{e}_{2, i}))$ where for all $j \in \{1,2\}, \mathbf{s_j} \gets \chi^{\kappa d}$ and $\mathbf{e_j} \gets \chi^{(\kappa+\lambda+3)d}$. Then, the public key for participant $i$ is $(\pk_{1,i},\pk_{2, i}) := (\bfA^T\mathbf{s}_{1, i}+\mathbf{e}_{1, i}, \bfA^T\mathbf{s}_{2, i}+\mathbf{e}_{2, i})$.
The parameters used are summarized in Table \ref{tab:param_detail_overview}. 

{
\setlength{\floatsep}{5pt}     
\setlength{\intextsep}{5pt}
\begin{table}[h!]
\centering
\begin{adjustbox}{minipage=\linewidth,scale=0.85, center}
\begin{tabular}{|c|l|}
\hline
\textbf{Symbol} & \textbf{Description} \\
\hline
$d,l$ & Parameters related to degree and splitting factors \\
\hline
$\kappa$, $\lambda$ & Parameter related to the MLWE and MSIS ranks \\
\hline
$\mathbf{A}, \mathbf{B}$ & BDLOP commitment key \\
\hline
$\mathbf{A}_1,\mathbf{A}_2, \mathbf{B}_1, \mathbf{B}_c', \mathbf{B}_c''$ & Commitment Key for Proof of Consistency \\
\hline
$\mathbf{A}_3,\mathbf{A}_4, \mathbf{B}'_{eq},\mathbf{B}''_{eq}$ & Commitment key for Proof of Equivalence \\
\hline
 $\mathbf{a}^T_{\text{bin}}, \mathbf{a}'^T_{\text{bin}}, \mathbf{a}^T_{g}$ & Commitment key for Proof of Asset \\
\hline
\end{tabular}
\end{adjustbox}
\caption{Parameter Overview Table for the Ledger}
\label{tab:param_detail_overview}
\end{table}
}

As discussed in \cref{sec:etl_transaction}, a transaction is created by a sender and sent to the ledger to be verified before it is appended to the ledger.
To create a transaction, the sender commits the values  $v_1,v_2,\cdots,v_{n}$ for the transaction into a commitment list $\{\commitment_{t,a,i}\}$ where $t$ is a transaction index, $a$ is an asset index, and $i$ is a participant index.
The main transaction commitment is given in Equation \ref{eq:commitment_extractable} and can be summarized as $\com_{t,a,i} := (\mathbf{A}\mathbf{r}_{t,a,i}, \pk^T_{1,i}\mathbf{r}_{t,a,i}+ v_{t,a,i}, \pk^T_{2,i}\mathbf{r}_{t,a,i}+ \sqrt qv_{t,a,i}, \mathbf{B}^T\mathbf{r}_{t,a,i}+ v_{t,a,i})$ 
where $\mathbf{r}_{t,a,i}$ is a secret low-norm vector drawn from distribution $\chi^{(k+\lambda +3)d}$ chosen by the sender for every commitment and $v_{t,a,i}$ is the value.  $\bfA,\bfB$ are sampled uniformly from $\CalRq^{\kappa \times (\kappa+\lambda +3)}$ and $\CalRq^{(\kappa+\lambda +3)}$, respectively. The details of transaction is given in the construction in \cref{sec:qpadl_con}.

\subsection{Proof Constructions}

This section details the construction of the four main proofs (proof of balance, proof of consistency, proof of equivalence, and proof of asset) which is discussed in \cref{sec:technical_overview_zkp}.

\subsubsection{Zero-Knowledge Proof of Balance (ZKPoB)}
We show that $\name$ achieves proof of balance by utilizing the inherent linear relationship due to the additive homomorphic property of BDLOP commitment scheme among the participants' commitments. For each asset $a$ involved in a transaction $t$, participant $i$ either spends or receives an amount $v_{t,a,i}$ (including $v_{t,a,i}=0$). 
The protocol is given in \cref{proto:pi_b} in the Appendix and is a standard instantiation of BDLOP linear relation with the relaxed proof of opening \cite{BDLOP16}.

\begin{theorem}
    \label{thm:pi_b}
    Protocol $\pi_{}^B$ given in \cref{proto:pi_b} is complete, (quantum) knowledge sound, and honest-verifier zero-knowledge.
\end{theorem}

The proof proceeds similarly to that of standard BDLOP linear proof and the proof is given in the \cref{app:proof_of_balance} for completeness. The quantum proof of knowledge is shown in Appendix \ref{sec:qpok-pob-collapsing}.

\subsubsection{Zero-Knowledge Proof of Consistency (ZKPoC)}
Proof of Consistency (PoC) is a "proof of well-formedness" of \cref{eq:commitment_extractable} with an approximate range proof of the binding factor $r_{t,a,i}$. We present our construction of ZKPoC in Protocol \ref{proto:pi_c}. At a high level, we follow the approach from~\cite{LNP22} to first commit to $\mathbf{r}_{t,a,i}$ and the value $v_{t,a,i}$ of \cref{eq:commitment_extractable} using ABDLOP commitment and then prove bounds on $\norm{ \mathbf{r}_{t,a,i}}_\infty$. Then, we perform the verification using linear relation check that ensure the committed messages satisfies \cref{eq:commitment_extractable}. We denote ABDLOP commitment in proof of consistency protocol by: 
\begin{align} 
    \label{eq:poc_abdlop_com}
    f_i = \begin{bmatrix}
        f_{0,i} \\ 
        f_{1,i}
    \end{bmatrix}
    = \begin{bmatrix}
        \bfA_1\\0
    \end{bmatrix} \mathbf{r}_{t,a,i} + 
    \begin{bmatrix}
        \bfA_2\\ \bfB_1
    \end{bmatrix}\mathbf{s}_{t,a,i}+
    \begin{bmatrix}
        \mathbf{0^\kappa}\\v_{t,a,i}
    \end{bmatrix}
\end{align}
with the same variables $\mathbf{r}_{t,a,i}$ and $v_{t,a,i}$ as in Eq.\eqref{eq:commitment_extractable}, and 

\begin{itemize}
    \item $\bfA_1$, $\bfA_2$, $\bfB_1$ are public parameters, and are sampled uniformly from $\CalRq^{\kappa \times (\kappa+\lambda +3)}$, $\CalRq^{\kappa \times (\kappa+\lambda +3)}$, and $\CalRq^{1 \times (\kappa+\lambda +3)}$, respectively.
    \item $\mathbf{s}_{t,a,i}$ is a secret low-norm vector drawn from distribution $\chi^{(\kappa+\lambda +3)d}$ chosen by the sender for every commitment. 
    \item $\bfB_c'\in \CalRq^{256/d\times(\kappa+\lambda+3)}, \bfB_c''\in\CalRq^{1\times (\kappa+\lambda+3)}$ are sampled uniformly from $\CalRq^{256/d\times(\kappa+\lambda+3)}$ and $\CalRq^{1\times (\kappa+\lambda+3)}$, respectively.
    \item $\mathfrak{s}_1$, $\mathfrak{s}_2$, and $\mathfrak{s}_3$ denote the standard deviations, where $\mathfrak{s}_1=11\omega\sqrt{(\kappa+\lambda+3)d}$, $\mathfrak{s}_2=11\omega\sqrt{(\kappa+\lambda +3)d}$, and $\mathfrak{s}_3=11\sqrt{337}\cdot \sqrt{(\kappa+\lambda+3)d}$.
    \item $\mathbf{R}_i\leftarrow \texttt{Bin}_1^{256\times(\kappa+\lambda+3)d}$ is a random matrix sampled by the verifier as a challenge, and let $\mathbf{r}_j\in\CalRq^{\kappa+\lambda+3}$ denote the polynomial vector so that its coefficient vector is the $j$-th row of $\mathbf{R}_i$, where $j\in[256]$. 
    \item  $\mathbf{e}_j\in\CalRq^{256/d} $ for $j \in [256]$ and $\mathbf{e}_j'\in\CalRq$ for $j \in [128]$ are the polynomial vectors such that their coefficient vectors consist of all zeroes and one 1 in the $j$-th position.
\end{itemize}

{
    \setlength{\floatsep}{5pt}     
    \setlength{\intextsep}{5pt}
    \begin{figure}[!ht]
    \noindent
    \begin{adjustbox}{minipage=1.176\linewidth,scale=0.85,frame}
    \pcnobox[codesize=\scriptsize,colspace=-3mm]{
    \\
    \underline{\textbf{Sender (Prover)}} \<\< \underline{\textbf{Participant i (Verifier)}}
    \\\text{Input:}
    \\\bfA,\bfB,\texttt{pk}_i\<\<\bfA,\bfB,\texttt{pk}_i
    \\\bfA_1,\bfA_2,\bfB_1,\bfB_c',\bfB_c''\<\<\bfA_1,\bfA_2,\bfB_1,\bfB_c',\bfB_c''
    \\\mathbf{r}_{t,a,i},\com_{t,a,i}\<\<\com_{t,a,i}
    \\[0.1\baselineskip][\hline] \\[-0.5\baselineskip]
    \mathbf{s}_{t,a,i}\xleftarrow[]{\$}\chi^{(\mathrm{k}+\lambda+3)d}, g\xleftarrow[]{\$}\CalRq: \text{ constant of } g=0\\
    f_i=
    \begin{pcmbox}\begin{bmatrix}
        f_{0,i}\\f_{1,i}
    \end{bmatrix}\end{pcmbox}=
    \begin{pcmbox}\begin{bmatrix}
        \bfA_1\\0
    \end{bmatrix}\end{pcmbox}\mathbf{r}_{t,a,i}+
    \begin{pcmbox}\begin{bmatrix}
        \bfA_2\\ \bfB_1
    \end{bmatrix}\end{pcmbox}\mathbf{s}_{t,a,i}+
    \begin{pcmbox}\begin{bmatrix}
        \mathbf{0^\kappa}\\v_{t,a,i}
    \end{bmatrix}\end{pcmbox}\\
    \mathbf{y}_{1,i}, \mathbf{y}_{2,i},  \mathbf{y}_{3,i} \xleftarrow[]{\$} \mathcal{D}_{\mathfrak{s}_1}^{(\kappa+\lambda +3)d} \times \mathcal{D}_{\mathfrak{s}_2}^{(\kappa+\lambda +3)d} \times \mathcal{D}_{\mathfrak{s}_3}^{256}\\
    \mathbf{u}_{1,i}=\bfB_c'\mathbf{s}_{t,a,i}+\mathbf{y}_{3,i}, \mathbf{u}_{2,i}=\bfB_c''\mathbf{s}_{t,a,i}+g\\
    \< \sendmessageright*[2.5cm]{f_i, \mathbf{u}_{1,i}, \mathbf{u}_{2,i}} \\
    \<\< \bfR_i\xleftarrow[]{\$} \texttt{Bin}_1^{256\times(\kappa+\lambda +3)d}\\
    \< \sendmessageleft*[2.5cm]{\bfR_i} \\
    \text{compute }\mathbf{z}_{3,i}\text{, s.t. }\vec{z}_{3,i}=\vec{y}_{3,i} + \bfR_i \vec{r}_{t,a,i} \\
    \rej(\vec{z}_{3,i},\bfR_i\vec{r}_{t,a,i}, \mathfrak{s}_3)\\
    \< \sendmessageright*[2.5cm]{\mathbf{z}_{3,i}}\\
    \<\<d_{1,j}\leftarrow Z_q, \forall j\in[256]\\
    \<\<d_{2,j}\leftarrow Z_q, \forall j\in[127]\\
    \< \sendmessageleft*[2.5cm]{\{d_{1,j}\}_{j \in [256]}, \\ \{d_{2,j}\}_{j \in [127]}}\\
    x=\sum_{j=1}^{256} d_{1,j}\Big(\begin{pcmbox}\begin{bmatrix}
        \sigma_{-1}(\mathbf{r}_j)^T\>\sigma_{-1}(\mathbf{e}_j)^T
    \end{bmatrix}\end{pcmbox}
    \begin{pcmbox}\begin{bmatrix}
        \mathbf{r}_{t,a,i}\\\mathbf{y}_{3,i}
    \end{bmatrix}\end{pcmbox}- \\ 
    \quad \sigma_{-1}(\mathbf{e}_j)^T\mathbf{z}_{3,i} \Big)
    +\sum_{j=2}^{128}d_{2,j-1}\sigma_{-1}(\mathbf{e}_j')^T v_{t,a,i}\\
    h=g+x, \mathbf{w}_{i} = \bfA_1\mathbf{y}_{1,i}+\bfA_2\mathbf{y}_{2,i}\\ 
    \bfT=
    \begin{pcmbox}\begin{bmatrix}
        \bfA \>0\>\mathbf{0}\>0\\
        \pk_{1,i}^T\>1\>\mathbf{0}\>0\\
        \pk_{2,i}^T\>\sqrt{q}\>\mathbf{0}\>0\\ 
        \bfB^T\>1\>\mathbf{0}\>0\\
        T_{sum}^r\> T_{sum}^{e'} \> T_{sum}^e \>1
    \end{bmatrix}\end{pcmbox}\\
    \mathbf{v}_{i} = \bfT\begin{pcmbox}\begin{bmatrix}
        \mathbf{y}_{1,i}\>-\bfB_1\mathbf{y}_{2,i}\>-\bfB_c'\mathbf{y}_{2,i}\>-\bfB_c''\mathbf{y}_{2,i}
    \end{bmatrix}^T\end{pcmbox}\\
    \< \sendmessageright*[2.5cm]{h,\mathbf{w}_{i}, \mathbf{v}_{i}} \\
    \<\< c_{i} \xleftarrow[]{\$} \mathbf{\mathcal{C}} \\
    \< \sendmessageleft*[2.5cm]{c_{i}} \\
    \mathbf{z}_{1,i}=\mathbf{y}_{1,i} + c_{i} \mathbf{r}_{t,a,i}, \mathbf{z}_{2,i}=\mathbf{y}_{2,i} + c_{i} \mathbf{s}_{t,a,i} \\
    \rej(\mathbf{z}_{1,i},c_{i} \mathbf{r}_{t,a,i} ,\mathfrak{s}_1), \rej(\mathbf{z}_{2,i},c_{i} \mathbf{s}_{t,a,i} ,\mathfrak{s}_2)\\
    \< \sendmessageright*[2.5cm]{\mathbf{z}_{1,i},\mathbf{z}_{2,i}} \<  \text{Runs } \mathcal{V}^C
    }
    \\[0.1\baselineskip]
    \rule{\linewidth}{0.4pt}
    \\[-0.2\baselineskip]
    $ \mathcal{V}^C:$ Accept if:
    \begin{enumerate}
    \renewcommand{\labelenumi}{(\alph{enumi})}
    \item \label{vrfy:pi_c_1} $\norm{\mathbf{z}_{1,i}} \stackrel{?}{\leq} \mathfrak{s}_1\sqrt{2(\kappa+\lambda+3)d}$
    \item \label{vrfy:pi_c_2} $\norm{\mathbf{z}_{2,i}} \stackrel{?}{\leq} \mathfrak{s}_2\sqrt{2(\kappa+\lambda+3)d}$
    \item \label{vrfy:pi_c_3} $\norm{\mathbf{z}_{3,i}}_{\infty} \stackrel{?}{\leq} \sqrt{2k}\mathfrak{s}_3$ (This implies, $\norm{\mathbf{r}_{t,a,i}}_\infty\stackrel{?}{\leq} 2 \sqrt{2k}\mathfrak{s}_3$ by \cref{lemma:infty_bound} where $k=128$) %
    \item \label{vrfy:pi_c_4} Constant coefficient of $h=0$
    \item \label{vrfy:pi_c_5} $\bfA_1\mathbf{z}_{1,i}+\bfA_2\mathbf{z}_{2,i}\stackrel{?}{=}\mathbf{w}_{i} + c_{i}f_{0,i}$
    \item \label{vrfy:pi_c_6} $
    \begin{bmatrix}
        \mathbf{v}_{0,i}\\
        \mathbf{v}_{1,i}\\
        \mathbf{v}_{2,i}\\
        \mathbf{v}_{3,i}\\
        \mathbf{v}_{4,i}
    \end{bmatrix}
    \stackrel{?}{=}\bfT\begin{bmatrix}
        \mathbf{z}_{1,i}\\
        c_if_{1,i}-\bfB_1\mathbf{z}_{2,i}\\
        c_i\mathbf{u}_{1,i}-\bfB_c'\mathbf{z}_{2,i}\\
        c_i\mathbf{u}_{2,i}-\bfB_c''\mathbf{z}_{2,i}
    \end{bmatrix}-c_i
    \begin{bmatrix}
        \com_{0, t,a,i} \\ 
        \com_{1,t,a,i} \\ 
        \com_{2,t,a,i} \\ 
        \com_{3,t,a,i} \\
        h+\sum_{j=1}^{256}d_{1,j}\sigma_{-1}(\mathbf{e}_j)^T\mathbf{z}_{3,i}
    \end{bmatrix}$
    \end{enumerate}
    
    \end{adjustbox}
    \caption{$\name$ Proof of Consistency Protocol $\pi_i^C$ where $T_{sum}^r=\sum_{j=1}^{256} d_{1,j}\sigma_{-1}(\mathbf{r}_j)^T$, $T_{sum}^{e'}=\sum_{j=2}^{128} d_{2,j-1} \sigma_{-1}(\mathbf{e}_j')^T$ and $T_{sum}^e=\sum_{j=1}^{256} d_{1,j}\sigma_{-1}(\mathbf{e}_j)^T$}
    \label{proto:pi_c}
\end{figure}
}

\begin{theorem}
    \label{thm:pi_c}
    Protocol $\pi_{i}^C$ given in \cref{proto:pi_c} is complete, (quantum) knowledge sound, and honest-verifier zero-knowledge.
\end{theorem}

The proof proceeds similarly to that of a ABDLOP linear proof, proof of opening, and approximate range proof given in \cite{LNP22} and the proof is given in the \cref{app:proof_of_consistency} for completeness.  The quantum proof of knowledge is shown in Appendix \ref{app:quantum_poc_poe_poa_etc}.

\subsubsection{Extractability}
Extractability is a necessary property to achieve for the chosen commitment as per design requirements of ETL. In particular, it is used by each participant to extract the value from the received commitment from the sender. For the construction of $\name$, we leverage the decryption techniques used in \cite{LNPS21,LNP22}. In Protocol \ref{proto:exctractability}, we recap the decryption process for achieving exact extractability when the sender is honest (We refer the reader to \cref{remark:extraction} for clarification). Note that our parameter setting satisfies both $ \|(\sqrt{q}\mathbf{e}_{1,i}^T-\mathbf{e}_{2,i}^T)\mathbf{r}_{t,a,i}\|_\infty\leq\frac{q}{2}$ and $\|\mathbf{e}_{2,i}^T\mathbf{r}_{t,a,i}\|_\infty \leq \sqrt{q}$, which allow us to perform the computation.

{
\setlength{\floatsep}{5pt}     
\setlength{\intextsep}{5pt}
\begin{figure}[ht!]
    \fbox{
    \begin{adjustbox}{minipage=\linewidth,scale=0.85}
    Conditioned on Protocol~\ref{proto:pi_c}'s verified, participant $i$ extracts $v_{t,a,i}$ from $\text{com}_{t,a,i}$ in the following: \\
    \underline{$\extract(\com, \sk):$}
    \begin{enumerate}
    \item Define $x_1$ and $x_2$ such that 
    \begin{align*}
        x_1 &= \com_{1,t,a,i} - \mathbf{s}_{1,i}^T\com_{0, t,a,i} = \mathbf{e}_{1,i}^T\mathbf{r}_{t,a,i}+v_{t,a,i}\\
        x_2 &= \com_{2,t,a,i} - \mathbf{s}_{2,i}^T\com_{0, t,a,i} = \mathbf{e}_{2,i}^T\mathbf{r}_{t,a,i}+\sqrt{q}v_{t,a,i}
    \end{align*}
    \item Let $k=\sqrt{q}x_1-x_2 \bmod{\sqrt{q}}$, we have:
    \begin{align*}
        k &  = \sqrt{q}x_1-x_2 \bmod{\sqrt{q}} \\
        & = \sqrt{q}(\mathbf{e}_{1,i}^T\mathbf{r}_{t,a,i}+v_{t,a,i}-v_{t,a,i})-\mathbf{e}_{2,i}^T\mathbf{r}_{t,a,i} \bmod{\sqrt{q}}\\
        &= -\mathbf{e}_{2,i}^T\mathbf{r}_{t,a,i}
    \end{align*}
    , where  $\|(\sqrt{q}\mathbf{e}_{1,i}^T-\mathbf{e}_{2,i}^T)\mathbf{r}_{t,a,i}\|_\infty\leq\frac{q}{2}$ and $\|\mathbf{e}_{2,i}^T\mathbf{r}_{t,a,i}\|_\infty \leq \sqrt{q}$.
    \item Returns $v_{t,a,i}=\frac{x_2+k}{\sqrt{q}}$.
    \end{enumerate}
    \end{adjustbox}
    }
    \caption{Extraction of $v_{t,a,i}$ from $\com_{t,a,i}$}
    \label{proto:exctractability}
\end{figure}
}

\begin{remark}\label{remark:extraction}
  We note to the reader that the commitment in \cref{eq:commitment_extractable} is decrypted conditioned on the success of Protocol~\ref{proto:pi_c}. Our Protocol~\ref{proto:pi_c} ensures that the exact relation as in  \cref{eq:commitment_extractable} holds for some short $r^*_{t,a,i}$ (with appropriate norm bounds from Protocol~\ref{proto:pi_c}) and $v^*_{t,a,i}$. This ensures the decryption from \cref{proto:exctractability} outputs $v^*_{t,a, i}$ exactly.
\end{remark}

\subsection{Re-commitment and Proof of Asset}

To finalize the transaction, the participant must prove that they have sufficient funds to cover the amounts involved in the transaction. To accomplish this, the participant will create a Zero-Knowledge Proof of Asset (ZKPoA). Technically, this involves computing a zero-knowledge range proof against their current balance of the asset involved in the transaction. However, the general proof of opening in the range proof requires knowledge of the secret randomness vector, which is chosen by the sender in the previous transactions and is unknown to the participants who now wishes to spend the received commitment token. Consequently, the participant must generate a re-commitment $\com_{t,a,i}^\prime$ of the same amount (which is known due to the provable extractability property), using the same public parameters/keys along with a ZKPoC on $\com_{t,a,i}^\prime$, as described in Protocol~\ref{proto:pi_c}. A Zero-Knowledge Proof of Equivalence (ZKPoE) is firstly presented to demonstrate that the committed amount in the re-commitment $\com_{t,a,i}^\prime$ is equal to that of the original commitment $\com_{t,a,i}$, followed by ZKPoA.

\subsubsection{Zero-Knowledge Proof of Equivalence (ZKPoE)}
Recall that the need for ZKPoE is to demonstrate that the amount committed in the participant's re-commitment is equal to the amount committed in the sender's original commitment, i.e., it allows us to verify in zero-knowledge if $v_{t,a,i} = v_{t,a,i}^\prime$ and $\sqrt{q}v_{t,a,i} = \sqrt{q}v_{t,a,i}^\prime$.

Given two commitments $\com_{t,a,i}$ and $\com'_{t,a,i}$, let $\com_{\text{diff},0} := \com_{0,t,a,i}-\com'_{0,t,a,i,}$, $\com_{\text{diff},1} := \com_{1,t,a,i}-\com'_{1,t,a,i}$, and $\com_{\text{diff},2} := \com_{2,t,a,i}-\com'_{2, t,a,i}$ 
With $\com_{t,a,i}, \com_{t,a,i}'$, conditioned on the success of proof of consistencies, we have
\begin{align}\label{eq:poe1}
    \nonumber
    \com_{\text{diff},0} &= \bfA(\mathbf{r}_{t,a,i}-\mathbf{r}_{t,a,i}')\\
    \nonumber
    \com_{\text{diff},1} &= \dotprod{\mathbf{s}_{1,i}}{\com_{\text{diff},0}} + \dotprod{\mathbf{e}_{1,i}}{\mathbf{r}_{t,a,i} - \mathbf{r}_{t,a,i}'} + (v_{t,a,i}-v_{t,a,i}')\\
    \com_{\text{diff},2} &= \dotprod{\mathbf{s}_{2,i}}{\com_{\text{diff},0}} + \dotprod{\mathbf{e}_{2,i}}{\mathbf{r}_{t,a,i} - \mathbf{r}_{t,a,i}'} + \sqrt{q}(v_{t,a,i}-v_{t,a,i}'),
\end{align}

We denote
$
    \tilde{C}_{t,a,i}=\begin{bmatrix}
        \com_{\text{diff},0}^T,0,0,0\\
        0,0,\com_{\text{diff},0}^T,0
    \end{bmatrix};
    \hspace{5mm}
        \tilde{u}_{t,a,i}=\begin{bmatrix}
        \com_{\text{diff},1}\\
        \com_{\text{diff},2}
    \end{bmatrix}
$.

To verify in zero-knowledge if $v_{t,a,i} = v_{t,a,i}'$ and $\sqrt{q}v_{t,a,i} = \sqrt{q}v_{t,a,i}'$, our construction allows the verifier to check the infinite norm of both $\dotprod{\mathbf{s}_{1,i}}{\com_{\text{diff},0}}-\com_{\text{diff},1}$ and $\dotprod{\mathbf{s}_{2,i}}{\com_{\text{diff},0}}-\com_{\text{diff},2}$. Furthermore, we give an upper bound on both of them, which is much smaller than $\sqrt{q}$, if and only if the conditions of equivalence hold. Our construction of ZKPoE is given in Protocol \ref{proto:pi_e}. We note that since the participant holds the knowledge of secret key corresponding to $\pk_{i}$ (for participant $i$), our ZKPoE essentially is proof of knowledge of secret key satisfying~\cref{eq:poe1} with appropriate norm bounds on $\mathbf{r}_{t,a,i}$ and $\mathbf{r}'_{t,a,i}$ from proof of consistencies. Since the proof construction remains similar to that of proof of consistency, we defer the ZKPoE Protocol $\pi_i^{Eq}$ to the Appendix \ref{app:construction_poe}.

We now show that if $v_{t,a,i}\neq v_{t,a,i}'$, then the bound resulting in attempting to decrypt message is large, that is $\norm{ \tilde{C}_{t,a,i} m_i -\tilde{u}_{t,a,i}}_\infty \geq \sqrt{q}/4$ where $\mathbf{m}_{i} = \mathbf{s}_{1,i}||\mathbf{e}_{1,i}||\mathbf{s}_{2,i}||\mathbf{e}_{2,i}$. 

\begin{lemma}
\label{lemma:comp_equil}
For $v_{t,a,i} \in \mathcal{R}_q$ and $v_{t,a,i}' \in \mathcal{R}_q$  
associated with two valid commitments $\com_{t,a,i}$ and $\com_{t,a,i}'$, respectively, the equality $v_{t,a,i} = v_{t,a,i}'$ holds if  $\norm{ \tilde{C}_{t,a,i} m_i -\tilde{u}_{t,a,i}}_\infty < \sqrt{q}/4, \norm{\dotprod{\mathbf{e}_{\{1,2\},i}}{\mathbf{r}_{t,a,i}-\mathbf{r}_{t,a,i}'}}_\infty \leq \sqrt q /4$.
\end{lemma}

\begin{proof}
Our proof is conducted by showing that its contraposition is correct and we consider the case where $\norm{\dotprod{\mathbf{e}_{\{1,2\},i}}{\mathbf{r}_{t,a,i}-\mathbf{r}_{t,a,i}'}}_\infty \leq \sqrt q /4$. That is, if $v_{t,a,i} \neq v_{t,a,i}'$, then 
\begin{align*}
    \norm{ \tilde{C}_{t,a,i} m_i -\tilde{u}_{t,a,i}}_\infty \geq \sqrt{q}/4
\end{align*}

Note that $ \norm{\tilde{C}_{t,a,i} m_i -\tilde{u}_{t,a,i}}_\infty  $ is the maximum of 
\begin{align*}
\{ k_0&=\norm{\dotprod{\mathbf{e}_{1,i}}{\mathbf{r}_{t,a,i} - \mathbf{r}_{t,a,i}'} + (v_{t,a,i}-v_{t,a,i}')}_\infty ,
\\ 
k_1&=\norm{\dotprod{\mathbf{e}_{2,i}}{\mathbf{r}_{t,a,i} - \mathbf{r}_{t,a,i}'} + \sqrt{q}(v_{t,a,i}-v_{t,a,i}')}_\infty \}. 
\end{align*}

Let $\delta=v_{t,a,i} - v_{t,a,i}'$, we first use the triangle inequality on $k_0$ and $k_1$, and obtain:

\begin{align*}
        k_0 &= \norm{(\dotprod{\mathbf{e}_{1,i}}{\mathbf{r}_{t,a,i}-\mathbf{r}_{t,a,i}'}+\delta)}_\infty \geq \abs{\norm{\delta}_{\infty}-\norm{(\dotprod{\mathbf{e}_{1,i}}{\mathbf{r}_{t,a,i}-\mathbf{r}_{t,a,i}'})}_\infty}\\
        k_1 &= \norm{(\dotprod{\mathbf{e}_{2,i}}{\mathbf{r}_{t,a,i}-\mathbf{r}_{t,a,i}'}+\sqrt{q}\delta)}_\infty 
        \\ & \geq \abs{\norm{\sqrt{q}\delta}_{\infty}-\norm{(\dotprod{\mathbf{e}_{2,i}}{\mathbf{r}_{t,a,i}-\mathbf{r}_{t,a,i}'})}_\infty}
    \end{align*}

\begin{itemize}
    \item  If $\delta_i \in(0,\sqrt{q}/2-1]$ or $\delta_i\in[-(\sqrt{q}/2-1),0)$ for any $i$-th coefficient $\delta_i$ of $\delta$: When the condition $k_0 \leq\sqrt{q}/4$ might still hold, if $q \geq 0$ (from $(q-1)/2 > \sqrt q(\sqrt q /2-1)$) we have 
    \begin{align*}
        k_1\geq \abs{\norm{\sqrt{q}\delta}_{\infty}-\norm{(\dotprod{\mathbf{e}_{2,i}}{\mathbf{r}_{t,a,i}-\mathbf{r}_{t,a,i}'})}_\infty} &\geq  \sqrt{q}-\sqrt{q}/4
        \\ &\geq\sqrt{q}/4,
    \end{align*}
    \item If $\delta_i \in[\sqrt{q}/2,\frac{(q-1)}{2}]$ or $\delta_i \in[-\frac{q-1}{2},-\sqrt{q}/2]$   for any coefficient $\delta_i$ of $\delta$: When the condition $k_1 \leq\sqrt{q}/4$ might still hold, we have, on the other hand,
    \begin{align*}
        k_0 \geq \abs{\norm{\delta}_{\infty}-\norm{(\dotprod{\mathbf{e}_{1,i}}{\mathbf{r}_{t,a,i}-\mathbf{r}_{t,a,i}'})}_\infty} \geq \sqrt{q}/2-\sqrt{q}/4=\sqrt{q}/4,
    \end{align*}

\end{itemize}

which concludes the proof.
\end{proof}

The formal statement of security properties is given as follows:
\begin{theorem}
    \label{thm:pi_e}
    Protocol $\pi_i^{Eq}$ is complete, (quantum) knowledge sound, and honest-verifier zero-knowledge in the classical setting.
\end{theorem}

The proof proceeds similarly to that of a ABDLOP linear proof, proof of opening, and approximate range proof given in \cite{LNP22} and the proof is given in the \cref{app:proof_of_equivalence} for completeness. The quantum proof of knowledge is shown in \cref{app:quantum_poc_poe_poa_etc}.

\textbf{Extending Proof of Equivalence to Prove Shortness of Secret Keys.}
To ensure \cref{lemma:comp_equil} can be used, the secret error term $\mathbf{e}$ need to be short as well. Since aforementioned shortness check of $\mathbf{e}$ and $\pi^{Eq}$ share the same structure but with a different linear relation check, we slightly modify protocol used for $\pi^{Eq}$ to prove the key shortness to achieve the proof $\pi^{Kw}$. The modified protocol is given in Appendix \ref{app:construction}.

\subsubsection{Zero-Knowledge Proof of Asset (ZKPoA)}
ZKPoA is used by a participant to prove they have enough balance, over a specific asset, to conduct an associated ongoing transaction. The current balance hence depends on the on-transaction amount and on all other past transacted amounts over that same asset for that same participant.

Let $\mathbf{r}_{\text{PoA}}$ be the new random value for the re-commitment and $v_{\text{tot}}$ be the sum of the amounts $v_{t,a,i}$. Here, for the $n$-th transaction over a fixed asset $a$ and some participant $i$, we compute the current balance by homomorphically summing all positive (inbound) and negative (outbound), past and on-transaction amounts $v_{t,a,i}$ (including the current pending tx $\widehat t$) and re-commit them to the re-commitments $\com^\prime$:
$
\com_{1,\widehat t,a,i}' = \pk_{1,i}^T\mathbf{r}_{\text{PoA}} +  \sum_{t=0}^n v_{t,a,i} = \pk_1^T\mathbf{r}_{\text{PoA}}+v_{\text{tot}}
$.
Note that we will also use proof of asset as a positive range proof for transaction commitments of receiving account.

{
\setlength{\textfloatsep}{0pt}
\setlength{\floatsep}{0pt}
\setlength{\intextsep}{0pt}
\setlength{\belowcaptionskip}{0pt}
\begin{figure}[!ht] 
    \centering
    \begin{adjustbox}{minipage=\linewidth,scale=0.85,center, frame}
    \pcnobox{
    \\
    \underline{\textbf{Participant (Prover)}} \<\< \underline{\textbf{Auditor (Verifier)}}
    \\\text{Input:}
    \\\bfA,\mathbf{a}_\text{bin}, \mathbf{a}_\text{bin}^\prime, \mathbf{a}_\text{g}\<\<\bfA,\mathbf{a}_\text{bin}, \mathbf{a}_\text{bin}^\prime, \mathbf{a}_\text{g}
    \\\mathbf{r}_{\text{PoA}}, v_\text{bin}, v_{\text{tot}}, \com_{1,\widehat t,a,i}^\prime\<\<\com_{1,\widehat t,a,i}^\prime
    \\[0.1\baselineskip][\hline] \\[-0.5\baselineskip]
    \mathbf{y}_i \xleftarrow[]{\$} \mathcal{D}_{\mathfrak{s}}^{(\mathrm{k}+\lambda+3)d}\\
    g\xleftarrow[]{\$}(\mathbf{0}_{d/l}\big|\mathbb{Z}_q^{d/l\times(l-1)}) \\
    f_i = \begin{pcmbox}\begin{bmatrix}
        f_{0,i} \\ 
        f_{1,i} \\
    \end{bmatrix}
    = \begin{bmatrix}
        \mathbf{A} \\
        \mathbf{a}_\text{bin}^T \\
    \end{bmatrix} \mathbf{r}_{\text{PoA}} + \begin{bmatrix}
        \mathbf{0^\kappa} \\
        v_\text{bin}\\ 
    \end{bmatrix}\end{pcmbox}\\
    \mathbf{u}_{1,i}= (\mathbf{a}_\text{bin}')^T\mathbf{r}_{\text{PoA}}+v_\text{bin}'\\
    \mathbf{u}_{2,i}= \mathbf{a}_\text{g}^T\mathbf{r}_{\text{PoA}}+g\\
    u_\text{3,i}=(\mathbf{a}_\text{bin}^T\mathbf{y}_i)^2+(\mathbf{a}_\text{bin}')^T\mathbf{y}_i \\
    \mathbf{w}_i = \bfA \mathbf{y}_i \\
    \< \sendmessageright*[6em]{f_i, \mathbf{u}_{1,i}, \mathbf{u}_{2,i}, u_\text{3,i}, \mathbf{w}_i} \\
    \<\< \phi_i \xleftarrow[]{\$} \mathbb{Z}_q^d \\
    \< \sendmessageleft*[6em]{\phi_i} \\ u_{4,i}=\big(\mathbf{a}_\text{g}^T+\mathbf{a}_\text{bin}^T\text{NTT}^{-1}(Q^T\phi_i) \\
    \tab\tab\tab-\pk_{1,i}\text{NTT}^{-1}(\phi_i)\big)\mathbf{y}_i \\
    h = (v_\text{bin})\text{NTT}^{-1}(Q^T\phi_i) \\
    \tab\tab\tab-v_{\text{tot}}\text{NTT}^{-1}(\phi_i) + g\\
    \< \sendmessageright*[6em]{h,u_\text{4,i}} \\
    \<\< c_i \xleftarrow[]{\$} \mathbf{\mathcal{C}} \\
    \< \sendmessageleft*[6em]{c_i} \\
    \mathbf{z}_i=\mathbf{y}_i+c_i\mathbf{r}_{\text{PoA}} \\
    \rej(\mathbf{z}_i,c_i\mathbf{r}_{\text{PoA}},\mathfrak{s}) \\
    \< \sendmessageright*[6em]{\mathbf{z}_i} \\
    \<\< \text{Runs } \mathcal{V}^A
    }
    \\[0.1\baselineskip]
    \rule{\linewidth}{0.4pt}
    \\[-0.2\baselineskip]
    $ \mathcal{V}^A:$ Accept if:
    \begin{enumerate}
        \renewcommand{\labelenumi}{(\alph{enumi})}
        \item \label{vrfy:pi_a_1} $\norm{\mathbf{z}_i} \stackrel{?}{<} \mathfrak{s}\sqrt{2(k+\lambda+3) d}$
        \item \label{vrfy:pi_a_2} $\bfA \mathbf{z}_i\stackrel{?}{=}\mathbf{w}_i+c_if_{0,i}$
        \item \label{vrfy:pi_a_3} $(\mathbf{a}_\text{bin}^T\mathbf{z}_i - c_if_{1,i})( \mathbf{a}_\text{bin}^T\mathbf{z}_i - c_if_{1,i}+c_i)+(\mathbf{a}_\text{bin}')^T\mathbf{z}_i -  c_i \mathbf{u}_{1,i}-u_\text{3,i}\stackrel{?}{=}0$
        \codecomment{Product Proof}
        \item \label{vrfy:pi_a_4} $c_i(\text{NTT}^{-1}(Q^T\phi_i)f_{1,i}-\text{NTT}^{-1}(\phi_i)\com_{1,\widehat t,a,i}^\prime+ \mathbf{u}_{2,i}-h)+u_\text{4,i}\stackrel{?}{=} \mathbf{a}_\text{g}^T\mathbf{z}_i+\mathbf{a}_\text{bin}^T\mathbf{z}_i\text{NTT}^{-1}(Q^T\phi_i)-\pk_{1,i}\mathbf{z}_i\text{NTT}^{-1}(\phi_i)$ 
        \item \label{vrfy:pi_a_5} $h_0,\dots,h_{d/l-1}\stackrel{?}{=}0$ 
        \codecomment{With (d) for Unstructured linear proof}
    \end{enumerate}
    \end{adjustbox}

    \caption{Protocol $\pi_i^A$}
    \label{proto:pi_a}
\end{figure}
}

Here, we give a 5-move interactive process of ZKPoA in Protocol~\ref{proto:pi_a}. Note that ZKPoA requires the prover to generate additionally a specially crafted BDLOP commitment as:

$
    f_i = \begin{bmatrix}
    f_{0,i} \\ 
    f_{1,i} \\
\end{bmatrix}
= \begin{bmatrix}
    \mathbf{A} \\
    \mathbf{a}_\text{bin}^T \\
\end{bmatrix} \mathbf{r}_{\text{PoA}} + \begin{bmatrix}
    \mathbf{0^\kappa} \\
    v_\text{bin}\\ 
\end{bmatrix}
$

\noindent with $v_\text{bin}$ as the inverse NTT of the binary representation of $v_{\text{tot}}$ value such that:
   $ v_\text{bin}=\text{NTT}^{-1}(\texttt{binary}(\allowbreak v_{\text{tot}})),$
and
\begin{itemize}
    \item $\mathbf{a}_\text{bin}, \mathbf{a}_\text{bin}^\prime, \mathbf{a}_\text{g} \leftarrow \mathcal{R}_q^{(\kappa+\lambda+3)}$ are public. 
    \item $v_\text{bin}^\prime=\dotprod{\mathbf{a}_\text{bin}}{\mathbf{y}}(1-2v_\text{bin})$. 
    \item $\mathfrak{s}$ denotes the standard deviation, where $\mathfrak{s}= 11 \omega \|  \mathbf{r}_{\text{PoA}}\|$.
\end{itemize}

ZKPoA aims to show $v_{\text{tot}}$ to be a 64-bit integer $\{0,1,\dots,2^{63}\}$, i.e., $v_{\text{tot}}$ is a non-negative integer as the current balance of participant $i$. To do so, we refer to \cite{LNPS21} and show the validity via the two different relations as follows: 
\begin{enumerate}
    \item $v_\text{bin}$ is indeed a polynomial created from a binary vector:
    \begin{align}
        \label{eq:ZKPoA_1st_cond}
        v_\text{bin}(1-v_\text{bin})=0    
    \end{align}
    This equation ensures that $v_\text{bin}$'s NTT vector is only made up of either 1's or 0's, making it indeed a binary encoding of a number. In particular, we have $
    \text{NTT}(v_\text{bin}) \circ (1-\text{NTT}(v_\text{bin})) = \text{NTT}(v_\text{bin}) \circ \text{NTT}(1-v_\text{bin}) = v_\text{bin}(1-v_\text{bin}).$
    \item The value $v_\text{bin}$ encoded in the commitment is indeed the binary representation of $v_{\text{tot}}$:
    \begin{align}
        \label{eq:ZKPoA_2nd_cond}
        Q(\text{NTT}(v_\text{bin}))=\text{NTT}(v_{\text{tot}})
    \end{align}
    where $Q$ is a zero-padded Vandermonde matrix while $\tilde Q$ is the embedded Vandermonde matrix evaluated at $2$: 
    \begin{align*}
    \tilde Q&=\begin{bmatrix}
    1 & 2 & \hdots & 2^{63} \\
    \vdots & \vdots & \vdots & \vdots \\
    1 & 2 & \hdots & 2^{63}
    \end{bmatrix}
    \\
    Q&=\begin{bmatrix}
    \tilde Q & 0^{l-64}\\
    \end{bmatrix}
    \end{align*}
\end{enumerate}


\begin{theorem}
    \label{thm:pi_a}
    Protocol $\pi_i^A$ is complete, (quantum) knowledge sound, and honest-verifier zero-knowledge.
\end{theorem}

The proof proceeds similarly to that of the BDLOP product proof from \cite{ALS20}, unstructured linear proof from \cite{ENS20}, proof of opening and the proof is given in the \cref{app:proof_of_asset} for completeness. The quantum proof of knowledge is shown in Appendix \ref{app:quantum_poc_poe_poa_etc}.

\section{A Lattice-based Transaction Scheme for Encrypted Table-based Ledger}
\label{sec:transaction_scheme}

\label{sec:model}

In this section we modify the model from MiniLedger, an ETL-like transaction scheme proposed in \cite{chatzigiannis2021miniledger} to support multi-asset transactions while omitting transaction pruning and auditing for simplicity. Let $\ledger$ to be the ledger state which act as a database that stores (i) a list of verified transactions $\txlist$ where $\{\tx_t\}_{t \in \txlist}$ and (ii) a list of registered participants $\pklist$. Let $\assetlist$ be the set of asset. Both of the lists are accessible using $\txlist_{\ledger}$ and $\pklist_\ledger$ respectively with respect to a ledger state $\ledger$. A row in the ledger is illustrated in Figure \ref{fig:etl_table} and is query-able using 3-dimensional array indexing with the indexing ordering of $\tx$-row, $\asset$-subsrow, and $\pk$-column. For example, 
$\ledger[\tx_t][\asset_a][\pk_i]$ will return $v_{t,a,i}$. In addition, we overload the array indexing to allow column retrieval for an associated public key $\pk_i$ and an asset $\asset_a$, for example $\ledger[\asset_a][\pk_i]$ will return the list of values associated with the public key $\pk_i$ and the asset $\asset_a$ as $\{v_{t,a,i}\}_{t\in \txlist_\ledger}$.
We denote $\tx \rightarrow \ledger$ as appending $\tx$ to the new row of ledger $\ledger$.
We assume, similar to that of \cite{matrict_2_2022}, there exist a consensus protocol that manage the state of the ledger and the exact protocol is outside the scope of the paper. The symbols used for the formal model are summarized in Table \ref{tab:symbol_table}.

\begin{definition}[$\etl$]
    An encrypted table-based transaction scheme $\etl$ is a tuple of algorithms, $\etl = \{\setup,\keygen, \allowbreak \createtx,\allowbreak\verifytx, \checkbalance\}$ defined as follows: 
    \begin{itemize}
        \item{$(\publicparam, \ledger) \gets \setup(\secparam, \bank, \valuelist)$: on input a security parameter $\secparam$, a participant set $\bank$, a genesis value list $\valuelist$, outputs a public system parameter $\publicparam$ which is passed implicitly to all algorithms and an initial ledger state $\ledger$.}
        \item{$(\sk, \pk, \pi) \gets \keygen(\mathbf{A})$: on input a commitment key matrix $\mathbf{A}$, outputs a key pair consisting of a secret key $\sk$, a public key $\pk$, and a proof $\pi$.}
        \item{$\{\tx\} \gets \createtx(\valuelist, \sklist, \ledger)$}: on input a value list $\valuelist$, a secret key list $\sklist$, and a ledger state $\ledger$, outputs a transaction $\tx$.
        \item{$\{0,1\}\gets \verifytx(\tx, \ledger)$: on input a transaction $\tx$ and a ledger state $\ledger$, outputs a verification indicator bit.}
        \item{$v \gets \checkbalance(\sk, a, \ledger)$: on input a secret key $\sk$, an asst index $a$, outputs a balance value $v$.}
    \end{itemize}
    \label{def:etl}
\end{definition}

\textbf{Security Definition.}
We modify the correctness, balance with ownership proof and privacy (k-anonymity) properties for the ETL-like transaction scheme from \cite{chatzigiannis2021miniledger} to further support multi-asset transaction.  Correctness definition is deferred to \cref{app:relation}. We consider the static corruption case. We define a valid value list set $\valueset := \{\valuelist_0,\valuelist_1, \cdots\}$ with respect to a ledger state $\ledger$ such that for any value list $\valuelist \in \valueset$, we have $\sum_{i} v_{a,i}=0$ for any asset $a$ where $v_{a,i} \in \valuelist$ and $\sum_{t \in (\txlist_\ledger \cup \valuelist)} v_{t,a,i} \geq 0$ for any asset $a$ and any participant $i$. We describe the oracles given to an adversary $\adv$ and the adversary implicitly has access to the ledger state $\ledger$.

\begin{itemize}
    \item $\mathcal{O}_{\tx}(\valuelist)$: on input a value list $\valuelist$, if the value list is valid $\valuelist \in \valueset$,  executes $ \tx \gets \createtx(\valuelist, \sklist, \ledger)$, records $\tx$ to $\txlist_{\mathcal{O}}$ and outputs $\tx$, otherwise outputs $\bot$.
    \item $\mathcal{O}_{\ledger}(\tx)$: on input a transaction $\valuelist$, if $\verifytx(\tx, \ledger)=1$, appends to ledger $ \tx \rightarrow \ledger$ then outputs $(1, \ledger)$, otherwise outputs $(0,\ledger)$.
\end{itemize}

Informally, balance with ownership proof requires that only the account owner can spend the associated assets.
Balance property requires that balance is preserved such that the sender cannot spend more than its total assets and the total value in the system are preserved after transaction. These requirements are the main integrity requirements for a ledger as shown in various ETL schemes \cite{chatzigiannis2021miniledger, PADL}.

\begin{definition}[Balance] It holds that every PPT adversary $\adv$ has at most negligible advantage in the following experiment, where we define the advantage as $\advantage{\expbalance}{\adv} := \prob{\expbalance_{\adv}(\secpar) = 1}$.
\begin{adjustbox}{minipage=\linewidth,scale=0.85, center}
    \procb{$\expbalance_{\adv}(\secpar)$}{
        (\bank, \bank_{C}, \valuelist) \gets \adv() \\
       \{v^a_{tot}\}_{a \in \assetlist} := \sum_{i} v_{0,a,i} \text{ where } v_{0,a,i} \in \valuelist \\
        \publicparam, \ledger \gets \setup(\secparam, \bank, \valuelist) \\
        \widehat\tx \gets \adv^{\mathcal{O}_{\tx}, \mathcal{O}_{\ledger}}(\publicparam, \sk_{C}), \text{ where } \sk_{C}=\{\sk_i\}_{i \in \bank_{C}}\\
        \{{\widehat v}^a_{tot}\}_{a \in \assetlist} := \sum_{t \in {(\txlist_\ledger \cup \widehat\tx),i \in \pklist_\ledger}} v_{t,a,i}\\ 
        \text{return 1 only if }\verifytx(\widehat \tx, \ledger)=1 \wedge (\\
        \phantom{\vee\ \ } 1) \exists a \in \assetlist,i \notin \pklist_{cor}, v_{\widehat t,a,i}\in \widehat \tx, \text{s.t. }v_{\widehat t,a,i} < 0 \codecomment{bypassing ownership}\\
         \vee\ 2) \exists a \in \assetlist, i \in \pklist_\ledger, \text{s.t. } v_{a,i}^{sum} = \sum_{t \in \txlist_\ledger \cup \tx^\dagger} v_{t,a,i} < 0\\
         \codecomment{spending more asset than owned}\\
         \vee\ 3) \exists a \in \assetlist, \text{s.t. }{\widehat v}^a_{tot} \neq v^a_{tot}  \codecomment{create or destroy total asset value}
    }
\end{adjustbox}
    \label{definition:exp_balance}
\end{definition}

\begin{definition}[Privacy] Every PPT adversary $\adv$ has at most negligible advantage in the following experiment, where we define the advantage as $\advantage{\expprivacy}{\adv} := |\prob{\expprivacy_{\adv}(\secpar) = 1}  - 1/2 |$.

\begin{adjustbox}{minipage=\linewidth,scale=0.85, center}
    \procb{$\expprivacy_{\adv}(\secpar)$}{
        (\bank, \valuelist) \gets \adv() \\
        \publicparam, \ledger \gets \setup(\secparam, \bank, \valuelist) \\
         \widehat \valuelist_0, \widehat \valuelist_1 \gets \adv^{\mathcal{O}_{\tx}, \mathcal{O}_{\ledger}}(\publicparam) \\
         \text{abort if } \widehat \valuelist_0 \notin \valueset \vee  \widehat \valuelist_1 \notin \valueset \\
         b \sample \{0,1\}, \widehat \tx_{b} \gets \createtx(\widehat \valuelist_b, \sklist,\ledger) \\
         \widehat \tx_{b} \rightarrow \ledger \\
         b' \gets \adv()\\
         \text{return 1 only if } b=b'
    }
\end{adjustbox}
    \label{definition:exp_privacy}
\end{definition}

\noindent Informally, the above captures that the ledger hides both transacting parties and the associated values with $k$-anonymity for inputs, outputs, and asset types across assets. 

\subsection{$\name$ Construction}
\label{sec:qpadl_con}
We construct $\name$ using ZKPs defined in Section \ref{sec:zkp_construction} to enforce the integrity of the ledger (Balance property defined in Definition \ref{definition:exp_balance} while maintaining the privacy of the transactions where the privacy is defined in Definition \ref{definition:exp_privacy}). In summary, we commit to the transaction values using BDLOP commitment scheme and prove that it is well-formed with respect to the ledger state. The construction of $\name$ as an instance of ETL (\cref{def:etl}) is given in \cref{fig:algorithm_qpadl} and setup algorithms are given in \cref{fig:algorithm_qpadl_setup}. We give an informal description of the ETL construction below:
\begin{enumerate}
    \item On $\createtx(\valuelist, \sklist, \ledger)$, spender commits to the values in $\valuelist$ (including decoy values of zero), generate all the necessary proofs and then return both the commitments and the proofs as a new transaction.
    \item On $\verifytx(\tx, \ledger)$, the ledger return 1 if all proofs are verified to be correct.
    \item On $\setup(\secparam, \bank, \valuelist)$, generate public keys, set the genesis transaction value as $\valuelist$, then return the public parameter and ledger. In practice, participants will register their public keys and the registration depends on the verification of the $\pi^{KW}$.
    \item On $\keygen(\mathbf{A})$, sample $\mathbf{s}$ and $\mathbf{e}$ twice to produce public keys with well-formness proof $\pi^{KW}$, then return them.
    \item On $\checkbalance(\sk_{i}, a, \ledger)$, run the $\extract(\Tilde{\com}, \sk_{i})$ from \cref{proto:exctractability} and return the value extracted.
\end{enumerate}

The transaction scheme flow is illustrated in \cref{fig:PADL_transaction_flow} and discussed in details in \cref{sec:etl_transaction}.
We apply Fiat-Shamir transform \cite{fiatshamir} to the Sigma protocols to obtain 
non-interactive zero-knowledge proofs.
To further hide the sign for (spending/receiving) accounts in the transaction, ZKP Sigma protocol OR-proof \cite{CDS94, Fischlin2020} \footnote{For the 7-moves sigma PoE protocol, the OR-Proof is run on only the challenge $c_i$ right before the last response from the proof of equivalence protocol since the verification equation can be made true as long as the last phase commitment $h,w_i,v_i$ is chosen after everything else is sampled. The protocol before the last three moves are part of the statement for the OR-proof.} 
can be used for a 1-out-of-2 proof for asserting the new commitment $\commitment'$ is either re-commitment to the sum of account proven with a proof of equivalence or a recommitment to the same commitment message $m=m'$ with simple BDLOP linear proof.
The relations enforced by the NIZK is deferred to the Appendix \ref{app:relation}. We further construct an extension $\nameextension$ by replacing the proof of asset to the compact version, briefly described in \cref{sec:compact_overview}, introduced formally in \cref{sec:compact} and constructed in \cref{proto:pi_a2}. The security analysis in the next section applies to the extended case.

\subsection{Security Analysis}

We provide security analysis for $\name$ under (quantum) random oracle model (ROM/QROM) in this section. $\name$ is a correct (defined in \cref{def:etl_correctness}) $\etl$ transaction scheme follows trivially by inspection. We instead focus on the balance and privacy.

\begin{theorem}[Balanced]
    $\name$ is a balanced $\etl$ transaction scheme in (Q)ROM, assuming that NIZK is zero-knowledge and (quantum) sound with validity (extractable), BDLOP and ABDLOP commitment schemes are binding, MLWE is hard, and MSIS is hard.
\end{theorem}

\begin{theorem}[Privacy]
    $\name$ is a private $\etl$ scheme, assuming that commitment is hiding, NIZK is ZK and MLWE is hard.
\end{theorem}

\textit{Proof Sketch. } The full proof is provided in Appendix \ref{app:sec_proof_analysis}. 
I) For balance:  the adversary are bounded to the committed messages (binding). Then, since the proofs are zero-knowledge and sound, the adversary have to simultaneously fulfill the conditions specified in the relation which rules out winning conditions (2) spending more asset than owned and (3) create/destroy total asset value. The remaining chance for the adversary to win is to forge a signature on a public key without secret key, but we can extract a MSIS solution out of the adversary upon forgeries and we thereby conclude that the $\name$ is balanced. In the QROM model, the result follows because ZKP protocols used are proven to be a quantum proof of knowledge scheme and random oracle is not used anywhere else. II) For privacy: Since the public key is indistinguishable from random and commitment scheme is hiding, the adversary cannot distinguish the committed value. Moreover, the proofs provided are zero-knowledge. Thus, the adversary gains no information from observing commitments and proofs provided in the transaction.

\subsection{Auditing}
\textbf{Customizable Auditing.} To enable customizable auditing, we modify the commitment key for the participant to further contain an auditor's public key, then modify the PoC to check the corresponding consistency.
This enables the auditor to decrypt the value. A participant can have more than one associated audit key and different participants can have different auditors' public keys. 

\textbf{Fine-Grained Selective Disclosure with ZKP.} Customizable audit keys avoid a single central audit key but still expose the full transaction history to an auditor. In practice, auditing can be scoped 
(e.g., liquidity, concentration, or solvency checks \cite{PADL,Dagher2015,narula2018zkLedger}) 
via purpose-bound zero-knowledge proofs that reveal only the required aggregate property without revealing individual transactions. We describe it in a generic setting in \cref{sec:zkp_auditing}.

\section{Implementation, Evaluation and Discussion}
\label{sec:implementation}
\label{sec:parameter}

\noindent\textbf{Parameters.}
We first set q to be around $2^{100}$ and the degree $d$ to be at least 256 so that the vectors in PoC and PoE proofs can be encoded using a single ring element, and then set the rest of the parameter to reach the target security level of around $2^{-128}$ estimated by \cite{lattice_estimator} given the parameters and target root Hermite factor of $\approx 1.0043$. The method for setting the parameters follows methodologies from \cite{esgin_thesis}. We use binomial distribution defined in Section \ref{sec:prelim_distribution} for error distribution and secret distribution $\chi$ for public key generation and the random values vector sampled for the commitments. We set $l$ to be $128$, $\lambda=16$ and $\kappa = 16$. We use $d=1024$ with appropriate parameters derived similarly as before.

\textbf{Performance and Communication Analysis Summary.} We implemented the proposed $\name$ scheme in Rust and benchmarked the implementation\footnote{PoC available at: \url{https://github.com/jpmorganchase/PADL}} using a M3 Pro Macbook. We state the summarized performance number and defer detailed analysis to \cref{sec:performance}. For the case of $\nameextension$, it incurs \textbf{4.26/4.27} ms transaction generation time and \textbf{0.34/0.37} ms verification time per participant per asset for $d=256/1024$ respectively. The proof size is \textbf{5.15/1.87} KB per participant per asset.

\textbf{Discussion.} We defer the comparison of lattice-based ETL and RingCT to the Appendix and provide the summary in \cref{tab:compare}. We also defer the discussion on account state compression, proving time optimization with Hint-MLWE, and others to the Appendix \ref{app:discussion}. 
{

\begin{table}[ht!]
\begin{center}
\rowcolors{1}{}{lightgray}
\begin{adjustbox}{minipage=\linewidth,scale=0.85}
\begin{tabular}{r|r|r}
  \hline
 (I) & $\nameextension$ (This) & MatRiCT+  \cite{matrict_2_2022} \\
  \hline
  Anonymity  & $k$-anonymity & I - Ring Signature \\
  & & O - Stealth Address\\
  AssetAnon & Yes & No \\
  Spendable & Always & Honest Sender only \\
  Auditability & Customizable Audit Key & Central Audit Key \\
  Multi-Asset & Hundred(s)-Thousand(s) & 1 Asset \\
  Scaling & Linear in Input/Output Acc & Linear in OuputsAcc\\
   &  & Polylog in InputAnon\\
      & Compact in Asset up to $d$ & Linear in Asset\\
   \hline
\end{tabular}
\end{adjustbox}

\begin{adjustbox}{minipage=\linewidth,scale=0.85}
\begin{tabular}{r|r|r|r}
  \hline
 (II) & $\nameextension$ & MatRiCT+ \cite{matrict_2_2022} & zkLedger \cite{narula2018zkLedger} \\
 Type & ETL-Lattice & RingCT-Lattice & ETL-DLog \\
  \hline
  $\tx$ Gen & $4.26$ / $4.27$ ms & $100$ ms & $8$ ms\\
  $\tx$ Verf & $0.34$ / $0.37$ ms & $2$ ms & $5.7$ ms \\
  $\tx$ Size & 5.15 / 1.87 KB & 47 KB & 4.69 KB \\
   \hline
\end{tabular}
\end{adjustbox}
\end{center}
\caption{ (I) Comparison of ETL and RingCT. (II) Quantitative results are amortized with respect to 256/1024 asset size. Performance results are benchmarked on different machines.}
\label{tab:compare}
\end{table}}



\section*{Disclaimer}
This paper was prepared for informational purposes by the Global Technology Applied Research center of JPMorganChase. This paper is not a product of the Research Department of JPMorganChase or its affiliates. Neither JPMorganChase nor any of its affiliates makes any explicit or implied representation or warranty and none of them accept any liability in connection with this position paper, including, without limitation, with respect to the completeness, accuracy, or reliability of the information contained herein and the potential legal, compliance, tax, or accounting effects thereof. This document is not intended as investment research or investment advice, or as a recommendation, offer, or solicitation for the purchase or sale of any security, financial instrument, financial product or service, or to be used in any way for evaluating the merits of participating in any transaction.

\bibliographystyle{ACM-Reference-Format}
\bibliography{ref}

\appendix
\crefalias{section}{appendix}
\crefalias{subsection}{appendix}

{
\setlength{\belowcaptionskip}{0pt}
\begin{table}[ht!]
\begin{adjustbox}{minipage=\linewidth,scale=0.85, center}
\begin{tabular}{|c|l|}
\hline
\textbf{Symbol} & \textbf{Description} \\
\hline
$v$ & transaction value \\
\hline
$\valuelist$ & value list consisting of transaction amounts \\
\hline
$\valueset$ & set of valid value list  \\
\hline
$\sk$ & private key of an account \\
\hline
$\pk$ & public key of an account \\
\hline
$\sklist$ & list of private key \\
\hline
$\mathsf{PK}$ & list of public key \\
\hline
$\pklist$ & list of participants \\
\hline
$\txlist$ & list of transaction \\
\hline
$\assetlist$ & list of asset \\
\hline
$\ledger$ & ledger state \\
\hline 
$t,a,i$ & transaction row $t$, asset index $a$, and participant index $i$ \\
\hline 
\end{tabular}
\end{adjustbox}
\caption{Symbol Table for the Formal Model}
\label{tab:symbol_table}
\end{table}
}

\section{Deferred Discussion}
\label{app:discussion}

\subsection{Comparison of Lattice-based ETL and RingCT}

\subsubsection{Qualitative Comparisons}
While RingCT (MatRiCT+ construction) suits the privacy-preserving transaction scheme for public blockchain transaction, we argue that ETL ($\name$ construction) is better suited for financial institution for the following reasons. 
First of all, $\name$ achieves $k$-anonymity through masking the real transaction among multiple decoy input and output accounts activities, while MatRiCT+ achieves anonymity through ring signature for input accounts and stealth address for output accounts. The stealth address complicates zero-knowledge  auditing for an account state as the statement need to be scan over the entire transaction history, while the ETL account state can be succinctly represented using a summed commitment. Every token received by an account is spendable in $\name$ but not in MatRiCT+ as the encryption of spending token key is not proved to be correct and is assumed to be communicated honestly. This leads to a scenario whereby the token asset would appear on the balance sheet but can never be spent by the bank. $\name$ supports customizable audit key whereby each participant commitment matrix can be augmented to open to specific custom auditor(s) while MatRiCT+ only permits a central audit key. In addition, $\name$ supports multi-asset transaction natively and compactly under $\nameextension$ which is missing from MatRiCT+ who only supports single-asset transaction by default. 

The limitation for ETL is that the ETL transaction size grows linearly per participant: specifically, it is linear in N (input/output accounts, including decoys). Meanwhile, the MatRiCT+ proof length is logarithmic in M (input accounts), polylogarithmic in N (decoy accounts), and linear in S (output accounts). On the other hand, $\nameextension$ construction shows that ETL can achieve better scaling with asset compactness with respect to the asset dimension while RingCT scheme scale linearly with number of asset. The comparison results are summarized in \cref{tab:compare}.

\subsubsection{Quantitative Comparisons}

We would like to highlight that the structure of transactions supported by $\nameextension$  and MatRiCT+ is different. Nevertheless, we provide a comparison that highlight under multi-asset transactions, the overall performance and communication profile of both transaction schemes. 
We compare the communication size using the number reported in MatRiCT+ \cite{matrict_2_2022}. 
In particular, we compare the size without applying any compression techniques. 
It is possible to further compress the ABDLOP commitment and its opening using technique from Dilithium \cite{Ducas2018} as done in \cite{LNP22, matrict_2_2022}. Since MatRiCT+ does not support multi-asset transactions natively, we assume linear scaling which is the optimistic case for MatRiCT+. We consider the case that participant $A$ would like to transfer to participant $B$ some asset(s) and hide the type(s) of assets being transferred. When amortized across $256/1024$ number of assets, it can be seen that $\nameextension$ outperforms MatRiCT+ where $\nameextension$ costs $5.15/1.87$ KB per participant per asset when compared to MatRiCT+ $47$ KB. 

For transaction generation and verification, we use the performance figures reported by \cite{matrict_2_2022}. The MatRiCT+ reference implementation relies on machine assembly-level x86 intrinsic functions to optimize performance-critical sections, making it incompatible with the ARM-based Mac used for our benchmarks. Removing these x86-specific optimizations while adding back in comparable ARM-specific optimizations would require substantial changes to the codebase and significant effort to re-verify the resulting structural modifications. We therefore adopt the numbers reported by MatRiCT+ as a reference point and note that the two implementations were benchmarked on different machines. Amortized across 256/1024 assets, $\nameextension$ takes $4.26/4.27$ ms to generate a transaction while MatRiCT+ uses $100$ ms, while transaction verification in $\nameextension$ takes $0.34/0.37$ ms while MatRiCT+ takes $2$ ms. Even when accounting for machine configuration differences, we can conclude that $\nameextension$ outperforms MatRiCT+ under the setting of multi-asset transactions for financial institution use cases.  
\begin{rev}
We can further observe that lattice-based ETL has similar performance and communication footprint when compared to DLog-based ETL zkLedger. zkLedger's performance and communication metrics as reported in \cite{narula2018zkLedger} are $8$ ms, $5.7$ ms and $4.69$ KB for transaction generation time, verification time, and communication size respectively.
\end{rev}

\subsection{Performance and Communication Analysis}
\label{sec:performance}

\textbf{Performance Analysis.} We implemented the proposed $\name$ scheme in Rust and benchmarked the implementation using a M3 Pro Macbook. The reported performance is the mean score reported by the Rust Criterion benchmarking framework with sample sizes of 100. The proof of key well-formness is a one-time setup cost and has a similar performance ad size footprint of that of proof of equivalence. We observe that proof of consistency and proof of equivalence is dominating the runtime cost with a proving time of $161.12$ms and $141.96$ms respectively. 
Meanwhile, proof of asset, proof of balance and proof of equivalent with known opening take $16.59$ms, $11.29$ms and $9.38$ms respectively. Commitment for value message (and re-commitment) takes around $1.08$ms. The verification is comparatively faster and only takes about around $25$ms in total.
Therefore, the time to generate the transaction per participant in a new transaction is $503.62$ms. The proof of asset compact incurs $604.41$ms proving and $60.82$ms verification time. For the case of $\nameextension$ where we replace proof of asset to the compact version and commit asset values in each coefficient,
it has \textbf{4.26ms} transaction generation time and \textbf{0.34ms} verification time per participant per asset. Amortized to 1024 assets, $\nameextension$ takes $4.27$ ms to generate the transactions while transaction verification in $\nameextension$ takes $0.37$ ms per participant per asset.

\textbf{Communication Analysis.}
Using the parameters, we calculated the communication size. For the response in the ZKP, we use the discrete gaussian tail-bound lemma from \cite{Lyu12} to upper-bound the range to be within $-6\mathfrak{s}$ and $6\mathfrak{s}$ with high probability. For a transaction $\tx$, the size per participant per asset can be calculated as $|\tx|=2\cdot |\com| + |\pi^{PoB}|/(|\pklist|) +2\cdot|\pi^{PoC}| + |\pi^{PoE}| + |\pi^{PoE2}| + |\pi^{PoA}|$ where commitments for commit-and-prove proof are included in the proof $\pi$. 
With the compact mode, $\nameextension$ 
transaction size per participant per asset 
is amortized to around \textbf{5.15 KB}. For when $d=1024$, we have the amortized size is \textbf{1.87 KB} per participant per asset.

\subsection{Others}

\textbf{Account State Optimization.} Account state can be succinctly represented by the sum of the commitment tokens in its column. 
The participant can "compress" its column by re-committing to the sum of the commitment value. 
The compressed account state enables efficient zero-knowledge based account state auditing. 

\textbf{Commitment Value is Non-Integer for Single-Asset.} 
This scenario cannot happen if the genesis transaction $\tx_0$ contains only integer provided by the ledger issuer. Conditioned on the aforementioned criterion, introducing non-integer message when preparing the new transaction commitments will be caught by proof of asset which checks for integer value in the message. In the event of malicious issuer (which almost does not exist in practice for general use-cases), the spender can provide a corrector message $\mathfrak{m}$ as part of the transcript that contain the negation of the non-integer coefficients such that $\mathfrak{m}+m=int(m)$ and $\com_1+\mathfrak{m}=\mathbf{Br}+m+\mathfrak{m}=\mathbf{B}r+int(m)$ for some commitment matrix $\mathbf{B}$ where $int$ mask all but the integer coefficients. $\mathfrak{m}$ can be checked to only contain non-integer coefficients by the verifier. Note that this section does not concern the mutli-asset extension $\nameextension$.

\textbf{Hint-MLWE for Proving Time Improvement.} 
Rejection sampling is sufficient but not necessary for achieving simulatability for zero-knowledge sigma protocol. In \cite{hintmlwe}, the authors reduced MLWE to Hint-MLWE under discrete Gaussian setting whereby the witness is sampled from discrete Gaussian. It enables simulation without the needs to restart the protocol since the ZKP response can be an instance of Hint-MLWE sample. Similar technique can be applied to the ZKP presented here to reduce the proving time.

\textbf{Resetting $\mathbf{r}$ in the Summed Commitment State.}
During operation, the summed “randomness” length is reset via compression of state using re-commitment performed by the corresponding participant. Unlike the consumer setting, participants in the financial-institution setting are expected to be online periodically to perform re-commitment (a batch job at the end of the trading day).
Verification function can serve as a guardrail to only allow columns with history state shorter than the fixed uncompressed height to perform transaction (including acting as decoys). A larger $q$ can be used to process a larger amount of uncompressed transactions before compression is required.

\section{$\nameextension$ Extension}
\label{sec:compact}
We describe an extension to $\name$ that compactly encode each $d-$asset value into each coefficient of a ring element with degree $d$. 
We construct $\nameextension$ by replacing proof of asset $\pi^{PoA}$ from the base scheme to simultaneously perform a range check on the committed message value encoded in different coefficients using $\pi^{PoA'}$. The strategy is similar but we instead prove the relations compactly using random linear combinations and the check is over the coefficient domain. 
In short, we first flatten the binary representation of the encoded asset values across a larger vector and then commit to it. Then, we perform a quadratic relation check on the committed binary representation is binary using the standard binary check $\langle v_{bin}, v_{bin}-\mathbf{1}\rangle=0$ and then we check that it can be reconstructed back to the original value using a linear relation $\langle pow(\beta),  v_{bin}\rangle= v_i$ for the corresponding coefficient $i$. Lastly, we perform a approximate $\ell_2$ norm bound check to ensure that the binary check holds over integer (i,e. no overflow). With these, we conclude that $v_i < \beta$ where $v=(v_1,\cdots,v_d)$. Note that all these relations can be compactly encoded as a single quadratic relation that vanishes at the constant coefficient and therefore we used the single quadratic relation proof from \cite{LNP22}.  Since polynomial addition is coefficient-wise, proof of balance still checks that the total asset value for each asset sum to zero.
It follows easily that $\nameextension$ is a balanced and private ETL transaction scheme since the replaced proof of asset protocol $\pi^{A'}$ is a (quantum) proof of knowledge protocol that checks all asset values are positive and there are no other changes made to the base scheme. 
We defer the formal construction to the full paper.

We further propose how to perform positive range check on all asset in one-shot using any generic infinity norm range proof. Compared to the aforementioned approach, the following approach is more general.
To perform positive range check for polynomial ring coefficients using some bound $B$, it is sufficient to perform a negative shift of $B/2$ and then check its infinity norm which is captured in the following lemma.

\begin{lemma}[Shifted Positive Range Proof]
\label{lemma:range_proof_plus}
If $\norm{ r- s }_{\infty} \leq B/2$ for $r = (r_1,\cdots,r_d) \in \mathcal{R}_q$ and $s = (B/2,\cdots,B/2)\in \mathcal{R}_q$, then $\forall i,r_i \in [0, B]$ for a given bound $B \in [0, (q-1)/2]$ and $B$ is even. 
\end{lemma}

\begin{proof}
We prove the contrapositive form, if $\exists i, r_i < 0 $ or $r_i > B$, then $\norm{ r- s }_{\infty} > B/2$. Let $Q=(q-1)/2$ for notational convenient. For arbitrary $i$ such that $-Q\leq r_i < 0 $ or $B < r_i \leq Q$, we consider the cases separately:

Case $r_i > B$: It is easy to see that $\norm{r_i-s_i}_{\infty} > B/2 $ since $r_i-s_i > B - B/2 = B/2$, then it immediately follows that $\norm{r-s}_{\infty} > B/2 $.

Case $r_i < 0$: We further subdivide the case into $-Q + B/2 \leq r_i < 0$ and $-Q \leq r_i < -Q + B/2$.
For $-Q + B/2 \leq r_i < 0$, we have $-Q + (B/2 - B/2) \leq r_i - s_i< 0- B/2 \equiv -Q \leq r_i - s_i < - B/2$, it immediately follows that $\norm{ r- s }_{\infty} > B/2$.

For the remaining case where $-Q \leq  r_i < -Q + B/2 $, we have $-Q-B/2 \leq  r_i - s_i < -Q + B/2 - B/2 $, converting to modulo q representation we have $q+(-Q-B/2) \leq r_i - s_i < q+(-Q) \equiv Q+1-B/2 \leq r_i-s_i < Q+1$. Then, we have  $r_i - s_i > B/2$ if $Q+1- B/2 > B/2$ or $Q+1 > B$. It immediately follows that $\norm{ r- s }_{\infty} > B/2$.
\end{proof}
\section{Privacy-Preserving Auditing}
\label{sec:zkp_auditing}

In this section, we give the notion of ZKP-based auditing over the state of the ledger formally. The ZKP-based auditing scheme simply runs the NIZK proving system for the interested relation $\relation$ over the encrypted ledger state.
 Transaction $\tx_{t}$ is composed of the commitments $\com_{t,a,i}$ and the proofs $\pi_{t,a,i}$ as $\tx_{t} := \{ \com_{t,a,i}, \allowbreak \pi_{t,a,i}\}_{a \in \assetlist, i \in \pklist}$. Auditor is given a history-recording ledger $\ledger$ defined as  $\ledger:=\{\tx_t\}_{t \in \txlist}$.
 Let $\publicparameter, \crs$ be generated during the setup phase of ETL scheme. A ZKP-based auditing scheme  $\auditscheme := (\auditprove, \auditverify)$ over an $\etl$ scheme where the commitment entries or transaction history is recorded in a table-like manner as shown in \cref{fig:etl_table} is defined as follows:

\begin{definition}[ZKP Auditing Scheme]
\label{def:auditing_def}
A ZK-based auditing scheme $\auditscheme$ constructed on top of a ledger state $\ledger$ maintained by the $\etl$ scheme consists of the following set of algorithms:
	\begin{itemize}
    \item $\auditprove(\crs, \relation, \ledger, \wit )$: on input a common reference string $\crs$, an auditing relation $\relation$, a ledger state $\ledger$, and a witness $\wit$, it executes $\pi \gets\nizkprove_\relation(\crs, \ledger, \wit)$ and returns the proof $\pi$.
    \item $\auditverify(\crs, \relation, \ledger, \pi)$: on input a common reference string $\crs$, an auditing relation $\relation$,  a ledger state $\ledger$, and a proof $\pi$, it executes $b \gets\nizkverify_\relation(\crs, \ledger, \pi)$ and returns the verification status $b$.
	\end{itemize}

We define $\auditscheme$ to satisfy the following properties similar to privacy-preserving auditing for proof of solvency from \cite{Dagher2015}:
\begin{itemize}
    \item \textbf{Correctness}. If $\relation(\ledger, \wit)=1$, then a prover $P$ can generate a proof $\pi$ using $\auditprove$ such that $\auditverify(\crs, \relation, \ledger, \allowbreak \pi)=1$ with overwhelming probability. 
    \item \textbf{Soundness}. If $\ledger \notin \lang_\relation$, then an auditor verifier $V$ running the verification algorithm $\auditverify$ for the relation $\relation$ will reject with overwhelming probability. 
   \item \textbf{Privacy}. The auditor verifier $V$ running the verification algorithm $\auditverify$ for a relation $\relation$ cannot learn anything more about the ledger $\ledger$ beyond $\ledger \in \lang_\relation$.
\end{itemize}

\end{definition}

It is straightforward to see that the auditing scheme satisfies the correctness, soundness and privacy properties because the underlying NIZK proving system is complete, sound, and zero-knowledge.

\textbf{ZKP Auditing Example. } 
We now demonstrate ZKP-based auditing using the syntax above for asset liquidity auditing. Asset liquidity position can be used by investors and stakeholders to assess a company's financial health. Regular liquidity audits provide transparency and build confidence among investors regarding the company's ability to meet its financial commitments. Asset liquidity can be assessed by calculating the fractional ratios among different assets without revealing the balance of the assets using ZKP-based audit. The participant proves that the ratio of its investment in an asset over its total assets value is lower than or equal to some threshold $f = f_n/f_d$. That is for a participant $i$, its liquidity for an asset $a^*$ can be calculated as 
$$
\frac{\sum_{t \in \txlist} v_{t,a^*,i}}{\sum_{a\in \assetlist}\sum_{t \in \txlist}v_{t,a,i}} \leq f.
$$
Let $s_1= \sum_{t \in \txlist} v_{t,a^*,i}$ and $s_2 =\sum_{a\in \assetlist}\sum_{t \in \txlist}v_{t,a,i} $. Therefore, we need to show $0 \leq s_2f_n-s_1f_d$. It is a simple linear relation for the RHS and a positive range proof for the inequality, both of which we have used to construct the ZKP presented in this work. Since the commitment scheme is additively homomorphic, we can set $s_1'= \sum_{t \in \txlist} \commitment_{t,a^*,i}$ and $s_2' =\sum_{a\in \assetlist}\sum_{t \in \txlist} \commitment_{t,a,i}$. 
Specifically, the relation $\relation_{liquid}$ can be specified as 
\begin{align*}
\relation_{liquid} : \left\{
    \begin{array}{c}
    \bigl((\commitkey, s_1',s_2'), (s_1,r_1,s_2,r_2)\bigr) : 
\\
s_1' = \commit_{\commitkey}(s_1, r_1) \wedge s_2' = \commit_{\commitkey}(s_2, r_2)\\ 
\wedge\ s_2f_n-s_1f_d \in [0,N)\\
    \end{array}
\right\}
\end{align*}
where $N$ is the precision, and for notation simplicity we consider the case that the opening is exact and the asset is not compact. 

Then, by running $\auditprove$ and $\auditverify$ using $\relation_{liquid}$ for the ledger state $\ledger$, we achieve auditing for liquidity in zero-knowledge. A similar approach can be used to audit for asset balance, inter-transaction rate, and other properties in zero-knowledge.


\section{Relations enforced by the NIZK Proof and Formal Construction}
\label{app:relation}

The description for the relations enforced by the NIZK proof used in $\name$ is as follows:

\begin{enumerate}
    \item $\relation_{PoB}$ by \hyperref[proto:pi_b]{$\pi^{B}$}: Given a list of commitments, the opening to their summed message is zero.
    \item $\relation_{PoC}$ by \hyperref[proto:pi_c]{$\pi^{C}$}: Given a commitment, the random value $r$ used in the commitment is  short and the message is of the form $v,\sqrt q \cdot v, v$.
    \item $\relation_{PoE}$ by \hyperref[proto:pi_e]{$\pi^{Eq}$}: Given two commitments and a public key, when the corresponding secret key is used to decrypt the commitment message difference, the decrypted message norm with error term is short. Note that if the error term is also short, then value/message equivalence will hold by Lemma \ref{lemma:comp_equil}. 
    \item $\relation_{PoE2}$ by \cref{proto:ZKPoO_bdlop_linear}: Given two commitments, their messages are the same.
    \item $\relation_{PoA/PoA'}$ by \hyperref[proto:pi_a]{$\pi^{A}$}/\hyperref[proto:pi_a2]{$\pi^{A'}$}: Given a commitment, the message is a 64-bit integer. For the $PoA'$, it is asserted that given a commitment, all the assets' coefficients are 64-bit (positive) integers.
    \item $\relation_{PoKW}$ by \cref{app:pokw_section}:  Given a public key and a commitment key matrix $\mathbf{A}$, the prover knows a short secret and error term that constitute the public key with respect to $\mathbf{A^T}$.
\end{enumerate}

The formal description for the relations enforced by the NIZK proof used in $\name$ is as follows:

\begin{align*}  
\relation_{PoB} : \left\{
\begin{array}{c}
\bigl( (t,a,\commitkey, \{ \com_{t,a,i} \}_{i \in \pklist}), ( \bar c, \mathbf{r^*},  \mathbf{m^*})\bigr) : \\
\quad \com_{\{0,3\},t,a}^{sum} = \sum_{i \in \pklist}\com_{\{0,3\},t,a,i}\\
\wedge \ \com_{\{0,3\},t,a}^{sum} = \commit_{\commitkey}(\mathbf{m^*}, \mathbf{r^*})_{\{0,3\}}
\\\wedge \  \norm{\bar c \mathbf{r^*}} \leq 2\mathfrak{s}_{PoB}\sqrt{2(\kappa+\lambda+3)d} 
\\ \wedge \ (\cdot,m_3^*) =  \mathbf{m^*} \wedge \ m_3^* =0 
\end{array}
\right\}
\end{align*} 

\begin{align*}
\relation_{PoC} : \left\{
    \begin{array}{c}
    \bigl( (\commitkey, \widehat \commitkey, \commitment, \widehat \commitment), (\bar c,\mathbf{r^*},\mathbf{s^*},\mathbf{m^*})\bigr)  : 
\\
\widehat \commitment = \commit_{\widehat \commitkey}((\mathbf{r}^*,\mathbf{m}^*),\mathbf{s}^*)\\
\wedge \ \norm{\bar c \mathbf{r^*}} \leq 2\mathfrak{s}_{1,PoC}\sqrt{2(\kappa+\lambda+3)d}\\ 
\wedge \ \norm{\bar c \mathbf{s^*}} \leq 2\mathfrak{s}_{2,PoC}\sqrt{2(\kappa+\lambda+3)d} 
\\
\wedge \ \norm{\mathbf{r}^*}_{\infty} \leq 2\sqrt{2k}\mathfrak{s}_{3,PoC}\ \wedge \ \commitment = \commit_{\commitkey}(\mathbf{m}^*, \mathbf{r}^*)\\
\wedge \ (v^*,\sqrt{q}v^*,v^*)= \mathbf{m}^*
    \end{array}
\right\}
\end{align*}

\begin{align*}
\relation_{PoE} : \left\{
    \begin{array}{c}
    \bigl((\commitkey, \widehat \commitkey, \commitment, \commitment', \widehat \commitment, \pk), (\bar c,\mathbf{m^*},\mathbf{s^*})\bigr) : 
\\
\widehat \commitment = \commit_{\widehat \commitkey}((\mathbf{m}^*,0), \mathbf{s}^*)\\ 
\wedge \ (\mathbf{s_1}||\mathbf{e_1}||\mathbf{s_2}||\mathbf{e_2})= \mathbf{m^*}\ \wedge (\pk_1,\pk_2) =\pk\\
\wedge \ (\mathbf{A},\cdots) = \commitkey \wedge \com_{\text{diff}} = \com - \com'\\
\wedge \  \tilde{u}_{t,a,i}=\begin{bmatrix}\com_{\text{diff},1}\\\com_{\text{diff},2}\end{bmatrix}\ \wedge \      \tilde{C}_{t,a,i}=\begin{bmatrix}\com_{\text{diff},0}^T,0,0,0\\0,0,\com_{\text{diff},0}^T,0\end{bmatrix}\\
\wedge\ \norm{\bar c \mathbf{m^*}} \leq 2\mathfrak{s}_{1,PoE}\sqrt{4(2\kappa+\lambda+3)d}\\ 
\wedge\ \norm{\bar c \mathbf{s^*}} \leq 2\mathfrak{s}_{2,PoE}\sqrt{2(\kappa+\lambda+3)d}\\
\wedge\ \forall i \in \{1,2\}, \pk_i = \mathbf{A}^T\mathbf{s}_i+\mathbf{e}_i\\
\wedge \ \norm{\tilde{C} \mathbf{m} -\tilde{u}}_{\infty} \leq 2\sqrt{2k}\mathfrak{s}_{3,PoE} \\
    \end{array}
\right\}
\end{align*}

For two commitment $\com$ and $\com'$, their values are the same ($v=v'$), if the commitments simultaneously satisfy both $PoC$ and $PoE$ and that $\pk$ satisfies $PoKW$ which allows the use of \cref{lemma:comp_equil} to equate the commitment values.
\begin{align*}  
\relation_{PoA} : \left\{
\begin{array}{c}
\bigl( (\commitkey, \com), ( \bar c, \mathbf{r^*},  \mathbf{m^*})\bigr) : \\
\quad \com^\prime = \commit_{\commitkey}(\mathbf{m^*}, \mathbf{r^*}) \\
\wedge\ \norm{\bar c \mathbf{r^*}} \leq 2\mathfrak{s}_{PoA}\sqrt{2(\kappa+\lambda+3)d} \\
\wedge\ ( m_1^*,\cdots) =  \mathbf{m^*} \wedge \ m_1^* \in [0, \rangebound)
\end{array}
\right\}
\end{align*} 

\begin{align*}  
\relation_{PoA'} : \left\{
\begin{array}{c}
\bigl( (\commitkey, \com), ( \bar c, \mathbf{r^*},  \mathbf{m^*})\bigr) : \\
\quad \com^\prime = \commit_{\commitkey}(\mathbf{m^*}, \mathbf{r^*}) \\
\wedge\ (\cdots, m_3^*) = \mathbf{m}^* \wedge \ (m_{3,1}^*, \cdots,m_{3,d}^*) = m_0^*\\
\wedge \ \forall a \in \assetlist, m_{3,a}^* \in [0, \rangebound)
\end{array}
\right\}
\end{align*} 

\begin{align*}
\relation_{PoKW} : \left\{
    \begin{array}{c}
    \bigl((\commitkey, \widehat \commitkey,\widehat \commitment, \pk), (\bar c,\mathbf{m^*},\mathbf{s^*})\bigr) : 
\\
\widehat \commitment = \commit_{\widehat \commitkey}((\mathbf{m}^*,0), \mathbf{s}^*)\\ 
\wedge \ (\mathbf{s_1}||\mathbf{e_1}||\mathbf{s_2}||\mathbf{e_2})= \mathbf{m^*}\ \wedge (\pk_1,\pk_2) =\pk\\
\wedge \ (\mathbf{A},\cdots) = \commitkey\\
\wedge\ \norm{\bar c \mathbf{m^*}} \leq 2\mathfrak{s}_{1,PoKW}\sqrt{4(2\kappa+\lambda+3)d}\\ 
\wedge\ \norm{\bar c \mathbf{s^*}} \leq 2\mathfrak{s}_{2,PoKW}\sqrt{2(\kappa+\lambda+3)d}\\
\wedge\ \forall i \in \{1,2\}, \pk_i = \mathbf{A}^T\mathbf{s}_i+\mathbf{e}_i\\
\wedge\ \norm{(\mathbf{s_1}||\mathbf{e_1}||\mathbf{s_2}||\mathbf{e_2})}_{\infty} \leq 2\sqrt{2k}\mathfrak{s}_{3,PoKW}
    \end{array}
\right\}
\end{align*}

\begin{align*}
\relation_{PoE2} : \left\{
    \begin{array}{c}
    \bigl((\commitkey, \commitment, \commitment'), (\bar c,\mathbf{m^*},\mathbf{r^*}, \bar c',\mathbf{m'^*},\mathbf{r'^*})\bigr) : 
\\
\commitment = \commit_{\commitkey}(\mathbf{m}^*, \mathbf{r}^*) \wedge \commitment' = \commit_{\commitkey}(\mathbf{m'}^*, \mathbf{r'}^*)\\ 
\wedge\ \mathbf{m=m'}\\
    \end{array}
\right\}
\end{align*}

The construction of $\name$ as an instance of ETL (\cref{def:etl}) is given in \cref{fig:algorithm_qpadl} and setup algorithms are given in \cref{fig:algorithm_qpadl_setup}. We give the correctness definition for ETL as follow:

\begin{definition}[Correctness]
\label{def:etl_correctness}
The $\etl$ transaction scheme is correct if for any system parameter $(\publicparam, \ledger) \gets \setup(\secparam)$, any valid value list $\valuelist \in \valueset$, for any secret key list $\sklist$ where the keys are generated using $\keygen(\cdot)$ and $\forall i, v_i \in\valuelist \wedge v_i < 0 => \sk_i \in \sklist$, we have  
\begin{align*}
    \Pr \left[
    \begin{array}{cc}
         & \verifytx(\tx, \ledger) = 1 : \\
         & \tx \gets \createtx(\valuelist, \sklist, \ledger)
    \end{array}
    \right] \geq  1- \negl
\end{align*}
\end{definition}

\begin{figure}[h!]
    \noindent
    \noindent\makebox[\linewidth][l]{%
    \begin{adjustbox}{minipage=\linewidth,scale=0.85}
    \raggedright
    \begin{pcvstack}[boxed,noindent,aboveskip=0pt]
    \noindent
        
    \begin{pchstack}
            \begin{pcvstack}
         
                \procedure[bodylinesep=0.5\pclnspace]{
               $\createtx(\valuelist, \sklist, \ledger):$
               }  {
                \text{for each } v_{\widehat t,a,i} \in \valuelist:\\
                \tab \mathbf{r}_{\widehat t, a,i},\mathbf{r}_{\widehat t,a,i}', \mathbf{s}_{\widehat t,a,i},\mathbf{s}_{\widehat t,a,i}',\mathbf{s}_{\widehat t,a,i}''  \sample (\chi^{(k+\lambda +3)d})^5 \\
                \tab \mathbf{m}_{\widehat t, a,i} := (v_{\widehat t, a,i}, \sqrt{q}v_{\widehat t, a,i},v_{\widehat t, a,i})\\ \tab \com_{\widehat t, a,i} \gets \commit(\mathbf{m}_{\widehat t,a,i},\mathbf{r}_{\widehat t,a,i})
                ;\ \widehat \commitment_{\widehat t,a,i}^{PoC} := \commit_{\commitkey_{PoC}}((\mathbf{r}_{\widehat t,a,i},\mathbf{m}_{\widehat t,a,i}),\mathbf{s}_{\widehat t,a,i})\\
                \tab \pi^{PoC}_{\widehat t,a,i} \gets  \nizkprove_{{\relation_{PoC}}}(\crs, (\commitkey_{i},\widehat \commitkey^{PoC}, \commitment_{\widehat t,a,i}, \widehat \commitment_{\widehat t,a,i}^{PoC}), \\
                \tab \tab \tab \tab \tab \tab \tab \tab\tab\tab\tab(1,\mathbf{r}_{\widehat t,a,i}, \mathbf{s}_{\widehat t,a,i},\mathbf{m}_{\widehat t,a,i}) ) \\
                \tab \text{if }\ v_{\widehat t,a,i} \geq 0: \mathbf{m}_{\widehat t,a,i}^\prime := \mathbf{m}_{\widehat t,a,i} \codecomment{Original message} \\
                \tab \text{else : }\ \mathbf{m}_{\widehat t,a,i}^\prime := \sum_{t \in \txlist \cup \widehat t} \mathbf{m}_{t,a,i} \codecomment{Sum of asset + pending spends}\\
                \tab \com'_{\widehat t,a,i} \gets \commit(\mathbf{m}_{\widehat t,a,i}^\prime, \mathbf{r}_{\widehat t,a,i}^\prime)
               ;\ \widehat \commitment_{\widehat t,a,i}^{\prime, PoC} := \commit_{\commitkey_{PoC}}((\mathbf{r}_{\widehat t,a,i}^\prime,\mathbf{m}_{\widehat t,a,i}^\prime),\mathbf{s}_{\widehat t,a,i}^\prime)\\
                \tab \pi^{\prime, PoC}_{\widehat t,a,i} \gets  \nizkprove_{{\relation_{PoC}}}(\crs, (\commitkey_{i},\widehat \commitkey^{PoC}, \commitment_{\widehat t,a,i}^\prime, \widehat \commitment_{\widehat t,a,i}^{\prime PoC}), \\
                \tab \tab \tab \tab \tab \tab \tab \tab\tab\tab\tab(1,\mathbf{r}_{\widehat t,a,i}^\prime, \mathbf{s}_{\widehat t,a,i}^\prime,\mathbf{m}_{\widehat t,a,i}^\prime) ) \\
                \tab \pi^{PoA}_{\widehat t,a,i} \gets \nizkprove_{{\relation_{PoA}}}(\crs, (\commitkey_i^{PoA}, \commitment_{\widehat t,a,i}'), (1,\mathbf{r}_{\widehat t, a,i}^{\prime},\mathbf{m}_{\widehat t,a,i}^\prime)) \\
                \tab \widehat \commitment_{\widehat t,a,i}^{PoE} := \commit_{\commitkey_{PoE}}((\sk_{i},0),\mathbf{s}_{\widehat t,a,i}^{''})\\
                \tab (\pi^{PoE}_{\widehat t,a,i},\pi^{PoE2}_{\widehat t,a,i}) \gets  \nizkprove_{(\relation_{PoE}\vee \relation_{PoE2})} \\
                \tab \tab \tab \tab \tab \tab \tab \tab \bigl((\crs, (\commitkey_{i}, \widehat \commitkey^{PoE}, \sum_{t^* \in \txlist \cup \widehat t}^{} \commitment_{t^*,a,i}, \commitment_{\widehat t,a,i}^\prime,\widehat \commitment_{\widehat t,a,i}^{PoE}),
                \\ \tab \tab \tab \tab \tab \tab \tab \tab (1,(\sk_i,0), \mathbf{s}_{\widehat t,a,i}^{''})) \vee (\crs, (\commitkey_i, \com_{\widehat t,a,i} \\
                \tab \tab \tab \tab \tab \tab \tab \tab  , \com'_{\widehat t,a,i}), (1,\mathbf{m}_{\widehat t,a,i},\mathbf{r}_{\widehat t,a,i},1,\mathbf{m}'_{\widehat t,a,i},\mathbf{r}'_{\widehat t,a,i}))\bigr)\\
                \text{for each } a \in \assetlist:\\
                \tab \pi^{PoB}_{\widehat t,a} \gets \nizkprove_{{\relation_{PoB}}}(\crs, (\widehat t, a , \commitkey, \{\commitment_{\widehat t,a,i}\}_{i \in \pklist}), (1,\sum_{i \in \pklist} \mathbf{r}_{\widehat t,a,i},0)) \\
                \commitment := \{\commitment_{\widehat t,a,i}, \commitment_{\widehat t,a,i}^\prime, \widehat \commitment_{\widehat t,a,i}^{PoC}, \widehat \commitment_{\widehat t,a,i}^{\prime PoC}, \widehat \commitment_{\widehat t,a,i}^{PoE}\}_{a \in \assetlist, i \in \pklist}\\
                \pi := (\{\pi^{PoB}_{\widehat t, a}\}_{a \in \assetlist},\{\pi^{PoE}_{\widehat t,a,i}, \pi^{PoE2}_{\widehat t,a,i},\pi^{PoC}_{\widehat t,a,i}, \pi^{\prime PoC}_{\widehat t,a,i}, \pi^{PoA}_{\widehat t,a,i}\}_{a \in \assetlist, i\in \pklist}) \\
                \pcreturn \tx := (\commitment, \pi) \\
               }
            \end{pcvstack}
    \end{pchstack}

    \begin{pcvstack}
       \procedure[bodylinesep=0.5\pclnspace]{
       $\verifytx(\tx, \ledger):$
       }  {
        (\{\commitment_{\widehat t,a,i}, \commitment_{\widehat t,a,i}^\prime, \widehat \commitment_{\widehat t,a,i}^{PoC}, \widehat \commitment_{\widehat t,a,i}^{\prime PoC}, \widehat \commitment_{\widehat t,a,i}^{PoE}\}_{a \in \assetlist, i \in \pklist},\pi) := \tx \\
       (\{\pi^{PoB}_{\widehat t, a}\}_{a \in \assetlist},\{\pi^{PoE}_{\widehat t,a,i}, \pi^{PoE2}_{\widehat t,a,i},\pi^{PoC}_{\widehat t,a,i}, \pi^{\prime PoC}_{\widehat t,a,i}, \pi^{PoA}_{\widehat t,a,i}\}_{a \in \assetlist, i\in \pklist}) := \pi \\
        b^{PoB} \gets \bigwedge_{a\in \assetlist} \nizkverify_{{\relation}_{PoB}}(\crs, (\widehat t, a , \commitkey, \{\commitment_{\widehat t,a,i}\}_{i \in \pklist}), \pi^{PoB}_{\widehat t, a}) \\
        b^{PoC} \gets \bigwedge_{a\in \assetlist, i \in \pklist} \nizkverify_{{\relation}_{PoC}}(\crs, (\commitkey_{i},\widehat \commitkey^{PoC}, \commitment_{\widehat t,a,i}, \widehat \commitment_{\widehat t,a,i}^{PoC}), \pi^{PoC}_{\widehat t,a,i}) \\
        b^{PoC'} \gets \text{Similar to above, but with the PoC' version. } \\ 
        b^{PoE \vee PoE2} \gets \bigwedge_{a\in \assetlist, i \in \pklist} \nizkverify_{{\relation}_{PoE \vee PoE2}}\bigl( (\crs, (\commitkey_{i},\widehat \commitkey^{PoE}, \sum_{t^* \in \txlist \cup \widehat t}^{} \commitment_{t^*,a,i}, \commitment_{\widehat t,a,i}^\prime, \\
        \quad\quad\quad\quad\quad\quad\quad\quad\quad\quad\quad\quad\quad\quad\quad \pk_i), \pi^{PoE}_{\widehat t,a,i}) \vee
        \\ \quad\quad\quad\quad\quad\quad\quad\quad\quad\quad\quad\quad\quad\quad\quad (\crs, (\commitkey_{i}, \commitment_{\widehat t,a,i}, \commitment_{\widehat t,a,i}^\prime), \pi_{\widehat t,a,i}^{PoE2} ) \bigr)  \\
        b^{PoA} \gets \bigwedge_{a\in \assetlist, i \in \pklist} \nizkverify_{{\relation}_{PoA}}(\crs, (\commitkey_{i}^{PoA}, \commitment_{\widehat t,a,i}^\prime), \pi_{\widehat t,a,i}^{PoA} ) \\
        \pcreturn b^{PoB} \wedge b^{PoC} \wedge b^{PoC'} \wedge b^{PoE \vee PoE2} \wedge b^{PoA}\\
       }
       
    \end{pcvstack}
    
    \end{pcvstack}
    \end{adjustbox}
    }

    \caption{$\name$ Construction}
    \label{fig:algorithm_qpadl}
\end{figure}

\begin{figure}[h!]
    \begin{adjustbox}{minipage=\linewidth,scale=0.85, center}
        
     \begin{pcvstack}[boxed]
               \procedure[bodylinesep=0.5\pclnspace]{
               $\setup(\secparam, \bank, \valuelist):$
               }{
                \mathbf{A} \sample \mathcal{R}_q^{\kappa \times \kappa+\lambda+3}; \mathbf{B} \sample \mathcal{R}_q^{\kappa+\lambda+3}  \\
                \mathbf{a}_{\text{bin}}, \mathbf{a}'_{\text{bin}}, \mathbf{a}_{g}  \sample  (\mathcal{R}_q^{\kappa + \lambda + 3})^3  \\
                 \text{for each participant bank i } \in \bank: \\
                 \tab (\sk_{1,i}, \sk_{2,i},\pk_{1,i}, \pk_{2,i}, \pi_{i}^{KW}) \gets \keygen(\mathbf{A})  \\
                 \tab \commitkey_i := (\mathbf{A}, \pk_{1,i}^T,\pk_{2,i}^T,\mathbf{B}^T) \\
                 \tab \commitkey_{i}^{PoA} := (\commitkey_i, \mathbf{a}^T_{\text{bin}}, \mathbf{a}'^T_{\text{bin}}, \mathbf{a}^T_{g}) \\
                 \mathbf{A}_1,\mathbf{A}_2 \sample \mathcal{R}_q^{\kappa \times (\kappa+\lambda+3)} \times \mathcal{R}_q^{\kappa \times (\kappa+\lambda+3)} \\
                 \mathbf{B}_1,\mathbf{B}'_{c},\mathbf{B}''_{c} \sample \mathcal{R}_q^{1 \times (\kappa+\lambda+3)} \times \mathcal{R}_q^{256/d \times (\kappa+\lambda+3)} \times \mathcal{R}_q^{1 \times (\kappa+\lambda+3)} \\
                 \widehat \commitkey^{PoC} :=   (\mathbf{A}_1,\mathbf{A}_2,\mathbf{B}_1,\mathbf{B}'_{c},\mathbf{B}''_{c})\\
                 \mathbf{A}_3,\mathbf{A}_4 \sample  \mathcal{R}_q^{\kappa \times 2(2\kappa+\lambda+3)} \times \mathcal{R}_q^{\kappa \times (\kappa+\lambda+3)} \\
                 \mathbf{B}'_{eq},\mathbf{B}''_{eq} \sample \mathcal{R}_q^{256/d\times(\kappa+\lambda+3)} \times \mathcal{R}_q^{1\times (\kappa+\lambda+3)} \\ 
                 \widehat \commitkey^{PoE} := (\mathbf{A}_3,\mathbf{A}_4,\mathbf{B}'_{eq},\mathbf{B}''_{eq})  \\
                 \widehat \commitkey^{KW} := \commitkey^{PoE} \\
                 \publicparam := (\{\commitkey_i\}_{i \in \bank}, \{\commitkey_i^{PoA}\}_{i \in \bank}, \widehat \commitkey^{PoC},  \widehat \commitkey^{PoE}, \widehat \commitkey^{KW})\\
                 \text{for each } v_{0,a,i} \text{ in }  \valuelist \\
                 \tab r_{0,a,i} \sample \chi^{(\kappa+\lambda+3)d}\\
                 \tab \ledger[0][a][i] := \commit_{\commitkey_i}(v_{0,a,i}, r_{0,a,i}) \\
                 \tab \ledger_{\pk}[i] := \pk_i \\
                 \pcreturn (\publicparam, \ledger)
               }

               \pcvspace

              \procedure[bodylinesep=0.5\pclnspace]{
               $\keygen(\mathbf{A}):$
               }{
                    \sk_1  := (\mathbf{s}_{1}, \mathbf{e}_{1}) \sample (\chi^{\kappa d}, \chi^{(\kappa+\lambda+3)d})\\
                    \sk_2 := (\mathbf{s}_{2}, \mathbf{e}_{2}) \sample (\chi^{\kappa d}, \chi^{(\kappa+\lambda+3)d})\\
                    \sk := (\sk_1,\sk_2) \\
                    \pk_{1}:=   \bfA^T\mathbf{s}_{1}+\mathbf{e}_{1}; \pk_{2} :=  \bfA^T\mathbf{s}_{2}+\mathbf{e}_{2} \\
                    \pk := (\pk_{1}, \pk_{2}) \\
                    \mathbf{s} \sample \chi^{(k+\lambda +3)d}\\
                    \widehat \commitment := \commit_{\commitkey_{KW}}((\sk,0),\mathbf{s})\\
                    \pi^{KW} \gets \nizkprove_{\relation_{{PoKW}}} (\crs, (\commitkey, \widehat \commitkey^{KW},\widehat \commitment, \pk), (1,\sk,\mathbf{s}))\\
                    \pcreturn (\sk, \pk, \pi^{KW}) 
               }

               \pcvspace
            \procedure[bodylinesep=0.5\pclnspace]{
               $\checkbalance(\sk_{i}, a, \ledger):$
               }  {
                \Tilde{\com} := \sum_{t \in \txlist} \ledger[t][a][i] \\
                v_{a,i} \gets \extract(\Tilde{\com}, \sk_{i}) \\
                \pcreturn v_{a,i}
               }

    \end{pcvstack}
    \end{adjustbox}

    \caption{$\name$ Setup and Check Balance, where $\extract$ is defined in \cref{proto:exctractability}}
    \label{fig:algorithm_qpadl_setup}
\end{figure}

\section{Additional Background}
\label{app:additional_background}
\subsection{Module SIS/LWE}

\begin{definition}[$\msis_{\kappa,m,\beta}$ \cite{LNPS21}]\label{def:msis} 
Given $\bfA \leftarrow \CalRq^{\kappa \times m}$, the homogeneous Module-SIS problem with parameters $\kappa, m > 0$ and $0 \leq \beta < q$ asks for $z \in \CalRq^m$ such that $\bfA z = 0$ over $\CalRq$ and $0 < \|z\| \leq \beta$. An algorithm $\adv$ is said to have advantage $\epsilon$ in solving $\msis_{\kappa,m,\beta}$ if 
\begin{align*}
\Pr \left[ 0 < \|z\| \leq \beta \land \bfA z = 0 \ \textnormal{over} \ \CalRq \mid \bfA \leftarrow \CalRq^{\kappa \times m},\ z \leftarrow \adv(\bfA) \right] \geq \epsilon.
\end{align*}
If the advantage $\epsilon$ is negligible, we say $\msis_{\kappa,m,\beta}$ is hard for $\adv$.
\end{definition}

\begin{definition}[$\mlwe_{m,\lambda,\chi}$ \cite{LNPS21}]\label{def:mlwe} 
The Module-LWE problem with parameters $\lambda, m > 0$ and a distribution $\chi$ over $\{-1, 0, 1\}$ asks one to distinguish between the following two distributions:
\begin{enumerate}
    \item $(\bfA, \bfA \mathbf{s} + \mathbf{e})$, where $\bfA \leftarrow \CalRq^{m \times \lambda}$, $\mathbf{s} \leftarrow \chi^{\lambda d}$, and $\mathbf{e} \leftarrow \chi^{m d}$;
    \item $(\bfA, \mathbf{u})$, where $\bfA \leftarrow \CalRq^{m \times \lambda}$ and $\mathbf{u} \leftarrow \CalRq^m$.
\end{enumerate}
An algorithm $\adv$ is said to have advantage $\epsilon$ in solving $\mlwe_{m,\lambda,\chi}$ if
\begin{align*}
     &\left| \Pr \left[ b = 1 \mid \bfA \leftarrow \CalRq^{m \times \lambda};\ \mathbf{s} \leftarrow \chi^{\lambda d};\ \mathbf{e} \leftarrow \chi^{m d};\ b \leftarrow \adv(\bfA, \bfA \mathbf{s} + \mathbf{e}) \right] \right. \\
    &\quad - \left. \Pr \left[ b = 1 \mid \bfA \leftarrow \CalRq^{m \times \lambda};\ \mathbf{u} \leftarrow \CalRq^m;\ b \leftarrow \adv(\bfA, \mathbf{u}) \right] \right| \geq \epsilon.
\end{align*}
If the advantage $\epsilon$ is negligible, we say $\mlwe_{m,\lambda,\chi}$ is hard for $\adv$.
\end{definition}

\begin{definition}[Dual $\mlwe_{\lambda,\kappa,\chi}$ \cite{Nguyen22}]\label{def:mlwe1} 
The (knapsack) Module-LWE problem with parameters $\kappa, \lambda > 0$ and a distribution $\chi$ over $\{-1, 0, 1\}$ asks one to distinguish between the following two distributions:
\begin{enumerate}
    \item $(\bfA, \bfA \mathbf{s})$, where $\bfA \leftarrow \CalRq^{\kappa \times (\kappa + \lambda)}$, $\mathbf{s} \leftarrow \chi^{(\kappa + \lambda) d}$;
    \item $(\bfA, \mathbf{u})$, where $\bfA \leftarrow \CalRq^{\kappa \times (\kappa + \lambda)}$ and $\mathbf{u} \leftarrow \CalRq^\kappa$.
\end{enumerate}
An algorithm $\adv$ is said to have advantage $\epsilon$ in solving the dual $\mlwe_{\lambda,\kappa,\chi}$ problem if
\begin{align*}
     &\left| \Pr \left[ b = 1 \mid \bfA \leftarrow \CalRq^{\kappa \times (\kappa + \lambda)};\ \mathbf{s} \leftarrow \chi^{(\kappa + \lambda) d};\ b \leftarrow \adv(\bfA, \bfA \mathbf{s}) \right] \right. \\
    &\quad - \left. \Pr \left[ b = 1 \mid \bfA \leftarrow \CalRq^{\kappa \times (\kappa + \lambda)};\ \mathbf{u} \leftarrow \CalRq^\kappa;\ b \leftarrow \adv(\bfA, \mathbf{u}) \right] \right| \geq \epsilon.
\end{align*}
If the advantage $\epsilon$ is negligible, dual $\mlwe_{\lambda,\kappa,\chi}$ is hard for $\adv$.
\end{definition}

\subsection{Ring and Norms}
\begin{definition}[\cite{LNPS21}]
We denote various set norms as follows:
\begin{itemize}
    \item For an element $w \in \mathbb{Z}_q$: $|w\|_\infty = |w \bmod^{\pm} q|$
    \item For a polynomial $w = w_0 + w_{1} X + \cdots + w_{d-1} X^{d-1} \in \mathcal{R}_q$:
    $$
    \|w\|_\infty = \max_j \|w_j\|_\infty, \quad \|w\|_p = \sqrt[p]{\|w_0\|_\infty^p + \cdots + \|w_{d-1}\|_\infty^p}.
    $$
    \item For $\mathbf{w} = (w_{1}, \ldots, w_k) \in \mathcal{R}_q^k$:
    $$
    \|\mathbf{w}\|_\infty = \max_j \|w_j\|_\infty, \quad \|\mathbf{w}\|_p = \sqrt[p]{\|w_{1}\|^p + \cdots + \|w_k\|^p}.
    $$
\end{itemize}
By default, we denote $\|\mathbf{w}\| := \|\mathbf{w}\|_2$.
\end{definition}

Let $\mathcal{M}_q := \{ p \in \mathbb{Z}_q[X] : \deg(p) < d/l \}$ be the $\mathbb{Z}_q$-module of polynomials of degree less than $d/l$. The Number Theoretic Transform ($\NTT$) of a polynomial $p \in \mathcal{R}_q$ is defined as follows:
\begin{align*}
&\NTT(p) :=
\begin{bmatrix} 
\hat{p}_0 \\  
\vdots \\ 
\hat{p}_{l-1} 
\end{bmatrix}
\in \mathcal{M}_q^l
\quad \text{where} \\
&\NTT(p)_j = \hat{p}_j = p \bmod \left(X^{\frac{d}{l}} - \zeta^{2j + 1}\right).
\end{align*}
Furthermore, we extend the definition of the NTT to vectors of polynomials $\mathbf{p} \in \mathcal{R}^k_q$, where the NTT operation is applied to each coefficient of $\mathbf{p}$, resulting in a vector in $\mathcal{M}_q^{kl}$.

We also define the inverse NTT operation. Namely, for a vector $\mathbf{v} \in \mathcal{M}_q^l$, $\NTT^{-1}(\mathbf{v})$ is the polynomial $p$ such that $\NTT(p) = \mathbf{v}$.

Let $\mathbf{v} = (v_0, \ldots, v_{l-1})$ and $\mathbf{w} = (w_0, \ldots, w_{l-1})$ be elements of $\mathcal{M}_q^l$. We define the component-wise product $\mathbf{v} \circ \mathbf{w}$ to be the vector $\mathbf{u} = (u_0, \ldots, u_{l-1}) \in \mathcal{M}_q^l$ such that
\begin{align*}
u_j = v_j w_j \bmod \left(X^{\frac{d}{l}} - \zeta^{2j + 1}\right)
\end{align*}
for $j \in \mathbb{Z}_{l}$.

By definition, we have the following property of the inverse NTT operation:
\begin{align*}
\NTT^{-1}(\mathbf{v}) \cdot \NTT^{-1}(\mathbf{w}) = \NTT^{-1}(\mathbf{v} \circ \mathbf{w}).
\end{align*}

\subsection{Approximate Range Proof}
We recall the necessary lemmas for approximate range proof here.
Let $\texttt{Bin}_k$ be the distribution $\sum_{i=1}^k(a_i - b_i)$, where $a_i, b_i \leftarrow \{0,1\}$, and its variance is $k/2$ for a positive integer parameter $k$.

\begin{lemma}[\cite{LNS21}]
    \label{lemma:infty_bound}
    For two vectors $\mathbf{w} \in \mathbb{Z}^m_q$ and $\mathbf{y} \in \mathbb{Z}^n_q$, by choosing $\mathbf{R} \leftarrow \texttt{Bin}_1^{n \times m}$, we have
    \begin{align*}
        \Pr_{\mathbf{R} \leftarrow \texttt{Bin}_1^{n \times m}}\left[\|\mathbf{R}\mathbf{w} + \mathbf{y} \|_\infty < \frac{1}{2}\|\mathbf{w}\|_\infty\right] \leq 2^{-n}
    \end{align*}
\end{lemma}

Furthermore, Gentry et al.~\cite{GHL22} propose an analogous result in terms of the $l_2$-norm and provide the following heuristic regarding the norm bound when setting concrete parameters:

\begin{lemma}[\cite{GHL22}]
    \label{lemma:2_bound}
    Under the heuristic substitution of $\texttt{Bin}_1$ with the normal distribution of variance $1/2$, for any $\mathbf{w} \in \mathbb{Z}^m$:
    \begin{enumerate}
        \item $\Pr_{\mathbf{R} \leftarrow \texttt{Bin}_1^{256 \times m}}\left[\|\mathbf{R}\mathbf{w}\|^2 < 30\|\mathbf{w}\|^2\right] \leq 2^{-256}$
        \item $\Pr_{\mathbf{R} \leftarrow \texttt{Bin}_1^{256 \times m}}\left[\|\mathbf{R}\mathbf{w}\|^2 > 337\|\mathbf{w}\|^2\right] \leq 2^{-128}$
    \end{enumerate}
   
\end{lemma}

\begin{lemma}[\cite{LNP22}]
    \label{lemma:2_bound_rwy}
    Fix m, P $\in \mathbb{N}$ and  a bound $b \leq P/41m$, and let $\vec w \in [\pm P/2]^m$ with $\norm{\vec w} \geq b$, and let $\vec y$ be an arbitrary vector in $[\pm P/2]^{m}$. Then
    $$\Pr_{R \gets \mathsf{Bin}_1^{256 \times m}}[\norm{R\vec w + \vec y \ \text{mod}\ P} < \frac{1}{2}b\sqrt{26}] < 2^{-128}$$
\end{lemma}

\subsection{Discrete Gaussian}
\begin{definition}
The \textit{discrete Gaussian distribution} on $\mathbb{Z}^\ell$ centered at $\mathbf{v} \in \mathbb{Z}^\ell$ with standard deviation $\mathfrak{s} > 0$ is given by
$$
D^{\ell}_{\mathbf{v},\mathfrak{s}}(\mathbf{z}) = \frac{e^{-\|\mathbf{z}-\mathbf{v}\|^2/2\mathfrak{s}^2}}{\sum_{\mathbf{x} \in \mathbb{Z}^\ell} e^{-\|\mathbf{x}\|^2/2\mathfrak{s}^2}}.
$$
When centered at $\mathbf{0} \in \mathbb{Z}^\ell$, we write $D^{\ell}_\mathfrak{s} = D^{\ell}_{0,\mathfrak{s}}$.
\end{definition}

\begin{lemma}[\cite{Ban93}, \cite{Lyu12}]
\label{lemma:probability_distribution_tail_bound}
Let $k, \ell, d > 1$. Then:
\begin{enumerate}
    \item For $z \leftarrow D_{\mathfrak{s}}$, $\Pr\left[|z| \leq k\mathfrak{s}\right] > 1 - 2e^{-k^2/2}$.
    \item For $\mathbf{z} \leftarrow D^{\ell d}_\mathfrak{s}$, $\Pr\left[\|\mathbf{z}\| \leq \mathfrak{s}\sqrt{2\ell d}\right] > 1 - 2^{-\log(e/2)\ell d/2} > 1 - 2^{-\ell d/8}$.
\end{enumerate}
\end{lemma}

\subsection{Challenge Space and Invertibility}

For two independent and uniform $c, c' \in \mathcal{C}$, we require that the difference $\overline{c} = c - c'$ is invertible in $\mathcal{R}_q$ with overwhelming probability. 
To show that our definition of $\mathcal{C}$ practically supports this requirement, it suffices to compute the min-entropy of a random $c \in \mathcal{C}$ modulo each NTT prime ideal $(x^{d/l} - \zeta^j)$ and show that the probability that $c$ hits a particular NTT component is smaller than the targeted soundness error. 

\begin{lemma}
\label{lemma:invertibility_from_NTT}
Let $q$ be a prime such that the polynomial ring $\mathbb{Z}_q[X] / \langle X^d + 1 \rangle$ splits into $l$ prime ideals, and let $\zeta$ be a $2l$-th primitive root of unity. An element $p \in \mathbb{Z}_q[X] / \langle X^d + 1 \rangle$ is invertible if and only if all of its NTT coefficients are nonzero.

\begin{proof}
Assume $p$ is zero in at least one of its NTT coefficients, say $\hat{p}_j = p \bmod (X^{\frac{d}{l}} - \zeta^{2j+1})$. If $p$ is also invertible, then $p$ has a multiplicative inverse $p^{-1}$ in the ring such that:
$$
p \cdot p^{-1} \equiv 1 \pmod{X^d + 1}
$$
We can rewrite this equation in the NTT domain:
\begin{align*}
\NTT(p) \circ \NTT(p^{-1}) &\equiv \NTT(1) \\
\mathbf{p} \circ \mathbf{p}^{-1} &\equiv [1, 1, \ldots, 1]
\end{align*}
Expanding the above operation on the $j$-th coefficient:
\begin{align*}
\hat{p}_j \cdot \hat{p}^{-1}_j &= 1 \pmod{X^{\frac{d}{l}} - \zeta^{2j+1}} \\
0 \cdot \hat{p}^{-1}_j &= 1 \pmod{X^{\frac{d}{l}} - \zeta^{2j+1}}
\end{align*}
Zero is not part of the multiplicative group modulo $(X^{\frac{d}{l}} - \zeta^{2j+1})$ and hence cannot have an inverse, which contradicts the hypothesis.
\end{proof}
\end{lemma}

\begin{definition}
\label{def:collision_entropy}
Given a probability distribution $D$ over a set $\mathcal{S}$, we define the \textit{collision entropy} or \textit{Rényi entropy} of order $2$ as the function:
$$
\mathsf{H}_2(D) = -\log\left(\sum_{x \in \mathcal{S}} D(x)^2\right)
$$
$\mathsf{H}_2$ measures the "spread" or "uniformity" of the distribution $D$. Collision entropy quantifies how likely it is for $2$ independently chosen samples from the distribution $D$ to be the same. A higher value of $\mathsf{H}_2$ indicates a more uniform distribution.
\end{definition}

\begin{lemma}[\cite{ALS20}]
\label{lemma:Y_max_bound_probability_challenge}
Let the random variable $Y$ over $\mathbb{Z}_q$ be defined as above. Then for all $x \in \mathbb{Z}_q$,
\begin{equation}
\Pr(Y = x) \leq M := \frac{1}{q} + \frac{1}{q} \sum_{j \in \mathbb{Z}_q^\times} \prod_{k=0}^{l-1} \left| p + (1 - p) \cos\left(\frac{2\pi j \zeta^k}{q}\right) \right|
\end{equation}
\end{lemma}

Here, $m$ is a bound on the maximum probability for $Y$ to assume a certain value for any polynomial coefficient $x_i$ of $c \bmod (X^{d/l} - \zeta^j)$. Therefore, the probability for $c$ to assume a certain value over an NTT component is at most $m^{d/l}$.

Assuming a fully-splitting setting, i.e., $X^d + 1 = \prod_{j=1}^d (x - \zeta^j) \bmod q$, $c \bmod (X - \zeta^j)$ can be written as:
\begin{gather*}
\sum_{i=0}^{d-1} c_i (\zeta^j)^i \bmod q \\ = c_0 + \zeta^j \left( c_1 + \zeta^j \left( c_2 + \ldots + \zeta^j \left( c_{d-2} + \zeta^j c_{d-1} \right) \right) \cdots \right) \bmod q.
\end{gather*}

The distribution of $c(\zeta^j)$ for $c \leftarrow \mathcal{C}$ is equivalent to the distribution of the random variable $Y = Y_0$ in the stochastic process $(Y_d, Y_{d-1}, \ldots, Y_0)$ where:
$$
\begin{array}{cc}
     Y_d = 0, &  Y_i = c_i + \zeta^j Y_{i+1}
\end{array}
$$
for $i < d$, and $c_i$ are independent and identically distributed with distribution $C_\mathrm{p}$.
This stochastic process generates a random walk over $\mathbb{Z}_q$ whose distribution converges to the uniform distribution in time $O\left(\frac{\log q}{\mathsf{H}_2(C_\mathrm{p})}\right)$.

In a more general factorization $x^d + 1 = \prod_{i=1}^l (x^{d/l} - \zeta^{j_i})$, the value $c \bmod (x^{d/l} - \zeta^{j_i})$ is not concentrated in any particular polynomial $c'_0 + c'_1 x + \cdots + c'_{d/l-1} x^{d/l-1}$. Proving this is a simple extension of the above case, since each of the $d/l$ coefficients $c'_i x^i$ of $c \bmod (x^{d/l} - \zeta^j)$ depends only on the coefficients $c_{i d/l + j}$ for $0 \leq i < l$ (i.e., the $d/l$ coefficients are mutually independent). Thus, the distribution of $c'_i$ follows the same stochastic process as above, except it consists of $l$ steps rather than $d$.

\begin{lemma}[\cite{ALS20}]
\label{lemma:correlation_R_q_NTT_probability}
Let $x \in \mathcal{R}_q$ be a random polynomial with independently and identically distributed coefficients. Then $\mathcal{R}_q/\langle X^{d/l} - \zeta^j \rangle \cong \mathcal{R}_q/\langle X^{d/l} - \zeta^i \rangle$, and $x \bmod (X^{d/l} - \zeta^j)$ and $x \bmod (X^{d/l} - \zeta^i)$ are identically distributed for all $i, j \in \mathbb{Z}_{2l}^{\times}$.
\end{lemma}

From Lemma~\ref{lemma:invertibility_from_NTT}, an element in $\mathcal{R}_q$ is invertible if and only if all of its NTT coefficients are non-zero. That is,
$$
\Pr[c \text{ is invertible}] = \Pr[c(\zeta^{j_1}) \neq 0 \wedge c(\zeta^{j_2}) \neq 0 \wedge \cdots \wedge c(\zeta^{j_{l-1}}) \neq 0].
$$

Consider the polynomial coefficients $x_i$ for the NTT component $c \bmod (X^{d/l} - \zeta^j)$ with $c \leftarrow \mathcal{C}$:
$$
\NTT(c)_j = \hat{c}_j = c \bmod (X^{d/l} - \zeta^j) = \sum_{i=0}^{\frac{d}{l}-1} x_i X^i,
$$
where all coefficients follow the same distribution over $\mathbb{Z}_q$. Let $Y$ denote the random variable over $\mathbb{Z}_q$ that follows this distribution. Attema et al.~\cite{ALS20} provide an upper bound on the maximum probability of $Y$.

\begin{lemma}
Let $c \leftarrow \mathcal{C}$ and $c' \leftarrow \mathcal{C}$. Then the difference of two distinct polynomials $\overline{c} = c - c'$ is invertible with probability $1 - \mathcal{O}(l q^{-d/l})$.
\begin{proof}
From Lemma~\ref{lemma:invertibility_from_NTT}, a polynomial is not invertible if and only if at least one of its NTT components is zero. We first upper bound the probability that, for $c \leftarrow \mathcal{C}$ and $c' \leftarrow \mathcal{C}$, one of the fixed NTT coefficients of $\overline{c}$ is zero. This probability is the same as the probability that one of the fixed NTT coefficients of $c$ and $c'$ are equal in $\mathbb{Z}_q[X]/(x^{d/l} - \zeta^j)$ for a fixed $j$. For some NTT components $c(\zeta^j)$ and $c'(\zeta^j)$ of $c$ and $c'$, we have:
$$
\Pr[c' \bmod (x^{d/l} - \zeta^j) = c \bmod (x^{d/l} - \zeta^j)] = \Pr[\hat{c}'_j = \hat{c}_j] \leq \mathcal{O}(q^{-d/l}).
$$
By the union bound, the probability that, for $c \leftarrow \mathcal{C}$ and $c' \leftarrow \mathcal{C}$, any of the NTT coefficients of $\overline{c}$ is zero is at most $\mathcal{O}(l q^{-d/l})$. This implies the desired result.
\end{proof}
\end{lemma}

Since $q$ will be much larger than $l$ in our work, $\overline{c}$ is invertible with overwhelming probability. As a conclusion, this gives us the following corollary:

\begin{corollary}
\label{corollary:c_invertible}
The parameters $p, q, d, l$ influence the bound on the maximum probability over a specific NTT component and the convergence of their distribution to uniformity. An immediate consequence is that the polynomial $\overline{c}$ is invertible in $\mathcal{R}_q$ with overwhelming probability, as long as the parameters are chosen appropriately.
\end{corollary}

For the parameters ($p,q,d,l$) given in \cref{sec:parameter}, we follow \cite{ALS20} method to show that $\overline{c}$ is independent of $j$ and uniform with bound of $\Pr(Y = x) \le M=1/q + \epsilon$:
\begin{equation}
M := \frac{1}{q} + \frac{1}{q} \sum_{j \in \mathbb{Z}_q^\times} \prod_{k=0}^{l-1} \left| p + (1 - p) \cos\left(\frac{2\pi j \zeta^k}{q}\right) \right|
\end{equation}
where we bound the probability that a random variable $Y$ over $\mathbb{Z}_q$ equals any $x \in \mathbb{Z}_q$. We found in a random sampling computation that $\epsilon \sim 2^{-287}$ for $q\sim2^{100}$, which strongly supports that $M\simeq 1/q$. 
Furthermore, note that for $q \sim 2^{64}$, $\epsilon\sim 2^{-150}$ is found for similar sample size.
\begin{table}[h]
\centering
\begin{tabular}{|c|c|c|c|c|c|c|}
\hline
$\log_2(\frac{N}{q})$ & $-83$ &$-81$& $-79$ &$-77$\\ \hline
$\log_2{\epsilon}$ & $-293$ & $-303$& $-290$ & $-293$ \\ \hline
\end{tabular}
\caption{Large random sampling of $\mathbb{Z}_q$ shows a strong evidence that $\log_2M\sim\log_2(1/q)\sim -100$}
\end{table}

\subsection{Rejection Sampling}

\begin{lemma}[Rejection Sampling~\cite{Lyu12}]
\label{lemma:rejsampling}
Let $\mathcal{R} = \mathbb{Z}[X]/(X^d+1)$. Let $V \subset \mathcal{R}^{l}$ 
be a set of polynomials with norm at most $T$,
and $\rho : V \rightarrow [0,1]$ be a probability distribution. Let $\mathfrak{s}=11T$ and $M=3$. Now, sample $\mathbf{v}  \leftarrow \rho$ and $\mathbf{y} \leftarrow D^{\ell d}_\mathfrak{s}$ and set $\mathbf{z}=\mathbf{y}+\mathbf{v}$, and
run $b  = \rej(\mathbf{z},\mathbf{v},\mathfrak{s})$. Then, the probability that $b =0$ (i.e., accept) is at least $(1-2^{-100})/M$,
and the distribution of $(\mathbf{v},\mathbf{z})|(b=0) \approx_{(2^{-100}/M)} \rho \otimes D^{\ell d}_\mathfrak{s}.$
\end{lemma}

\begin{figure}[!ht]
  \centering
  \begin{minipage}[t]{0.45\textwidth} 
    \begin{algorithm}[H]
    \caption{$\rej_0(\mathbf{z},\mathbf{v},\mathfrak{s})$~\cite{Lyu12}}\label{alg:rej1}
    \begin{algorithmic}
    \State $u\leftarrow[0,1)$
    \If{$u>\frac{1}{M}\exp(\frac{-2\dotprod{\mathbf{z}}{\mathbf{v}}+\|\mathbf{v}\|^2}{2\mathfrak{s}^2})$}
        \State return $1$
    \Else
        \State return $0$
    \EndIf
    \end{algorithmic}
    \end{algorithm}
  \end{minipage}
  \hspace{0.05\textwidth} 
 
    \caption{Rejection sampling algorithm. By default, we denote $\rej(\cdot)=\rej_0(\cdot)$.} 
\end{figure}

Here, the repetition rate $M$ is chosen to be an upper bound on:
\begin{align}
    \frac{D_\mathfrak{s}^{ld}(\mathbf{z})}{D_{\mathbf{v},\mathfrak{s}}^{ld}(\mathbf{z})} &= \exp\left(\frac{-2\dotprod{\mathbf{z}}{\mathbf{v}}+\|\mathbf{v}\|^2}{2\mathfrak{s}^2}\right) \leq \exp\left(\frac{24s\|\mathbf{v}\|+\|\mathbf{v}\|^2}{2\mathfrak{s}^2}\right) = M
\end{align}
For the inequality, we used the tail bound, which says that with probability at least $1-2^{-100}$, $|\dotprod{\mathbf{z}}{\mathbf{v}}|<12\|\mathbf{v}\|\mathfrak{s}$. By setting $\mathfrak{s}=11\|\mathbf{v}\|$, we have $M\approx 3$.

\subsection{Commit-and-Prove System}
\label{app:commit_and_prove}

An interactive proof system $\Pi = (\calP, \calV)$ for an NP relation $R \subseteq \{0,1\}^* \times \{0,1\}^*$ consists of a pair of efficient probabilistic algorithms, where $\calP$ denotes the prover and $\calV$ denotes the verifier. For any $(x, w) \in R$ with $x$ as a statement and $w$ as a witness, $\calP$ takes both $(x, w)$ as input, whereas $\calV$ takes only $x$ as input. The verifier $\calV$ outputs a bit $b \in \{0,1\}$ as the last message of the interaction. In the scope of our paper, we are interested in the relations with commit-and-prove system \cite{CLOS02}: It requires the prover to commit to witness $w$, then the prover shows the verifier certain statements about $w$: If the prover knows a witness $w$ such that $R(x, w) = 1$, then he should always be able to make the verifier output $b = 1$ (i.e., make the verifier accept - correctness). If the prover does not know such a witness, he should not be able to do so (soundness). We also require our protocols to be such that the prover does not reveal any additional information about his witness (zero-knowledge). 

Let $\langle \calP(x, w), \calV(x) \rangle$ denote the interactive protocol execution, and $\eta$ be the security parameter. In our manuscript, we typically consider public-coin protocols where the interaction between $\calP$ and $\calV$ consists of a constant number of moves (a maximum of $7$ rounds). By denoting commitment, challenge and response as $\com,\ch,\resp$ respectively, we now give formal definitions of completeness, $2$-special soundness, and honest-verifier zero-knowledge for $\Pi$:

\begin{definition}[Completeness]\label{def:completeness}
    The proof system $\Pi = (\calP, \calV)$ is complete if and only if there exists a negligible function $\texttt{neg}(\cdot)$ such that for all $(x, w) \in R$, we have:
    \begin{align}
        \Pr[\langle \calP(x, w), \calV(x) \rangle = 1] \geq 1 - \texttt{neg}(\eta).
    \end{align}
\end{definition}

\begin{definition}[2-Special Soundness]\label{def:soundness}
    The proof system $\Pi = (\calP, \calV)$ is 2-special sound if and only if there exists a probabilistic polynomial-time machine $\extractor$ (the extractor) such that for any $x$ in the language of $R$, and for a pair of accepting conversations $\{(\com, \ch^i, \resp^i)\}_{i=1,2}$ with $\ch^1 \neq \ch^2$, we have:
    \begin{align*}
        \extractor(x, \{(\com, \ch^i, \resp^i)\}_{i=1,2}) \rightarrow w \quad \text{such that} \quad (x, w) \in R.
    \end{align*}
\end{definition}
Note that this can be extended to $k$-special soundness if $\extractor$ takes as input $k$ transcripts, with pairwise distinct $\ch$ and a common $\com$. 

\begin{definition}[Honest-Verifier Zero-Knowledge]\label{def:hvzk}
    The proof system $\Pi = (\calP, \calV)$ is statistical (or computational) honest-verifier zero-knowledge if and only if there exists a probabilistic polynomial-time machine $\simulator$ (the simulator) and a negligible function $\texttt{neg}(\cdot)$ such that for all $(x, w) \in R$ and for any (efficient) unbounded distinguisher $\distinguisher$, we have the following probability is negligible:
    \[ 
    \begin{array}{l}
        \Big| \Pr\ [ b = 1 \mid (\com, \ch, \resp) \leftarrow \langle \calP(x, w), \calV(x) \rangle, \ b \leftarrow \distinguisher(\\
        \com, \ch, \resp)] - \Pr\ [b = 1 \mid (\com', \ch', \resp') \leftarrow \simulator(x, \eta),\\ 
        \ b \leftarrow \distinguisher(\com', \ch', \resp')] \Big|.
    \end{array}
    \]
    That is, the distribution of the simulated transcript is statistically (or computationally) indistinguishable from the real one.
\end{definition}

Meanwhile, it is known that 2-special soundness implies knowledge soundness in the classical setting \cite{ACK21}. Here, we give a formal definition of knowledge soundness in an interactive setting in Def.~\ref{def:pok}. 

\begin{definition}[Soundness - Proofs of Knowledge]\label{def:pok}
    Let $(\calP, \calV)$ be an interactive proof system for an NP relation $\relation$. An extractor $K$ is an algorithm that has resettable black-box access to a prover $\calP$. We call a proof system $(\calP, \calV)$ a proof of knowledge with knowledge error $\epsilon$ if and only if there exists a constant $d > 0$, a polynomially bounded function $p > 0$, and a probabilistic (or quantum) polynomial-time oracle machine $K$ such that for any interactive machine $\calP^*$ and every $x$, the following holds:
    \begin{align*}
        \Pr\left[\langle \calP^*(x), \calV(x) \rangle = 1\right] = \varepsilon > \epsilon \Longrightarrow  \\ \Pr[(x, w) \in R : w \leftarrow K^{\calP^*(x)}(x)] \geq \frac{1}{p(\eta)}\left(\varepsilon - \epsilon\right)^d
    \end{align*}
    where $\varepsilon$ is the probability that $\calP^*$ convinces $\calV$ a statement holds with respect to a relation $\relation$.
    The oracle machine $K$ is known as a knowledge extractor with running time polynomial in $\eta$ and the running time of $\calP^*$, and on input $x$ outputs a witness $w$ for $x$ such that $(x, w) \in R$.
\end{definition}

\subsection{Commitment Scheme}

\begin{definition}[Commitment Scheme]
A commitment scheme consists of the following set of algorithms:
	\begin{itemize}
	\item{ $\commitgen\textnormal{(}\secparam\textnormal{)}$:}  on input a security parameter $\secparam$, this algorithm outputs a commitment key $\commitkey$.
    \item{$\commit\textnormal{(}\commitkey, m,r \textnormal{)}$:} on input a commitment key $\commitkey$, a message $m$ and a random value $r$, this algorithm outputs a commitment $\commitment$.
    \item{$\open\textnormal{(}\commitkey, \commitment, m,r \textnormal{)}$:} on input a commitment key $\commitkey$, a commitment $\commitment$, a message $m$ and a random value $r$, this algorithm outputs a bit $b\in \{0,1\}$.
	\end{itemize}
	\label{definition:com}

A non-interactive scheme $(\commitgen, \commit, \open)$ is called a commitment scheme if it satisfies the hiding and binding properties defined in \cref{def:commitment_property}.
\end{definition}

\begin{definition}
\label{def:commitment_property}
We define the hiding and biding property of a commitment scheme below:

\textbf{Hiding. } Let the advantage $\hidingadv$ of breaking the binding property be the probability defined below. A scheme satisfies hiding if for all PPT adversary $\adv$, the following probability is negligible:
\begin{align*}
\abs{ \Pr\left[
\begin{array}{cc}
    \multirow{3}{*}{$\adv(\commitment)=b:$}  
    & \commitkey \gets \commitgen(\secparam); \\
    & (m_0,m_1)\gets \adv(\commitkey);\ b\sample \bin; \\
    & \commitment \gets \commit(\commitkey, m_b, r)
\end{array}
\right] - \frac{1}{2}}
\end{align*}

\textbf{Binding. } Let the advantage $\bindingadv$ of breaking the binding property be the probability defined below. A scheme satisfies binding if for all PPT adversary $\adv$, the following probability is negligible:

\[
\Pr\left[
\begin{gathered}
m_0 \neq m_1 \\
\open(\commitkey, \commitment,m_0,r_0)=\\
\open(\commitkey, \commitment,m_1,r_1)= 1
\end{gathered}
\ \middle| \ 
\begin{gathered}
\commitkey \gets (\commitgen(\secparam)); \\
(\commitment, m_0,r_0,m_1,r_1)\gets \adv(\commitkey)
\end{gathered}
\right]
\]
\label{definition:com_properties}
\end{definition}

\subsection{Non-interactive Zero-knowledge Proof}
\subsubsection{Zero-Knowledge Proof.} A non-interactive zero-knowledge proof (NIZK) allows a prover to prove to a verifier of some statement in zero-knowledge. Fiat-Shamir transform \cite{fiatshamir} can be used to transform an interactive Sigma ZKP protocol to a NIZK proof system for the specific relation. NIZK is defined as follows:

\begin{definition}[Non-interactive Zero-Knowledge Argument System] Let $\relation$ be an NP-relation and $\lang_{\relation}$ be the language defined by $\relation$. A non-interactive zero-knowledge argument system for $\lang_{\relation}$ consists of the following algorithms:
    \begin{description}
    \item[$\nizksetup_{{\relation}}(\secpar)$:] Takes as input a security parameter $\secpar$, and outputs a common reference string $\crs$.
    \item[$\nizkprove_{{\relation}}(\crs, \stmt, \wit)$:] Takes as input a common reference string $\crs$, a statement $\stmt$ and a witness $\wit$, and  outputs either a proof $\nizkproof$ or $\bot$.
    \item[$\nizkverify_{{\relation}}(\crs, \stmt, \nizkproof)$:] Takes as input a common reference string $\crs$,  a statement $\stmt$, and a proof $\nizkproof$,  and outputs either a $0$ or $1$.
    \end{description}
	\label{definition:nizk}

A NIZK proof system properties of completeness, soundness and zero-knowledge defined below.

\end{definition}
NIZK has the following properties.

\noindent\textbf{Completeness.} For all security parameters $\secpar$, all statements $\stmt \in \lang_{\relation}$, all witness $\wit$ with $\relation(\stmt, \wit)=1$, let $\crs \gets \nizksetup(\secpar)$ and $\nizkproof \gets \nizkprove(\crs,\stmt,\wit)$, it holds that $\nizkverify(\allowbreak\crs,\allowbreak\stmt,\nizkproof)=1$ with overwhelming probability.

\noindent\textbf{(Quantum) Proof of Knowledge Soundness \cite{LZ19}} NIZK has (Q)PoK soundness if there is a (quantum) polynomial time extractor $E$, a constant $c$, and a polynomial $p(\cdot)$, and negligible functions $\kappa(\cdot)$ and $negl(\cdot)$, such that for any (quantum) prover $P'$, for any statement $x$ satisfying $\Pr[\mathcal{V}(x,\pi)=1 : \pi \gets P'(x)]\geq \kappa(\secpar)$, we have the following:
\begin{align*}
\Pr[(x,E^{P(x)}(x))\in \relation] \geq \frac{1}{p(\secpar)}\cdot(\Pr[V(x, \pi)=1 : \pi \gets P'(x)]\\
-\kappa(\secpar))^c-negl(\secpar)
\end{align*}
We refer to \cite{LZ19} for the complete definition.

\noindent\textbf{Zero-Knowledge} The advantage $\sf{ZK}^{\adv}(\secpar)$ of an adversary $\adv$ in  breaking the zero-knowledge property is the probability defined below. The proof system is zero-knowledge if there exists simulator algorithms ($\zksimulator_1, \zksimulator_2$) such that for all PPT adversaries $\adv$, there exists a negligible function $\negl$ such that $\sf{ZK}^{\adv}(\secpar) \leq \negl$.
\begin{align*}
\Big|&\Pr\left[\begin{array}{cc}
\multirow{4}{*}{$\adv(\nizkproof^{\ast}) = 1~:$}  & \crs \gets \nizksetup(\secpar); \\
  & (\stmt, \wit) \gets \adv(\crs); \\
  &   \relation(\stmt, \wit) = 1;   \\
   &  \nizkproof^\ast \gets \nizkprove(\crs, \\
   &  \stmt,\wit)
\end{array}
\right]   - \\ 
&\Pr\left[\begin{array}{cc}
\multirow{4}{*}{$\adv(\nizkproof^{\ast}) = 1~:$} &  (\crs, {\sf st}) \gets \zksimulator_1(\secpar);  \\
 &  (\stmt, \wit) \gets \adv(\crs); \\
  &   \relation(\stmt, \wit) = 1; \\
  &    \nizkproof^\ast \gets \zksimulator_2(\crs,\stmt,{\sf st}) 
\end{array}\right]\Big| 
\end{align*}

\subsection{Generalized Forking Lemma.} 
We recall the general forking lemma \cite{forking_lemma} below.

\begin{lemma}
\label{lemma:forking_lemma}
Let $q \geq 1$ and a set H of size $h \geq 2$. Let $\mathsf{A}$ be a randomized algorithm that on input $x,h_1,\dots,h_q$ returns a pair, the first element of which is an integer in the range $0,\dots,q$ and second element of which we refer to as a side output. Let $\mathsf{IG}$ be a randomized algorithm that we call the input generator. The accepting probability of $\mathsf{A}$, denoted as $\mathsf{acc}$ is defined as the probability that $J\geq 1$ in the experiment 
$$
x \sample \mathsf{IG}; h_1,\dots,h_q \sample H;(J,\sigma) \sample \mathsf{A}(x,h_1,\dots,h_q)
$$
The forking algorithm $\mathsf{F_A}$ associated to $\mathsf{A}$ is the randomized algorithm that takes input $x$ proceeds as follows:

\procb{\text{Algorithm} $\mathsf{F_A}(x)$}{
        \text{Pick coins } \rho \text{ for } \mathsf{A} \text{ at random }\\
        h_1,\dots,h_q \sample H \\
        (I, \sigma) \gets \mathsf{A}(x,h_1,\dots,h_q; \rho)\\
        \text{ If }\ I=0 \text{ then return } (0, \bot, \bot)\\
        h_I', \dots, h_q' \sample H  \\
        (I',\sigma') \gets \mathsf{A}(x,h_1,\dots, h_{I-1}, h'_I,\dots, h'_q; \rho)\\
        If (I=I' \text{ and } h_I \neq h'_I)\text{ then return }(1,\sigma, \sigma')\\
        \text{ Else return }(0,\bot,\bot).
}

Let 
$$
\mathsf{frk} = \Pr[b=1: x\sample \mathsf{IG}; (b, \sigma,\sigma')\sample \mathsf{F_A}(x)].
$$
Then, $\mathsf{frk} \geq \mathsf{acc} \cdot (\frac{\mathsf{acc}}{q}-\frac{1}{h})$ or alternatively $\mathsf{acc} \leq \frac{q}{h}+ \sqrt{q \cdot \mathsf{frk}}$.

\end{lemma}

\subsection{BDLOP, ABDLOP Commitments with Linear and Quadratic Proofs}
\label{app:commitment_bdlop_abdlop}

We recall the BDLOP commitment scheme from~\cite{BDLOP16} to commit to a message vector $m \in \CalRq^n$ using randomness $\mathbf{r}\in\chi^{(\kappa+\lambda +n)d}$, where the module ranks $\kappa,\lambda$ and bound $\beta_{\text{msis}}$ are parametrized for MSIS and MLWE security. Meanwhile, it requires a matrix $\bfA \leftarrow \CalRq^{\kappa \times (\kappa+\lambda +n)}$
and vectors $\mathbf{b}_i \leftarrow \CalRq^{(\kappa+\lambda +n)}$ for every $i \in [n]$ as public parameters. The BDLOP commitment $\com$ is constructed as follows:

\begin{align*}
    \com_0&=\bfA\mathbf{r} \\ 
    \com_i&=\mathbf{b}_i^T\mathbf{r}+m_i ~\text{for}~i\in [n].
\end{align*}
For such a commitment scheme, the binding property relies on the hardness of $\msis_{\kappa,\kappa+\lambda+n,\beta_{\text{msis}}}$ problem, and the hiding property is established upon the hardness of dual $\mlwe_{\lambda,\kappa+n,\chi}$ problem.
We recall the weak opening and binding property of BDLOP below.

\begin{definition}[Weak Opening of BDLOP Commitment~\cite{ALS20}]
    A weak opening for the commitment $\com$ consists of a polynomial $\bar{c}\in\mathcal{R}_q$, randomness $\mathbf{r}^\ast$ over $\mathcal{R}_q$, and message $m^\ast\in\mathcal{R}_q$ such that 
    \begin{itemize}
        \item $\|\bar{c}\|_1\leq2\omega$ and $\bar{c}$ is invertible over $\mathcal{R}_q$. 
        \item $\|\bar{c}\mathbf{r}^\ast\| \leq 2 \beta = 2\mathfrak{s}\sqrt{2(\kappa+\lambda+n)d}$. 
        \item $\bfA\mathbf{r}^\ast=\com_0$.
        \item $\mathbf{b}_i^T\mathbf{r}^\ast+m_i^\ast =\com_i$ for $i\in [n]$.
    \end{itemize}
\end{definition}

\begin{lemma} [Binding Property of BDLOP \cite{ALS20}]
\label{lemma:binding_bdlop}
    The BDLOP commitment scheme is binding with respect to weak openings defined in \cref{def:opening_bdlop} if $\msis_{\kappa,\kappa+\lambda+n,8\omega\beta}$ is hard. 
\end{lemma}

\subsubsection{ABDLOP Commitment and Zero-Knowledge Proof of Opening}
While BDLOP allows commitment to any arbitrary ring element, commitment to message with short norm but large dimension is more efficient using Ajtai's commitment \cite{Ajtai96}. We recall the ABDLOP commitment scheme from \cite{LNP22, Nguyen22}. In brief, it combines BDLOP with the construction by Ajtai \cite{Ajtai96}, while using one short randomness vector $\mathbf{s}$ that binds the committed messages:

\begin{align*}
    \com_0&= \bfA_1\mathbf{r}+\bfA_2\mathbf{s}\\ 
    \com_i&=\mathbf{b}_i^T\mathbf{s}+m_i~\text{for}~i\in [n].
\end{align*}
where $\bfA_1 \leftarrow \CalRq^{\kappa \times (\mu_1)}, \bfA_2 \leftarrow \CalRq^{\kappa \times (\mu_2)}, \mathbf{b}_i \leftarrow \CalRq^{(\mu_2)}, \mathbf{r}\in\CalRq^{(\mu_1)}$ and $\mathbf{s}\in\chi^{(\mu_2)d}$. Let $\lambda := \mu_2-\kappa-n$. We further recall the interactive process of a zero-knowledge proof of opening for ABDLOP commitment $\com$ in Protocol \ref{proto:ZKPoO_abdlop} in the Appendix. Similarly to BDLOP, for ABDLOP commitment scheme, the binding property is associated with the hardness of $\msis_{\kappa,\mu_1+\mu_2,\beta_{\text{msis}}}$ problem.

\begin{definition}[Weak Opening of ABDLOP Commitment~\cite{LNP22}]
     A weak opening for the commitment $\com$ consists of a polynomial $\bar{c}\in\mathcal{R}_q$, two random elements $\mathbf{r}^\ast$ and $\mathbf{s}^\ast$ over $\mathcal{R}_q$, and message $m^\ast\in\mathcal{R}_q$ such that 
     \begin{itemize}
        \item $\|\bar{c}\|_1\leq2\omega$ and $\bar{c}$ is invertible over $\mathcal{R}_q$.
        \item $\|\bar{c}\mathbf{r}^\ast\|\leq2\mathfrak{s}_1\sqrt{2(\mu_1)d}$ and $\|\bar{c}\mathbf{s}^\ast\|\leq2\mathfrak{s}_2\sqrt{2(\mu_2)d}$.
        \item $\bfA_1\mathbf{r}^\ast+\bfA_2\mathbf{s}^\ast=\com_0$.
        \item $\mathbf{b}_i^T\mathbf{s}^\ast+m_i^\ast =\com_i$ for $i\in [n]$.
    \end{itemize}
\end{definition}

\begin{lemma} [Binding Property of ABDLOP \cite{LNP22}]
\label{lemma:binding_abdlop}
    The ABDLOP commitment scheme is binding with respect to weak openings defined in \cref{def:opening_abdlop} if $\msis_{\kappa,\mu_1+\mu_2,8\omega\sqrt{ (\mathfrak{s}_1\sqrt{2\mu_1d})^2+ (\mathfrak{s}_2\sqrt{2\mu_2d})^2 }}$\footnote{We use the norm bound of $\norm{ca} \leq \norm{c}_1 \cdot \norm{a}$ similar to that of \cite{ALS20} rather than the $\eta$ factor bound in \cite{LNP22}.} is hard. 
\end{lemma}


\begin{lemma} [Hiding Property of ABDLOP \cite{LNP22}]
\label{lemma:hiding_abdlop}
    The ABDLOP commitment scheme is hiding if dual $\mlwe_{\lambda,\kappa+n,\chi}$ is hard. 
\end{lemma}

We  further recall proof of opening for BLDOP protocol and \colorbox{Gainsboro}{linear relation protocol (with relation set to equality, $\allowbreak m \allowbreak = \allowbreak m'$)} in \cref{proto:ZKPoO_bdlop_linear} from \cite{BDLOP16}.
In addition, we recall proof of opening protocol for ABDLOP commitment given in \cref{proto:ZKPoO_abdlop}. 

\begin{theorem}[\cite{BDLOP16,ALS20}]
   \label{coro:security_bdlop_linear}
   Protocol~\ref{proto:ZKPoO_bdlop_linear} (and \colorbox{Gainsboro}{Protocol~\ref{proto:ZKPoO_bdlop_linear}} ) is complete, 2-special sound, and honest-verifier zero-knowledge.
\end{theorem}

\begin{theorem}[\cite{LNP22}]
    \label{coro:security_abdlop}
    Protocol~\ref{proto:ZKPoO_abdlop} is complete, 2-special sound, and honest-verifier zero-knowledge.
\end{theorem}

\begin{figure}[!ht]
     \begin{adjustbox}{minipage=0.5\textwidth,scale=0.85}
    \pcb{
    \underline{\textbf{Prover}} \<\< \underline{\textbf{Verifier}}
    \\\text{Input:}
    \\\bfA,\{\mathbf{b}_i\}_{i\in[n]}\<\<\bfA,\{\mathbf{b}_i\}_{i\in[n]}
    \\\mathbf{r},\com, \colorbox{Gainsboro}{$\mathbf{r'}, \com'$}\<\<\com, \colorbox{Gainsboro}{$\com'$}
    \\[0.1\baselineskip][\hline] \\[-0.5\baselineskip]
    \mathbf{y}\xleftarrow[]{\$} \mathcal{D}_{\mathfrak{s}}^{(\mathrm{k}+\lambda+n)d}\\
    \colorbox{Gainsboro}{$\mathbf{y'}\xleftarrow[]{\$} \mathcal{D}_{\mathfrak{s}}^{(\mathrm{k}+\lambda+n)d}$}\\
    \mathfrak{s}=11\omega\norm{\mathbf{r}}\\ 
    \mathbf{w} = \bfA \mathbf{y}; \colorbox{Gainsboro}{$\mathbf{w'} = \bfA \mathbf{y'}$}\\ 
    \colorbox{Gainsboro}{$\forall i \in [n], $}\\
    \colorbox{Gainsboro}{$\mathbf{u_i} = \mathbf{b_i\cdot y - b_i\cdot y'}$}\\
    \< \sendmessageright*[4em]{\mathbf{w}, \colorbox{Gainsboro}{$\mathbf{w',\{u_i\}}$}} \\
    \<\< c \xleftarrow[]{\$} \mathbf{\mathcal{C}} \\
    \< \sendmessageleft*[4em]{c} \\
    \mathbf{z}=\mathbf{y} + c \mathbf{r}; \colorbox{Gainsboro}{$\mathbf{z'}=\mathbf{y'} + c \mathbf{r'}$}\\
    \rej(\mathbf{z},c\mathbf{r},\mathfrak{s}); \colorbox{Gainsboro}{$ \rej(\mathbf{z'},c\mathbf{r'},\mathfrak{s})$}\\
    \< \sendmessageright*[4em]{\mathbf{z}, \colorbox{Gainsboro}{$\mathbf{z'}$}}\\
    \<\< \norm{\mathbf{z}} \stackrel{?}{\leq} \mathfrak{s}\sqrt{2(\kappa+\lambda+n)d}\\
    \<\< \colorbox{Gainsboro}{$\norm{\mathbf{z'}} \stackrel{?}{\leq} \mathfrak{s}\sqrt{2(\kappa+\lambda+n)d}$}\\
    \<\< \bfA \mathbf{z} \stackrel{?}{=} \mathbf{w} + c\com_0\\
    \<\< \colorbox{Gainsboro}{$\bfA \mathbf{z'} \stackrel{?}{=} \mathbf{w'} + c\com_0'$} \\
    \<\< \colorbox{Gainsboro}{$\forall i \in [n],$} \\
    \<\< \colorbox{Gainsboro}{$\mathbf{b_i} \mathbf{z} - \mathbf{b_i z'} \stackrel{?}{=} (\mathbf{\com_i}-\mathbf{\com_i'}) \cdot c + \mathbf{u_i}$}
    }
    \end{adjustbox}
    \caption{Zero-Knowledge Proof of Opening for BDLOP Commitment~\cite{BDLOP16} and the addition for \colorbox{Gainsboro}{BDLOP message equality} marked by \colorbox{Gainsboro}{light gray background.}}
    \label{proto:ZKPoO_bdlop_linear}
\end{figure}

\begin{figure}[!ht]
    \centering
    \begin{minipage}{0.5\textwidth}
    \pcb{
    \underline{\textbf{Prover}} \<\< \underline{\textbf{Verifier}}
    \\\text{Input:}
    \\\bfA_1, \bfA_2,\{\mathbf{b}_i\}_{i\in[n]}\<\<\bfA_1, \bfA_2,\{\mathbf{b}_i\}_{i\in[n]}
    \\\mathbf{r},\mathbf{s},\com\<\<\com
    \\[0.1\baselineskip][\hline] \\[-0.5\baselineskip]
    \mathbf{y}_1\xleftarrow[]{\$} \mathcal{D}_{\mathfrak{s}_1}^{(\mu_1)d},\mathfrak{s}_1=11\omega\norm{\mathbf{r}}\\
    \mathbf{y}_2\xleftarrow[]{\$} \mathcal{D}_{\mathfrak{s}_2}^{(\mu_2)d},\mathfrak{s}_2=11\omega\norm{\mathbf{s}}\\
    \mathbf{w} = \bfA_1 \mathbf{y}_1+\bfA_2 \mathbf{y}_2\\ 
    \< \sendmessageright*[4em]{\mathbf{w}} \\
    \<\< c \xleftarrow[]{\$} \mathbf{\mathcal{C}} \\
    \< \sendmessageleft*[4em]{c} \\
    \mathbf{z}_1=\mathbf{y}_1 + c \mathbf{r} \\
    \mathbf{z}_2=\mathbf{y}_2 + c \mathbf{s} \\
    \rej(\mathbf{z}_1,c\mathbf{r},\mathfrak{s}_1)\\
    \rej(\mathbf{z}_2,c\mathbf{s},\mathfrak{s}_2)\\
    \< \sendmessageright*[4em]{\mathbf{z}_1,\mathbf{z}_2}\\
    \<\< \norm{\mathbf{z}_1} \stackrel{?}{\leq} \mathfrak{s}_1\sqrt{2(\mu_1)d}\\
    \<\< \norm{\mathbf{z}_2} \stackrel{?}{\leq} \mathfrak{s}_2\sqrt{2(\mu_2)d}\\
    \<\< \bfA_1 \mathbf{z}_1 + \bfA_2 \mathbf{z}_2 \stackrel{?}{=} \mathbf{w} + c\com_0
    }
    \end{minipage}
    \caption{Zero-Knowledge Proof of Opening for ABDLOP Commitment~\cite{LNP22}}
    \label{proto:ZKPoO_abdlop}
\end{figure}

\subsubsection{Polynomial Evaluations with Vanishing Constant Coefficients \cite{LNP22}}
\label{sec:quad_proof}

We briefly recall the protocol for proving quadratic relations with vanishing polynomial evaluations from \cite{LNP22} needed in this paper.
Let the polynomial evaluation function be $$F_j(\mathbf{s}):= \mathbf{s}^T\mathbf{R_j}\mathbf{s} + \mathbf{r}_j^T \mathbf{s} + r_j \in R_q.$$ Let $\widehat {F_j(\mathbf{s})}$ be the constant coefficient of the evaluation, we have that $\widehat {F_j(\mathbf{s})}=0$. First, we combine multiple evaluations into one using random linear combination specified by the verifier $x = \sum_{j=1}^M d_j(F_j(\mathbf{s}))$ and send the masked version of it by sampling $g \gets \{a\in R_q: \widehat a=0\}$ and set $h:=g+x$. Note that we can apply automorphism $\sigma$ onto the secret $\mathbf{s}$ as needed by the relation. For simplicity, we write without considering the extra polynomial $g$ first for the relation, i.e. $h=x$. A quadratic relation can again be specified as $\mathbf{R}' = \sum_{j=1}^{M} d_j \mathbf{R}_j$, $\mathbf{r}'= \sum_{j=1}^{M} d_j\mathbf{r}_j$, $r'=\sum_{j=1}^{M} d_j r_j- h$.

Now, given the updated combined equation
$$\mathbf{s}^T\mathbf{R'}\mathbf{s} + \mathbf{r}'^T \mathbf{s} + r' \in R_q$$

Let the messages  be $\mathbf{s}:= (\mathbf{s_1}||\mathbf{m})$ and the commitments be 
\begin{align*}
    \com_0&= \bfA_1\mathbf{s_1}+\bfA_2\mathbf{s_2}\\ 
    \com_i'&=\mathbf{b}_i^T\mathbf{s_2}+m_i~\text{for}~i\in [n].
\end{align*} where $\com' = \{\com_i\}_{i \in [n]}$.

Let $\mathbf{z} := [\mathbf{z_1} \quad \mathbf{z_m}]^T$
be defined as the opening to the message. For example, in the ABDLOP, the masked opening is $\mathbf{z_m}=\mathbf{y_2}+c\mathbf{s_2}$, then the masked opening of the committed message under BDLOP portion will be $\mathbf{z_m}:= c\mathbf{com'} - \mathbf{Bz}_2= c\mathbf{m}-\mathbf{B}\mathbf{y_2}$. Let $\mathbf{z}$ be the masked opening for both the Atjai's part and BDLOP's part. 
Let define $\mathbf{y}:= [(\mathbf{y_1}) \quad (-\mathbf{By}_2)]^T$.
We define $g_1:= \mathbf{s}^T\mathbf{R}'\mathbf{y}+ \mathbf{y}^T\mathbf{R}'\mathbf{s}+\mathbf{r}'^T \mathbf{y}$ and $g_0 := \mathbf{y}^T\mathbf{R}'\mathbf{y}$.

If $\sigma(c)=c$, then we have
$$\mathbf{z}^T\mathbf{R'}\mathbf{z} + c\mathbf{r}'^T\mathbf{z}+c^2r'= c^2(\mathbf{s}^T\mathbf{R'}\mathbf{s} + \mathbf{r}'^T \mathbf{s} + r') + cg_1+g_0,$$
and if we have that $\mathbf{s}^T\mathbf{R'}\mathbf{s} + \mathbf{r}'^T \mathbf{s} + r'=0$, then we can just check that $\mathbf{z}^T\mathbf{R'}\mathbf{z} + c\mathbf{r}'^T\mathbf{z}+c^2r' =cg_1+g_0$. 

First, we commit to $g_1$ ahead of time as $\mathbf{com}_{g_1} := \mathbf{B}_{g_1}\mathbf{s_2}+g_1$. We send both $\mathbf{com}_{g_1}$ and $v:= g_0 + \mathbf{B}_{g_1}\mathbf{y_2}$ as well as the usual opening commitment before receiving a challenge. The prover then sends the masked opening as usual using the received challenge. Then, the verifier can compute $f=c\com_{g_1}-\mathbf{B}_{g_1}\mathbf{z_2}=cg_1-\mathbf{B}_{g_1}\mathbf{y}_2$. The verifier now check
$$\mathbf{z}^T\mathbf{R'}\mathbf{z} + c\mathbf{r}'^T\mathbf{z}+c^2r'-f =^? v$$
where $f+v=cg_1+g_0$. The above holds when the relation is true ($\mathbf{s}^T\mathbf{R'}\mathbf{s} + \mathbf{r}'^T \mathbf{s} + r'=0$). Now consider that $h=g+x$, we just need to add in $g$ as part of the secret committed ($\mathbf{s}:=(\mathbf{s}_1||\mathbf{m}||g)$) and update the relation accordingly which translate to roughly one extra linear equation to check.

If the polynomial is constructed correctly as checked previously, then we can now check the constant coefficient of $h$ to confirm that it vanishes as required.
We refer to (Fig.8 and Fig. 6, \cite{LNP22}) for the detailed protocol framework for proving polynomial evaluations with vanishing constant coefficients.

\section{Deferred Construction}
\label{app:construction}

\subsection{Proof of Balance Protocol}

The proof of balance protocol is given in Figure \ref{proto:pi_b}.

\begin{figure}[!ht]
    \centering
    \begin{adjustbox}{minipage=\linewidth,scale=0.85}
    \pcb{
    \underline{\textbf{Sender (Prover)}} \<\< \underline{\textbf{Participant i (Verifier)}}
    \\\text{Input:}
    \\\bfA,\bfB,\texttt{pk}_i\<\<\bfA,\bfB,\texttt{pk}_i
    \\\mathbf{r}_{t,a}=\sum_i\mathbf{r}_{t,a,i},\<\<\com_{t,a}= \sum_i\com_{t,a,i}
    \\\com_{t,a}= \sum_i\com_{t,a,i}
    \\[0.1\baselineskip][\hline] \\[-0.5\baselineskip]
    \mathbf{y}\xleftarrow[]{\$} \mathcal{D}_{\mathfrak{s}}^{(\mathrm{k}+\lambda+3)d},\\ 
    \mathfrak{s}= 11 \omega \| \mathbf{r}_{t,a} \|\\
    \mathbf{w} = \bfA \mathbf{y}\\ 
    u_{t,a}=\bfB^T\mathbf{y} \\
    \< \sendmessageright*[4em]{\mathbf{w}, u_{t,a}} \\
    \<\< c \xleftarrow[]{\$} \mathbf{\mathcal{C}} \\
    \< \sendmessageleft*[4em]{c} \\
    \mathbf{z}_{t,a}=\mathbf{y} + c\mathbf{r}_{t,a} \\
    \rej(\mathbf{z}_{t,a},c\mathbf{r}_{t,a} ,\mathfrak{s})\\
    \< \sendmessageright*[4em]{\mathbf{z}_{t,a}}\\
    \<\< \norm{\mathbf{z}_{t,a}} \stackrel{?}{\leq} \mathfrak{s}\sqrt{2(\kappa+\lambda+3)d}\\
    \<\< \bfA\mathbf{z}_{t,a} \stackrel{?}{=} \mathbf{w} + c \left(\sum_i \com_{0,t,a,i}\right)\\
    \<\< \bfB^T\mathbf{z}_{t,a}  \stackrel{?}{=} u_{t,a} + c\left(\sum_i \com_{3,t,a,i}\right)
    }
    \end{adjustbox}
    \caption{$\name$ Proof of Balance Protocol $\pi^B$}
    \label{proto:pi_b}
\end{figure}

\subsection{Proof of Equivalence Protocol}
\label{app:construction_poe}

We include an additional Ajtai commitment in our proof of equivalence protocol such that:
\begin{align} 
    \label{eq:poeq_ajtai_com}
    f_i= \bfA_3 \mathbf{m}_{i} + \bfA_4   \mathbf{s}_{t,a,i}
\end{align}
with $\mathbf{m}_{i} = \mathbf{s}_{1,i}||\mathbf{e}_{1,i}||\mathbf{s}_{2,i}||\mathbf{e}_{2,i}$ and 
\begin{itemize}
    \item $\bfA_3$, $\bfA_4$ are public parameters, and are sampled uniformly from $\CalRq^{\kappa \times 2(2\kappa+\lambda +3)}$ and $\CalRq^{\kappa \times (\kappa+\lambda +3)}$.
    \item $\mathbf{s}_{t,a,i}$ is a secret low-norm vector drawn from distribution $\chi^{(\kappa+\lambda +3)d}$ chosen by the sender (participant) for every commitment during Proof of Equivalence.
    \item $\bfB_{eq}'\in\CalRq^{256/d\times (\kappa+\lambda +3)}, \bfB_{eq}''\in\CalRq^{1\times (\kappa+\lambda +3)}$.
    \item $\mathbf{R}_i\leftarrow \texttt{Bin}_1^{256\times2d}$ is a random matrix sampled by the verifier as a challenge, and let $\mathbf{r}_j\in\CalRq^{256/d}$ denote the polynomial vector so that its coefficient vector is the $j$-th row of $\mathbf{R}_i$, where $j\in[256]$. 
    \item $e_j\in\CalRq^{256/d}$ is the polynomial vector such that its coefficient vector consists of all zeroes and one 1 in the $j$-th position.
    \item $\mathfrak{s}_1$, $\mathfrak{s}_2$, and $\mathfrak{s}_3$ denote the standard deviations, where $\mathfrak{s}_1=11\omega\sqrt{(4\kappa+2\lambda+6)d}$, $\mathfrak{s}_2=11\omega\sqrt{(\kappa+\lambda +3)d}$, and $\mathfrak{s}_3=11\sqrt{337}\cdot 2(\kappa+\lambda+3)d$.
\end{itemize}

First, we define
\begin{align*}
    \tilde{\bf A}=\begin{bmatrix}
        \bfA^T,\texttt{Id},0,0\\
        0,0,\bfA^T,\texttt{Id}
    \end{bmatrix};
    \hspace{5mm}
    \tilde{C}_{t,a,i}=\begin{bmatrix}
        \com_{\text{diff},0}^T,0,0,0\\
        0,0,\com_{\text{diff},0}^T,0
    \end{bmatrix};
\end{align*}

\begin{align*}
    u_i=\begin{bmatrix}
        \pk_{1,i}\\
        \pk_{2,i}
    \end{bmatrix};
    \hspace{5mm}
    \tilde{u}_{t,a,i}=\begin{bmatrix}
        \com_{\text{diff},1}\\
        \com_{\text{diff},2}
    \end{bmatrix}
\end{align*}
and give the protocol in \cref{proto:pi_e} and \cref{proto:pi_e_verification}

\begin{figure*}[!ht]
    \centering
    \begin{minipage}{\textwidth}
    \pcb[codesize=\scriptsize,colspace=-2mm]{
    \\
    \underline{\textbf{Participant $i$ (Prover)}} \<\< \underline{\textbf{Auditor (Verifier)}}
    \\\text{Input:}
    \\\bfA,\pk_i, \sk_{i} = (\mathbf{s}_{1,i},\mathbf{e}_{1,i},\mathbf{s}_{2,i},\mathbf{e}_{2,i}) \<\<\bfA,\texttt{pk}_i
    \\\bfA_3,\bfA_4,\bfB_{eq}',\bfB_{eq}'',\tilde{\bfA},\tilde{C}_{t,a,i},\tilde{u}_{t,a,i}\<\<\bfA_3,\bfA_4,\bfB_{eq}',\bfB_{eq}'',\tilde{\bfA},\tilde{C}_{t,a,i},\tilde{u}_{t,a,i}
    \\\com_{t,a,i}, \com_{t,a,i}', \mathbf{m}_{i} = \mathbf{s}_{1,i}||\mathbf{e}_{1,i}||\mathbf{s}_{2,i}||\mathbf{e}_{2,i}\<\<\com_{t,a,i}, \com_{t,a,i}'
    \\[0.1\baselineskip][\hline] \\[-0.5\baselineskip]
    \mathbf{s}_{t,a,i}\xleftarrow[]{\$}\chi^{(\kappa+\lambda +3)d}, g\xleftarrow[]{\$}\CalRq: \text{ constant of } g=0\\
    f_i= \bfA_3 \mathbf{m}_{i} + \bfA_4   \mathbf{s}_{t,a,i} \\
    \mathbf{y}_{1,i}\xleftarrow[]{\$} \mathcal{D}_{\mathfrak{s}_1}^{2(\kappa+\kappa+\lambda +3)d}, \mathbf{y}_{2,i}\xleftarrow[]{\$} \mathcal{D}_{\mathfrak{s}_2}^{(\kappa+\lambda +3)d}\\ 
    \mathbf{y}_{3,i}\xleftarrow[]{\$} \mathcal{D}_{\mathfrak{s}_3}^{256}\\
    \mathbf{u}_{1,i}=\bfB_{eq}'\mathbf{s}_{t,a,i}+\mathbf{y}_{3,i}, \mathbf{u}_{2,i}=\bfB_{eq}''\mathbf{s}_{t,a,i}+g\\
    \< \sendmessageright*{f, \mathbf{u}_{1,i}, \mathbf{u}_{2,i}} \\
    \<\< \bfR_i\xleftarrow[]{\$} \texttt{Bin}_1^{256\times2d}\\
    \< \sendmessageleft*{\bfR_i} \\
    \mathbf{l}_{t,a,i}=\tilde{C}_{t,a,i}\mathbf{m}_i-\tilde{u}_{t,a,i}\\
    \text{compute }\mathbf{z}_{3,i}\text{, s.t. }\vec{z}_{3,i}=\vec{y}_{3,i} + \bfR_i \vec{l}_{t,a,i} \\
    \rej(\vec{z}_{3,i},\bfR_i\vec{l}_{t,a,i}, \mathfrak{s}_3)\\
    \< \sendmessageright*{\mathbf{z}_{3,i}}\\
    \<\<d_{1,j}\leftarrow Z_q, \forall j\in[256]\\
    \< \sendmessageleft*{\{d_{1,j}\}_{j \in [256]}}\\
    x=\sum_{j=1}^{256} d_{1,j}(\begin{pcmbox}\begin{bmatrix}
        \sigma_{-1}(\mathbf{r}_j)^T \tilde{C}_{t,a,i}\>\sigma_{-1}(\mathbf{e}_j)^T
    \end{bmatrix}\end{pcmbox}
    \begin{pcmbox}\begin{bmatrix}
        \mathbf{m}_{i}\\\mathbf{y}_{3,i}
    \end{bmatrix}\end{pcmbox}\\
    - ( \sigma_{-1}(\mathbf{e}_j)^T\mathbf{z}_{3,i}+ \sigma_{-1}(\mathbf{r}_j)^T \tilde{u}_{t,a,i}))\\
    h=g+x,\mathbf{w}_{i} = \bfA_3\mathbf{y}_{1,i}+\bfA_4\mathbf{y}_{2,i}\\ 
    \bfT=
    \begin{pcmbox}\begin{bmatrix}
        \tilde{\bfA}\>0\>0\\
        \sum_{j=1}^{256} d_{1,j}\sigma_{-1}(\mathbf{r}_j)^T\tilde{C}_{t,a,i} \>\sum_{j=1}^{256} d_{1,j}\sigma_{-1}(\mathbf{e}_j)^T\>1
    \end{bmatrix}\end{pcmbox}\\
    \mathbf{v}_{i} = \bfT\begin{pcmbox}\begin{bmatrix}
        \mathbf{y}_{1,i}\\
        -\bfB_{eq}'\mathbf{y}_{2,i}\\
        -\bfB_{eq}''\mathbf{y}_{2,i}
    \end{bmatrix}\end{pcmbox}\\
    \< \sendmessageright*{h,\mathbf{w}_{i}, \mathbf{v}_{i}} \\
    \<\< c_{i} \xleftarrow[]{\$} \mathbf{\mathcal{C}} \\
    \< \sendmessageleft*{c_{i}} \\
    \mathbf{z}_{1,i}=\mathbf{y}_{1,i} + c_{i} \mathbf{m}_{i}, \mathbf{z}_{2,i}=\mathbf{y}_{2,i} + c_{i} \mathbf{s}_{t,a,i} \\
    \rej(\mathbf{z}_{1,i},c_{i} \mathbf{m}_{i} ,\mathfrak{s}_1), \rej(\mathbf{z}_{2,i},c_{i} \mathbf{s}_{t,a,i} ,\mathfrak{s}_2)\\
    \< \sendmessageright*{\mathbf{z}_{1,i},\mathbf{z}_{2,i}}\\
    \<\< \text{Runs } \mathcal{V}^{Eq} 
    }
    \end{minipage}
    \caption{$\name$ Proof of Equivalence Protocol $\pi_i^{Eq}$}
    \label{proto:pi_e}
\end{figure*}

\begin{figure}[!ht]
    \fbox{
    \begin{minipage}{0.45\textwidth}
    Accept if:
    \begin{enumerate}
    \renewcommand{\labelenumi}{(\alph{enumi})}
    \item \label{vrfy:pi_e_1} $\norm{\mathbf{z}_{1,i}} \stackrel{?}{\leq} \mathfrak{s}_1\sqrt{4(2\kappa+\lambda+3)d}$
    \item \label{vrfy:pi_e_2} $\norm{\mathbf{z}_{2,i}} \stackrel{?}{\leq} \mathfrak{s}_2\sqrt{2(\kappa+\lambda +3)d}$
    \item \label{vrfy:pi_e_3} $\norm{\mathbf{z}_{3,i}}_{\infty} \stackrel{?}{\leq} \sqrt{2k}\mathfrak{s}_3$ (This implies, $ \norm{\mathbf{l}_{t,a,i}}_\infty\stackrel{?}{\leq} 2\sqrt{2k}\mathfrak{s}_3$ by \cref{lemma:infty_bound} where $k=128$) 
    \item \label{vrfy:pi_e_4} Constant coefficient of $h=0$
    \item \label{vrfy:pi_e_5} $\bfA_3\mathbf{z}_{1,i}+\bfA_4\mathbf{z}_{2,i}\stackrel{?}{=}\mathbf{w}_{i} + c_{i}f_i$
    \item \label{vrfy:pi_e_6} $\mathbf{v}_i\stackrel{?}{=}\bfT\begin{bmatrix}
        \mathbf{z}_{1,i}\\
        c_i\mathbf{u}_{1,i}-\bfB_{eq}'\mathbf{z}_{2,i}\\
        c_i\mathbf{u}_{2,i}-\bfB_{eq}''\mathbf{z}_{2,i}
    \end{bmatrix}-c_i
    \begin{bmatrix}
        u_i \\ h+\sum_{j=1}^{256}d_{1,j}\sigma_{-1}(\mathbf{e}_j)^T\mathbf{z}_{3,i}+\sum_{j=1}^{256}d_{1,j}\sigma_{-1}(\mathbf{r}_j)^T\tilde{u}_{t,a,i}
    \end{bmatrix}$
    \end{enumerate}
    \end{minipage}
    }
    \caption{$\mathcal{V}^{Eq}$ verify routine for Protocol \hyperref[proto:pi_e]{$\pi^{Eq}_i$}}
    \label{proto:pi_e_verification}
\end{figure}

\subsection{Proof of Key Well-formness}
\label{app:pokw_section}
\label{proto:pi_key_verification}
\label{proto:pi_key}

Since the protocol is essentially the same with some slight modification on the linear relation when compared to proof of equivalence, we describe the  difference between them for simplicity. From \hyperref[proto:pi_e]{$\pi^{EQ}$}, $\mathbf{l}_{t,a,i}$ is just $\mathbf{m}_{t,a,i}$ because we are interested in the shortness of $\mathbf{m}_i$ rather than $\tilde{C}_{t,a,i}\mathbf{m}_i-\tilde{u}_{t,a,i}$. Consequently, we remove $\tilde{C}_{t,a,i}$ in $x$ and $\mathbf{T}$ as well as $\tilde{u}_{t,a,i}$ from $x$ and verification equation (f). The modified protocol is called $\pi^{KW}$.

\begin{theorem}
    Protocol $\pi_{}^{KW}$ is complete, knowledge sound, and honest-verifier zero-knowledge.
\end{theorem}
The protocol is essentially similar to that of proof of equivalence but with modified linear relation. Therefore, the proof proceeds similarly to that of a ABDLOP linear proof, proof of opening, and approximate range proof given in \cite{LNP22}.

\subsection{Proof of Asset for Compact Multi-Asset}
\label{sec:proof_of_asset_compact_construction_formal}

Let $\mathbf{-1} = (-1,\cdots,-1) \in R_q$ . For a polynomial ring $x:= (x_1,\cdots,x_d) \allowbreak \in R_q$ where $\forall i \in [d]: x_i<2^\beta$, we define $x_{bin} = bin(x)$, $pow(x):= \sum_{i=0}^{\lfloor \log_2 x \rfloor} 2^i \cdot X^i \in R_q$ and $\langle x_{bin}, pow(2^\beta-1)\rangle = x_i$. We further define a shifted version, let $1_j := (0_1,\cdots,0_{j-1},1,0_j,\cdots,0_d)$ and $pow(x, j):= \sum_{i=0}^{\lfloor \log_2 x \rfloor} 2^i \cdot X^{i+j} \in R_q$. We assume that $d = a\beta$ for some $a \in \mathbb{Z}^+$ for the simplicity of the protocol. Note that the challenge $c$ need to be stable under automorphism, $\sigma_{-1}(c)=c$ which means effectively that about half of the coefficients uniquely determine the rest. 
The proof first performs an approximate range proof check similar to PoC, PoE. Then, given $z=y+Rv_{bin}$ where $v_{bin}$ represents the flatten binary representation of the values encoded in $v$, we proceed to check the relations that z is well-formed, $v_{bin}$ is binary and $v_{bin}$ is the correct binary decomposition of $v$.
The relation is checked as one quadratic relation shuffled with random linear combinations. All these relations holds when the constant coefficient of the polynomial evaluation vanishes. Therefore, we proceed to invoke the zero-knowledge proof that checks single quadratic equation with automorphisms from \cite{LNP22} and further checks that the quadratic polynomial has vanishing constant coefficient. We also briefly recall the quadratic relation protocol in \cref{sec:quad_proof}. 

With the above, we construct proof of asset compact as shown in \cref{proto:pi_a2} and \cref{proto:pi_a2_verification}.

\begin{theorem}
    Protocol $\pi_{}^{A'}$ is complete, knowledge sound, and honest-verifier zero-knowledge.
\end{theorem}
The proof proceeds similarly to that of ABDLOP quadratic proof, proof of opening, and approximate range proof given in \cite{LNP22}.

\begin{figure*}[!ht]
    \centering
    \begin{minipage}{\textwidth}
    \pcb[codesize=\scriptsize,colspace=-2mm]{
    \underline{\textbf{Participant $i$ (Prover)}} \<\< \underline{\textbf{Auditor (Verifier)}}
    \\\text{Input:}
    \\\bfA, \mathbf{B}, \bfA_a,\bfA_a',\bfB_{a}',\bfB_{a}'', \beta, \hat \beta := \{ pow(2^\beta-1, (j-1)\beta)\}_{j \in [d/\beta]} \<\<\bfA, \mathbf{B}, \bfA_a, \bfA_a',\bfB_{a}',\bfB_{a}'', \beta, \hat \beta
    \\\com := [\com_0, \dots, \com_3]^T := [\mathbf{Ar}, \dots, \mathbf{Br}+v]^T  \<\< \com
    \\ \mathbf{r}, v
    \\[0.1\baselineskip][\hline] \\[-0.5\baselineskip]
    g\xleftarrow[]{\$}\CalRq: \text{ constant of } g=0\\
    \mathbf{y}_{0}\xleftarrow[]{\$} \mathcal{D}_{\mathfrak{s}_0}^{(\beta)d},
    \mathbf{y}_{1}\xleftarrow[]{\$} \mathcal{D}_{\mathfrak{s}_1}^{(\kappa+\lambda +3)d}, \mathbf{y}_{2}\xleftarrow[]{\$} \mathcal{D}_{\mathfrak{s}_2}^{256}\\
    \mathbf{y}_{3}\xleftarrow[]{\$} \mathcal{D}_{\mathfrak{s}_3}^{(\kappa+\lambda +3)d}, s \gets \chi^{(\kappa+\lambda+3)d} \\
    \forall \alpha \in [\beta], v_{bin_\alpha} = bin(v_{(\alpha-1) \cdot (d / \beta) + 1}||\cdots||v_{ (\alpha) \cdot (d / \beta)}) \\ 
    \quad \quad \quad \quad {v}_{bin} = {v}_{bin}||v_{bin_\alpha} \\
    \mathbf{u}_{0}= \bfA'_{a}v_{bin} + \bfA_{a} \mathbf{s} , \mathbf{u}_{y_2}=\bfB_{y_2}\mathbf{s}+\mathbf{y}_{2}, \mathbf{u}_{g}=\bfB_{g}\mathbf{s}+g\\
    \< \sendmessageright*{\mathbf{u}_{0}, \mathbf{u}_{y_2}, \mathbf{u}_{g}} \\
    \<\< \bfR\xleftarrow[]{\$} \texttt{Bin}_1^{256\times \beta d}\\
    \< \sendmessageleft*{\bfR} \\
    \vec l = \vec v_{bin_1} || \cdots || \vec v_{bin_{\beta}} \\ 
    \text{Compute } \mathbf{z}_{2}\text{, s.t. }\vec{z}_{2}=\vec{y}_{2} + \bfR \vec l; \quad \rej(\vec{z}_{2},\bfR \vec l, \mathfrak{s}_2)\\
    \< \sendmessageright*{\mathbf{z}_{2}}\\
    \<\<\{d_{j}\}, \{d_j'\}, \{d_j''\} \sample (Z_q)^{256} \times (Z_q)^\beta \times (Z_q)^d\\
    \< \sendmessageleft*{\{d_{j}\}_{j \in [256]}, \{d_{j}'\}_{j \in [\beta]}, \{d_{j}''\}_{j \in [d]}}\\
    x=\sum_{j=1}^{256} d_{j}(\begin{pcmbox}\begin{bmatrix}
        \sigma_{-1}(\mathbf{r}_j)^T \ \>\sigma_{-1}(\mathbf{e}_j)^T
    \end{bmatrix}\end{pcmbox}
    \begin{pcmbox}\begin{bmatrix}
        v_{bin_1} || \cdots || v_{bin_{\beta}} \\\mathbf{y}_{2}
    \end{bmatrix}\end{pcmbox}\\
    - ( \sigma_{-1}(\mathbf{e}_j)^T\mathbf{z}_{2}) ) + \sum_{j=1}^{\beta} d'_j (\sigma_{-1}(v_{bin_j})^T v_{bin_j} + \sigma_{-1}(\mathbf{-1})^T v_{bin_j})\\
     + \sum_{j=1}^{d} d_j''(\sigma_{-1}(\hat \beta_{j \mod (d/\beta)}))^T v_{bin_{\lfloor j.\beta / (d+1) \rfloor}} - \sigma_{-1}(1_j)^T v) \\
     \codecomment{Well-formness of ($z_2$) $\wedge$  <$v_{bin_\alpha}, v_{bin_\alpha} - \mathbf{1}$>=0 $\wedge$ $v_{bin_\alpha}$ is the correct binary representation) }\\
    h=g+x,\mathbf{w_1} = \bfA_a' \mathbf{y}_{0} + \bfA_a \mathbf{y}_{1}, \mathbf{w_2} = \bfA \mathbf{y}_{3}\\ 
    y_{{bin_{j}}} = \mathbf{y}_0, y_{v} = -\mathbf{B}\mathbf{y}_3, y_{y_2} = -\mathbf{B}_{y_2}\mathbf{y}_1, y_g = -\mathbf{B}_g\mathbf{y}_1\\
    g_1 = \sum_{j=1}^{256} d_{j}(\begin{pcmbox}\begin{bmatrix}
        \sigma_{-1}(\mathbf{r}_j)^T \ \>\sigma_{-1}(\mathbf{e}_j)^T
    \end{bmatrix}\end{pcmbox}
    \begin{pcmbox}\begin{bmatrix}
        y_{bin_1} || \cdots || y_{bin_\beta} \\ y_{y_2}
    \end{bmatrix}\end{pcmbox}\\   
    + \sum_{j=1}^{\beta} d_j' (\sigma_{-1}(v_{bin_j})^T\cdot (y_{bin_j}) + \sigma_{-1}(y_{bin_j})^T \cdot v_{bin_j} + \sigma_{-1}(\mathbf{-1})^T \cdot (y_{bin_j})) \\
    + \sum_{j=1}^{d} d_j'' (\sigma_{-1}(\hat \beta_{j \mod (d/\beta)}))^T  y_{bin_{\lfloor j.\beta / (d+1) \rfloor}} - \sigma_{-1}(1_j)^T y_v) + y_g \\
    \mathbf{u}_{g_1} = \bfB_{g_1}\mathbf{s} + g_1; \quad \mathbf{v} = \sum_{j=1}^{\beta} d_j' \sigma_{-1}(\mathbf{y}_{bin_j})^T \cdot \mathbf{y}_{bin_j}  + \bfB_{g_1}\mathbf{y}_1 \\
    \< \sendmessageright*{h,\mathbf{w_1}, \mathbf{w_2}, \mathbf{v}, \mathbf{u}_{g_1}} \\
    \<\< c \xleftarrow[]{\$} \mathbf{\mathcal{C}} \\
    \< \sendmessageleft*{c} \\
    \mathbf{z}_{0}=\mathbf{y}_{0} + c \mathbf{v_{bin}}, \mathbf{z}_{1}=\mathbf{y}_{1} + c \mathbf{s}, \mathbf{z}_{3}=\mathbf{y}_{3} + c \mathbf{r}; \quad \rej(\mathbf{z}_{0},c \mathbf{v_{bin}} ,\mathfrak{s}_0), \rej(\mathbf{z}_{1},c \mathbf{r} ,\mathfrak{s}_1), \rej(\mathbf{z}_{3},c \mathbf{r} ,\mathfrak{s}_3)\\ \\
    \< \sendmessageright*{\mathbf{z}_{0},\mathbf{z}_{1}, \mathbf{z}_3}\\
    \<\< \text{Runs } \mathcal{V}^{A'} 
    }
    \end{minipage}
    \caption{$\nameextension$ Proof of Asset for Compact Multi-Asset Protocol $\pi_i^{A'}$ }
    \label{proto:pi_a2}
\end{figure*}

\begin{figure}[!ht]
    \fbox{
    \begin{minipage}{0.45\textwidth}
    Compute $z_g = c\mathbf{u}_{_g}-\mathbf{B}_g\mathbf{z}_1$, $z_{y_2} = c\mathbf{u}_{y_2}- \bfB_{y_2}\mathbf{z}_1$, $z_v = c\com_{3}- \bfB\mathbf{z}_3$,$ \ f = c\mathbf{u}_{g_1}-\bfB_{g_1}\mathbf{z}_1$ and parse $[z_{bin_1}\ \ldots\ z_{bin_{\beta}}]:= z_0^T$ \\
    Accept if:
    \begin{enumerate}
    \renewcommand{\labelenumi}{(\alph{enumi})}
    \item \label{vrfy:pi_a2_1} $\norm{\mathbf{z}_{1}} \stackrel{?}{\leq} \mathfrak{s}_1\sqrt{2(\kappa+\lambda+3)d} \wedge \norm{\mathbf{z}_{0}} \stackrel{?}{\leq} \mathfrak{s}_0\sqrt{2(\beta)d} \wedge  \norm{\mathbf{z}_{3}} \stackrel{?}{\leq} \mathfrak{s}_3\sqrt{2(\kappa+\lambda+3)d}$
    \item \label{vrfy:pi_a2_3} $\norm{\mathbf{z}_{2}} \stackrel{?}{\leq} 1.64\cdot\mathfrak{s}_2\sqrt{256}$ (This implies, $ \norm{v}\stackrel{?}{\leq} \frac{2}{\sqrt{26}} (1.64) \mathfrak{s}_2 \sqrt{256}$ by \cref{lemma:2_bound_rwy}) 
    \item \label{vrfy:pi_a2_4} Constant coefficient of $h=0$
    \item \label{vrfy:pi_a2_5} $ \mathbf{A}_a'\mathbf{z}_0 + \bfA_a \mathbf{z}_1 = \mathbf{w}_1+c\cdot \mathbf{u}_0 \wedge \ \bfA\mathbf{z}_{3}\stackrel{?}{=}\mathbf{w}_2 + c \cdot \com_0$
    \item \label{vrfy:pi_a2_6} $\mathbf{v}\stackrel{?}{=}
    \sum_{j=1}^{256} d_{j}(
    \begin{bmatrix}
    \sigma_{-1}(\mathbf{r}_j)^T \quad \sigma_{-1}(\mathbf{e}_j)^T
    \end{bmatrix}\
    \begin{bmatrix}
        z_{bin_1}||\cdots||z_{bin_\beta} \\ z_{y_2}
    \end{bmatrix})\\   
    + \sum_{j=1}^{\beta} d_j'(\sigma_{-1}(z_{bin_{j}})^T z_{bin_j}+c\sigma_{-1}(\mathbf{-1})^T z_{bin_j})) \\
    + c \sum_{j=1}^{d} d_j''(\sigma_{-1}(\hat \beta_{j \mod d/\beta})^T z_{bin_{\lfloor j.\beta / (d+1) \rfloor}} - \sigma_{-1}(1_j)^T z_v) + cz_g
    \\-c^2
    \begin{bmatrix}
        h+\sum_{j=1}^{256}d_{j}\sigma_{-1}(\mathbf{e}_j)^T\mathbf{z}_{2}
    \end{bmatrix} - f$
    \end{enumerate}
    \end{minipage}
    }
    \caption{$\mathcal{V}^{A'}$ verify routine for Protocol \hyperref[proto:pi_a2]{$\pi^{A'}_i$}}
    \label{proto:pi_a2_verification}
\end{figure}

\section{Deferred Proofs}

\subsection{Proof of Balance}
\label{app:proof_of_balance}

We restate the theorem for convenient.
\begin{theorem}
\label{thm:g1}
    Protocol $\pi_{}^B$ is complete, knowledge sound, and honest-verifier zero-knowledge.
\end{theorem}

\begin{proof}

\textit{Completeness:} For each transaction $t$ and asset $a$, an honest sender computes the commitment $\com_{t,a} = \sum_i\com_{t,a,i}$ and $\mathbf{r}_{t,a}=\sum_i\mathbf{r}_{t,a,i}$ and follows the behavior outlined in $\pi^B$ (see \cref{proto:pi_b}). Unless the honest sender aborts due to rejection sampling, we show that the honest verifier accepts with overwhelming probability.

Conditioned on non-aborting by the sender, from \cref{lemma:rejsampling}, the distribution of $\mathbf{z}_{t,a}$ is $\frac{2^{-k}}{3}$ close to the distribution $\mathcal{D}_{\mathfrak{s}}^{(\kappa+\lambda+3)d}$, where $\mathfrak{s} = 11 \omega \| \mathbf{r}_{t,a} \|$. Furthermore, from \cref{lemma:probability_distribution_tail_bound}, we have that $\Pr (\| \mathbf{z}_{t,a} \| \leq \mathfrak{s} \sqrt{2(\kappa+\lambda+3)d})$ is at least $1-2^{-2(\kappa+\lambda+3)d}$. Thus, we have that the first verification step passes with probability at least $1-2^{-(\kappa+\lambda+3)d/8}- \frac{2^{-k}}{3}$. One can also observe that the other two verification steps are always true for an honest sender (evaluating $\mathbf{z}_{t,a}= y+c \mathbf{r}_{t,a}$ ), conditioned on non-abortion in rejection sampling as follows:
\begin{itemize}
    \item $ \mathbf{A} \mathbf{z}_{t,a} =\mathbf{A} \mathbf{y} + c \mathbf{A} \mathbf{r}_{t,a} = \mathbf{w} + c \left(\sum_i \com_{0,t,a,i}\right) $
    \item $\mathbf{B}^T \mathbf{z}_{t,a} =\mathbf{B}^T \mathbf{y} + c \mathbf{B}^T \mathbf{r}_{t,a} = u_{t,a} + c\left(\sum_i \com_{3,t,a,i} - \sum_i v_{t,a,i}\right)$. In the honest case $\sum_i v_{t,a,i} = 0$,
    \\
    Thus, $u_{t,a} + c\left(\sum_i \com_{3,t,a,i} - \sum_i v_{t,a,i}\right) = u_{t,a} + c\left(\sum_i \allowbreak \com_{3,t,a,i} \right)$.
\end{itemize}
This completes the completeness proof.

\textit{Honest-verifier zero-knowledge:}
The proof follows similarly as the proof of honest-verifier zero-knowledge for proof of consistency with minimal modifications as follows: 
\begin{itemize}
    \item The simulator begins by sampling $\mathbf{z}_{t,a}$ from the distribution $\mathcal{D}_{\mathfrak{s}}^{(\mathrm{k}+\lambda+3)d}$. This distribution is statistically close to the one that would be generated in a non-aborting honest transcript, as enforced by \cref{lemma:rejsampling}. 
    \item The \cref{lemma:rejsampling} also implies that $\mathbf{z}_{t,a}$ is independent of the term $c \mathbf{r}_{t,a}$. This independence allows the simulator to sample $c$ separately from the distribution $\mathcal{C}$ without affecting the statistical closeness of the transcript.
    \item The remaining messages in the first round of the transcript, namely $\mathbf{w}, u_{t,a}$, are determined by the verification equations in an honest transcript, i.e., 
$\mathbf{w}_{}=\bfA\mathbf{z}_{t,a}-c_{}\com_{0,t,a}$ and $u_{t,a}=\bfB^T\mathbf{z}_{t,a} - c_{}(\com_{3,t,a})$. Due to the correctness of the protocol, these messages can be computed by the simulator in such a way that the verification equations hold true.
\end{itemize}
By ensuring that the verification equations are satisfied, the simulator produces a transcript that is statistically close to an honest transcript. This means that the simulated transcript is indistinguishable from a real one to the verifier, thus maintaining zero-knowledge.

\textit{Knowledge Soundness:}  We construct an extractor $\extractor$ with rewinding black-box access to the sender specified in $\pi^C_i$, that extracts the weak opening to the commitment. The extractor $\extractor$ repeatedly runs $\calP^*$ with freshly sampled challenges until it hits an accepting transcript. Then the extractor runs $\calP^*$ with freshly sampled challenges until it hits another accepting transcript. The goal is to obtain two accepting transcripts such that $ \bar{c}=c-c_{}^\prime $ is invertible. Moreover, these transcripts
need to all contain the same prover commitments $\mathbf{w}_{}, u_{t,a}$ as in the first accepting transcript. From Corollary \ref{corollary:c_invertible}, one can note that the probability that $ \bar{c}=c-c_{}^\prime $ is non-invertible for a uniformly drawn $c_{}^\prime$ is at most $\mathcal{O}\left(\frac{l}{q^{d/l}}\right)$.

Let $(\mathbf{w}_{}, u_{t,a},  c_{}, \mathbf{z}_{})$ and $(\mathbf{w}_{}, u_{t,a},  c_{}^\prime, \mathbf{z}_{t,a}^\prime)$ be two accepting transcripts obtained by $\extractor$ while performing BDLOP weak opening, it allows a computation of $\bar{c}=c_{}-c_{}^\prime$ and $\bar{\mathbf{z}}_{t,a}=\mathbf{z}_{t,a}-\mathbf{z}_{t,a}^\prime$ such that:
\begin{align}
    \label{eq:pib_soundness1}
    \bfA\bar{\mathbf{z}}_{t,a}=\bar{c}\com_{0,t,a}.
\end{align}    
Now the $\extractor$ defines ${v}^*$ such that:
\begin{align}
    \label{eq:pib_soundness2}
    {v}^* &=  \com_{3,t,a}-\bar{c}^{-1}\bfB^T\bar{\mathbf{z}}_{t,a}.
\end{align}
From verification equations in $\pi^B$, we have:
\begin{align}
    \label{eq:pib_soundness3}
   {\bfB}^T\mathbf{z}_{t,a} &= c_{}\com_{3,t,a}  + u_{t,a}\\
    \label{eq:pib_soundness4}
    \ {\bfB}^T\mathbf{z}'_{t,a} &= c'_{}\com_{3,t,a}  + u_{t,a}.
\end{align}
By subtracting between \eqref{eq:pib_soundness3} and \eqref{eq:pib_soundness4} and substituting values from \eqref{eq:pib_soundness2}, we obtain ${v}^*=0$.

Finally, let $\mathbf{r}^* = \bar{c}^{-1}\bar{\mathbf{z}}_{t,a}$. $\extractor$ can therefore obtain $({v^*},\mathbf{r}^*)$ that satisfies 
\begin{align} 
    \label{eq:pib_soundness5}
    \begin{bmatrix}
        \bfA \\
        \bfB^T
    \end{bmatrix} \mathbf{r}^* + \begin{bmatrix}
        \mathbf{0}^n \\
        v^* \\ 
    \end{bmatrix}
    =\begin{bmatrix}
        \com_{0,t,a} \\ 
        \com_{3,t,a}
    \end{bmatrix},
\end{align}
and $\mathbf{y}^*$ similarly to the proof of \cref{thm:pi_c}. 

This concludes that a malicious sender who does not know the information of $(\sum_i{v_{t,a,i}},\mathbf{r}_{t,a})$ can only convince the verifier by either breaking the commitment by solving the $\msis_{\kappa,\kappa+\lambda+3,8\omega\mathfrak{s}\sqrt{2(\kappa+\lambda+3)d}}$ problem on $\bfA$ as shown in Lemma 4.3 in \cite{ALS20}, or by at most responding to specific challenges per commitment such that $\bar{c}$ is non-invertible. That is, with knowledge error as $\mathcal{O}\left(\frac{l}{q^{d/l}}\right)$.
\end{proof}

\subsection{Proof of Consistency}
\label{app:proof_of_consistency}

We restate the theorem for convenient.

\begin{theorem}
    Protocol $\pi_{i}^C$ given in \cref{proto:pi_c} is complete, sound, and honest-verifier zero-knowledge.
\end{theorem}

\begin{proof}
    \textit{Completeness:} For each transaction $t$ and asset $a$, an honest sender computes the commitment $\com_{t,a,i}$ (for participant $i$) and follows the behavior outlined in $\pi^C_i$ (see Protocol~\ref{proto:pi_c}). The goal is to show that the honest verifier accepts with overwhelming probability, unless the honest prover aborts due to rejection sampling. First, note that $\norm{\bfR_i \vec{r}_{t,a,i}}\leq\sqrt{337(\kappa+\lambda+3)d}$ holds with probability at least $1-2^{-128}$ by \cref{lemma:2_bound}. Conditioned on non-aborting by the sender with probability
    \begin{align*}
        \left(\frac{1-2^{-k}}{M}\right)^3\geq \left(\frac{1}{M}\right)^3-\texttt{neg}(k)  
    \end{align*}
    by \cref{lemma:rejsampling}, the verification steps \ref{vrfy:pi_c_1}, \ref{vrfy:pi_c_2} and \ref{vrfy:pi_c_3} pass with overwhelming probabilities: 
    \begin{align*}
        \Pr (\| \mathbf{z}_{1,i} \| \leq \mathfrak{s}_1 \sqrt{2(\kappa+\lambda+3)d})-\frac{2^{-k}}{M} &\geq 1-2^{-(\kappa+\lambda+3)d/8}-\frac{2^{-k}}{3}\\
        \Pr (\| \mathbf{z}_{2,i} \| \leq \mathfrak{s}_2 \sqrt{2(\kappa+\lambda+3)d})-\frac{2^{-k}}{M} &\geq 1-2^{-(\kappa+\lambda+3)d/8}-\frac{2^{-k}}{3}\\
        \Pr (\| \mathbf{z}_{3,i} \|_\infty \leq \sqrt{2k}\mathfrak{s}_3)-\frac{2^{-k}}{M} &\geq 1-256\times 2e^{-k}-\frac{2^{-k}}{3},
    \end{align*}
    by \cref{lemma:probability_distribution_tail_bound} and union bound. In verification step \ref{vrfy:pi_c_4}, for $j\in[256]$ and $v_{t,a,i}\in Z_q$, we have 
    \begin{align*}
        \dotprod{\mathbf{e}_j}{\mathbf{z}_{3,i}}=\dotprod{\mathbf{e}_j}{\mathbf{y}_{3,i}}+\dotprod{\mathbf{r}_j}{\mathbf{r}_{t,a,i}}
    \end{align*}
    by construction, which is equivalent to the constant coefficient of the following polynomial being zero:
    \begin{align*}
        \sum_{j=1}^{256} d_{1,j}(\begin{bmatrix}
        \sigma_{-1}(\mathbf{r}_j)^T\>\sigma_{-1}(\mathbf{e}_j)^T
    \end{bmatrix}
    \begin{bmatrix}
        \mathbf{r}_{t,a,i}\\\mathbf{y}_{3,i}
    \end{bmatrix}-\sigma_{-1}(\mathbf{e}_j)^T\mathbf{z}_{3,i})
    \\+\sum_{j=2}^{128}d_{2,j-1}\sigma_{-1}(\mathbf{e}_j')^T v_{t,a,i}
    \end{align*}
    by \cref{lemma:constant_coeff}. Then if the constant coefficient of $g$ is zero, the constant coefficient of $h$ must be zero.
    
    The verification step \ref{vrfy:pi_c_5} is given by ABDLOP proof of opening as in Protocol~\ref{proto:ZKPoO_abdlop}. Finally, we have the last verification step \ref{vrfy:pi_c_6} for each row as in the following:
    \begin{itemize}
    \item  $\bfA \mathbf{z}_{1,i}-c_i\com_{0,t,a,i}=\bfA \mathbf{z}_{1,i}-c_i\bfA \mathbf{r}_{t,a,i} = \bfA \mathbf{y}_{1,i}=\mathbf{v}_{0,i},$
    \item $\pk^T_{1,i}\mathbf{z}_{1,i}+c_if_{1,i}-\bfB_1 \mathbf{z}_{2,i}-c_i\com_{1,t,a,i}=\pk^T_{1,i}\mathbf{z}_{1,i}+c_i(\bfB_1 \mathbf{s}_{t,a,i}+v_{t,a,i})-\bfB_1 \mathbf{z}_{2,i}-c_i(\pk^T_{1,i}\mathbf{r}_{t,a,i}+v_{t,a,i}) =\pk^T_{1,i}\mathbf{y}_{1,i}-\bfB_1 \mathbf{y}_{2,i}=\mathbf{v}_{1,i},$
    \item  $\pk^T_{2,i}\mathbf{z}_{1,i}+\sqrt{q}(c_if_{1,i}-\bfB_1 \mathbf{z}_{2,i})-c_i\com_{2,t,a,i}=\pk^T_{2,i}\mathbf{z}_{1,i}+\sqrt{q}c_i(\bfB_1 \mathbf{s}_{t,a,i}+v_{t,a,i})-\sqrt{q}\bfB_1 \mathbf{z}_{2,i}-c_i(\pk^T_{2,i}\mathbf{r}_{t,a,i}+\sqrt{q}v_{t,a,i})=\pk^T_{2,i}\mathbf{y}_{1,i}-\sqrt{q}\bfB_1 \mathbf{y}_{2,i}=\mathbf{v}_{2,i},$
    \item $\bfB^T\mathbf{z}_{1,i}+c_if_{1,i}-\bfB_1 \mathbf{z}_{2,i}-c_i\com_{3,t,a,i}=\bfB^T\mathbf{z}_{1,i}+c_i(\bfB_1 \mathbf{s}_{t,a,i}+v_{t,a,i})-\bfB_1 \mathbf{z}_{2,i}-c_i(\bfB^T\mathbf{r}_{t,a,i}+v_{t,a,i})=\bfB^T\mathbf{y}_{1,i}-\bfB_1 \mathbf{y}_{2,i}=\mathbf{v}_{3,i},$
    \item   $\sum_{j=1}^{256} d_{1,j}\sigma_{-1}(\mathbf{r}_j)^T \mathbf{z}_{1,i}+ \sum_{j=2}^{128} d_{2,j-1} \sigma_{-1}(\mathbf{e}_j')^T(c_if_{1,i}-\bfB_1 \mathbf{z}_{2,i})+\sum_{j=1}^{256} d_{1,j}\sigma_{-1}(\mathbf{e}_j)^T(c_iu_{1,i}-\bfB_c' \mathbf{z}_{2,i})+c_iu_{2,i}-\bfB_c'' \mathbf{z}_{2,i}-c_i(h+\sum_{j=1}^{256}d_{1,j}\sigma_{-1}(\mathbf{e}_j)^T\mathbf{z}_{3,i})\\=\sum_{j=1}^{256} d_{1,j}\sigma_{-1}(\mathbf{r}_j)^T \mathbf{y}_{1,i}-\sum_{j=2}^{128} d_{2,j-1} \sigma_{-1}(\mathbf{e}_j')^T \bfB_1 \mathbf{y}_{2,i}\\-\sum_{j=1}^{256} d_{1,j}\sigma_{-1}(\mathbf{e}_j)^T\bfB_c'\mathbf{y}_{2,i}-\bfB_c''\mathbf{y}_{2,i} \\= \mathbf{v}_{4,i}.$
    \end{itemize}
    This concludes the proof of completeness.

\textit{Soundness:} 
    We construct an extractor $\extractor$ with rewinding black-box access to the prover specified in $\pi_i^C$. The extractor $\extractor$ repeatedly runs $\calP^*$ with freshly sampled challenges until it hits an accepting transcript. Then, let $\calP^*$ fix the same randomness, with the same challenges until just after the third last round, and get a fresh $c_i'$. The goal is to obtain a second accepting transcript with response $\mathbf{z}_{1,i}',\mathbf{z}_{2,i}'$ using rewinding techniques, while $ \bar{c}=c_i-c_{i}^\prime $ is invertible. Meanwhile, the rest in the transcript should remain the same as the first accepting transcript. From \cref{corollary:c_invertible}, one can note that the probability that $ \bar{c}$ is non-invertible for a uniformly drawn $c_{i}^\prime$ is at most $\mathcal{O}\left(\frac{l}{q^{d/l}}\right)$. With two accepting transcripts, in our case, denoted by:
    \begin{align*}
        (f_i,\mathbf{u}_{1,i},\mathbf{u}_{2,i},\mathbf{R}_i,\mathbf{z}_{3,i},\{d_{1,j}\},\{d_{2,j}\},h,\mathbf{w}_i,\mathbf{v}_i,c_i,\mathbf{z}_{1,i},\mathbf{z}_{2,i})\\
        (f_i,\mathbf{u}_{1,i},\mathbf{u}_{2,i},\mathbf{R}_i,\mathbf{z}_{3,i},\{d_{1,j}\},\{d_{2,j}\},h,\mathbf{w}_i,\mathbf{v}_i,c_i',\mathbf{z}_{1,i}',\mathbf{z}_{2,i}')
    \end{align*}
    the $\extractor$ is capable of performing a weak opening of ABDLOP commitment as defined in \cref{def:opening_abdlop} and \cref{coro:security_abdlop} hold: Let $\bar{\mathbf{z}}_{1,i} =\mathbf{z}_{1,i}-\mathbf{z}_{1,i}' $, $\bar{\mathbf{z}}_{2,i} =\mathbf{z}_{2,i}-\mathbf{z}_{2,i}'$, $\extractor$ further extracts:  
    \begin{align*}
        &\mathbf{r}^*_{t,a,i} =\bar{\mathbf{z}}_{1,i} / \bar{c}\\
        &\mathbf{s}^*_{t,a,i} =\bar{\mathbf{z}}_{2,i} / \bar{c}\\
        &v^*_{t,a,i} = f_{1,i}- \bfB_1  s^*_{t,a,i}\\
        &\mathbf{y}^*_{3,i} = \mathbf{u}_{1,i}-\bfB'_c s^*_{t,a,i}\\
        &g^* = \mathbf{u}_{2,i}-\bfB''_c s^*_{t,a,i}.
    \end{align*}
    Finally, $\extractor$ computes 
    \begin{align*}
        &\mathbf{y}_{1,i}^*=\mathbf{z}_{1,i}-c_i\mathbf{r}^*_{t,a,i}\\
        &\mathbf{y}_{1,i}'^*=\mathbf{z}_{1,i}'-c_i'\mathbf{r}^*_{t,a,i}\\
        &\mathbf{y}_{2,i}^*=\mathbf{z}_{2,i}-c_i\mathbf{s}^*_{t,a,i}\\
        &\mathbf{y}_{2,i}'^*=\mathbf{z}_{2,i}'-c_i'\mathbf{s}^*_{t,a,i}
    \end{align*}
    and $\mathbf{y}_{1,i}^*=\mathbf{y}_{1,i}'^*$, $\mathbf{y}_{2,i}^*=\mathbf{y}_{2,i}'^*$. Otherwise, the prover finds a solution of \\   $\msis_{\kappa,2(\kappa+\lambda+3),8\omega\sqrt{(\mathfrak{s_1}\sqrt{2(\kappa+\lambda+3)d})^2+(\mathfrak{s_2}\sqrt{2(\kappa+\lambda+3)d})^2}}$.

    We next compute the probability of getting an acceptance for any malicious prover. This can be performed in the following way: First, conditioned on verification passes, as specified in \ref{proto:pi_c}. Suppose for some $j$ in 
    \begin{align*}
        x^*=\sum_{j=1}^{256} d_{1,j}(\begin{bmatrix}
        \sigma_{-1}(\mathbf{r}_j)^T\>\sigma_{-1}(\mathbf{e}_j)^T
    \end{bmatrix}
    \begin{bmatrix}
        \mathbf{r}^*_{t,a,i}\\\mathbf{y}^*_{3,i}
    \end{bmatrix}-\sigma_{-1}(\mathbf{e}_j)^T\mathbf{z}_{3,i})\\+\sum_{j=2}^{128}d_{2,j-1}\sigma_{-1}(\mathbf{e}_j')^T v^*_{t,a,i},
    \end{align*}
    the constant coefficient of 
    \begin{align*}
        \begin{bmatrix}
        \sigma_{-1}(\mathbf{r}_j)^T\>\sigma_{-1}(\mathbf{e}_j)^T
    \end{bmatrix}
    \begin{bmatrix}
        \mathbf{r}^*_{t,a,i}\\\mathbf{y}^*_{3,i}
    \end{bmatrix}-\sigma_{-1}(\mathbf{e}_j)^T\mathbf{z}_{3,i}
    \end{align*}
    or $\sigma_{-1}(\mathbf{e}_j')^T v^*_{t,a,i}$ is not equal to zero. As $\{d_{1,j}\},\{d_{2,j}\}$ are chosen uniformly at random by the verifier, the probability that the constant coefficient of $h=0$ holds is at most $1/q$. Secondly, we are checking when 
    $\norm{\mathbf{z}_{3,i}}_{\infty} \leq \sqrt{2k}\mathfrak{s}_3$, but $ \norm{\mathbf{r}^*_{t,a,i}}_\infty> 2*\sqrt{2k}\mathfrak{s}_3$. Here, the probability is at most $2^{-256}$ according to \cref{lemma:infty_bound}.
    
    In conclusion, a malicious sender who does not know the information of $(v_{t,a,i},\mathbf{r}_{t,a,i}, \mathbf{s}_{t,a,i})$ can only convince the verifier by either breaking the commitments by solving \\
    $\msis_{\kappa,2(\kappa+\lambda+3),8\omega\sqrt{(\mathfrak{s_1}\sqrt{2(\kappa+\lambda+3)d})^2+(\mathfrak{s_2}\sqrt{2(\kappa+\lambda+3)d})^2}}$ as shown in Lemma 4.3 in \cite{ALS20}, or by at most responding to specific challenges per commitment such that $\bar{c}$ is non-invertible. That is, Protocol $\pi_{i}^C$ is sound with knowledge error as
    \begin{align*}
        \mathcal{O}\left(\frac{l}{q^{d/l}}\right)+\frac{1}{q}+\frac{1}{2^{256}}.
    \end{align*}

    \textit{Honest-verifier zero-knowledge:} In the context of honest-verifier zero-knowledge proofs, the goal is to simulate a transcript between an honest prover and an honest verifier such that the simulated transcript is indistinguishable from a non-aborting accepting transcript. This ensures that no additional information is leaked to the verifier beyond the validity of the statement being proven. 

    First, each $(\mathbf{pk}_{1,i}, \mathbf{pk}_{2,i})$ is indistinguishable from uniform under the $\mlwe_{\kappa+\lambda+3, \kappa, \chi}$ assumption. And the hiding property of the input commitment is guaranteed under the hardness of dual $\mlwe_{\lambda,\kappa+3,\chi}$. Then, we construct a polynomial-time simulator $\simulator$ such that:
    \begin{enumerate}
        \item The simulator $\simulator$ samples $\mathbf{z}_{1,i},\mathbf{z}_{2,i},\mathbf{z}_{3,i}$ from the distributions $\mathcal{D}_{\mathfrak{s}_1}^{(\kappa+\lambda+3)d},\mathcal{D}_{\mathfrak{s}_2}^{(\kappa+\lambda+3)d},\mathcal{D}_{\mathfrak{s}_3}^{256}$, respectively. These distributions are statistically close to the ones that would be generated in a non-aborting honest transcript, as enforced by \cref{lemma:rejsampling}. This step ensures that the sampled value is consistent with what would be expected in a real interaction in the end.
        \item The \cref{lemma:rejsampling} also implies that  $\mathbf{z}_{1,i},\mathbf{z}_{2,i},\mathbf{z}_{3,i}$ are independent of the terms $c_i \mathbf{r}_{t,a,i},c_i \mathbf{s}_{t,a,i},\mathbf{R}_i\mathbf{r}_{t,a,i}$, respectively. This independence allows $\simulator$ to sample $c_i$ and $\mathbf{R}_i$ separately from the distributions $\mathcal{C}$ and $\texttt{Bin}_1^{256\times(\kappa+\lambda +3)d}$ without affecting the statistical closeness of the transcript.
        \item Since $g$ is chosen uniformly at random from $\CalRq$, where the constant coefficient is $0$, it allows $\simulator$ to sample $h$ from the same distribution, where the distribution is statistically indistinguishable from $g$, as well as $\{d_{1,j}\}_{j \in [256]},\{d_{2,j}\}_{j \in [127]}$.
        \item Finally, $\simulator$ simulates $f_i, \mathbf{u}_{1,i}, \mathbf{u}_{2,i}$ by sampling from $\CalRq^{\kappa+1+256/d+1}$ uniformly at random, and indistinguishable from the actual transcript under the dual $\mlwe_{\lambda,\kappa+1+256/d+1,\chi}$ assumption.
    \end{enumerate}
    Then, $\simulator$ computes the remaining messages of the transcript, namely $\mathbf{w}_i$ and $\mathbf{v}_i$, by the verification equations \ref{vrfy:pi_c_5} and \ref{vrfy:pi_c_6} in an honest transcript. Due to the correctness of the protocol, these messages can be computed by the simulator in such a way that the verification equations hold true, and therefore the protocol is honest-verifier zero-knowledge.

\end{proof}

\subsection{Proof of Equivalence}
\label{app:proof_of_equivalence}

We restate the theorem for convenient.

\begin{theorem}
    Protocol $\pi_i^{Eq}$ is complete, sound, and honest-verifier zero-knowledge in the classical setting.
\end{theorem}

\begin{proof}
\textit{Completeness:} For each transaction $t$ and asset $a$, an honest participant $i$ computes the re-commitment $\com'_{t,a,i}$ after verifying the commitment from the sender and successfully extracting $v_{t,a,i}$. Unless the protocol aborts due to rejection sampling, we show that the honest verifier accepts with overwhelming probability. 

Conditioned on non-aborting by the participant $i$, we follow the computation similarly to the proof of Theorem \ref{thm:pi_c} and obtain the first three verification steps \ref{vrfy:pi_e_1}, \ref{vrfy:pi_e_2} and \ref{vrfy:pi_e_3} pass with overwhelming probabilities:
\begin{align*}
    \Pr (\| \mathbf{z}_{1,i} \| \leq \mathfrak{s}_1\sqrt{4(2\kappa+\lambda+3)d})-\frac{2^{-k}}{M} &\geq 1-2^{-(2\kappa+\lambda+3)d/4}-\frac{2^{-k}}{3}\\
    \Pr (\| \mathbf{z}_{2,i} \| \leq \mathfrak{s}_2\sqrt{2(\kappa+\lambda +3)d})-\frac{2^{-k}}{M} &\geq 1-2^{-(\kappa+\lambda +3)d/8}-\frac{2^{-k}}{3}\\
    \Pr (\| \mathbf{z}_{3,i} \|_\infty \leq \sqrt{2k}\mathfrak{s}_3)-\frac{2^{-k}}{M} &\geq 1-256\times 2e^{-k}-\frac{2^{-k}}{3},
\end{align*}
Furthermore, we combine Lemma \ref{lemma:comp_equil}, and observe that the remaining four verification steps are always true for an honest participant. This completes the completeness proof. 

\textit{Honest-verifier zero-knowledge:} First, each $(\mathbf{pk}_{1,i}, \mathbf{pk}_{2,i})$ is indistinguishable from uniform under the $\mlwe_{\kappa+\lambda+3, \kappa, \chi}$ assumption. And the hiding property of the input commitment is guaranteed under the hardness of dual $\mlwe_{\lambda +3,\kappa,\chi}$. Then, we construct a polynomial-time simulator $\simulator$ such that:
    \begin{enumerate}
        \item The simulator $\simulator$ samples $\mathbf{z}_{1,i},\mathbf{z}_{2,i},\mathbf{z}_{3,i}$ from the distributions $\mathcal{D}_{\mathfrak{s}_1}^{2(2\kappa+\lambda+3)d},\mathcal{D}_{\mathfrak{s}_2}^{(\kappa+\lambda+3)d},\mathcal{D}_{\mathfrak{s}_3}^{256}$, respectively. These distributions are statistically close to the ones that would be generated in a non-aborting honest transcript, as enforced by \cref{lemma:rejsampling}. This step ensures that the sampled value is consistent with what would be expected in a real interaction in the end.
        \item The \cref{lemma:rejsampling} also implies that  $\mathbf{z}_{1,i},\mathbf{z}_{2,i},\mathbf{z}_{3,i}$ are independent of the terms $c_i \mathbf{m}_{i},c_i \mathbf{s}_{t,a,i},\mathbf{R}_i\mathbf{l}_{t,a,i}$, respectively. This independence allows $\simulator$ to sample $c_i$ and $\mathbf{R}_i$ separately from the distributions $\mathcal{C}$ and $\texttt{Bin}_1^{256\times 2d}$ without affecting the statistical closeness of the transcript.
        \item Since $g$ is chosen uniformly at random from $\CalRq$, where the constant coefficient is $0$, it allows $\simulator$ to sample $h$ from the same distribution, where the distribution is statistically indistinguishable from $g$, as well as $\{d_{1,j}\}_{j \in [256]}$.
        \item Finally, $\simulator$ simulates $f_i, \mathbf{u}_{1,i}, \mathbf{u}_{2,i}$ by sampling uniformly at random, and indistinguishable from the actual transcript under the dual $\mlwe_{\lambda+3,\kappa+256/d+1,\chi}$ assumption.
    \end{enumerate}
    Then, $\simulator$ computes the remaining messages of the transcript, namely $\mathbf{w}_i$ and $\mathbf{v}_i$, by the verification equations \ref{vrfy:pi_e_5} and \ref{vrfy:pi_e_6} in an honest transcript. Due to the correctness of the protocol, these messages can be computed by the simulator in such a way that the verification equations hold true, and therefore the protocol is honest-verifier zero-knowledge.

\textit{Soundness:}  The knowledge soundness follows in similar lines to that of ZKPoC.  The extractor $\extractor$ repeatedly runs $\calP^*$ to get two accepting transcripts, denoted by:
    \begin{align*}
        (f_i,\mathbf{u}_{1,i},\mathbf{u}_{2,i},\mathbf{R}_i,\mathbf{z}_{3,i},\{d_{1,j}\},h,\mathbf{w}_i,\mathbf{v}_i,c_i,\mathbf{z}_{1,i},\mathbf{z}_{2,i})\\
        (f_i,\mathbf{u}_{1,i},\mathbf{u}_{2,i},\mathbf{R}_i,\mathbf{z}_{3,i},\{d_{1,j}\},h,\mathbf{w}_i,\mathbf{v}_i,c_i',\mathbf{z}_{1,i}',\mathbf{z}_{2,i}')
    \end{align*}
    the $\extractor$ is capable of performing a weak opening of commitment $f_i$. Let $\bar{\mathbf{z}}_{1,i} =\mathbf{z}_{1,i}-\mathbf{z}_{1,i}' $, $\bar{\mathbf{z}}_{2,i} =\mathbf{z}_{2,i}-\mathbf{z}_{2,i}'$, $\extractor$ further extracts:  
    \begin{align*}
        &\mathbf{m}^*_{i} =\bar{\mathbf{z}}_{1,i} / \bar{c}\\
        &\mathbf{s}^*_{t,a,i} =\bar{\mathbf{z}}_{2,i} / \bar{c}\\
        &\mathbf{y}^*_{3,i} = \mathbf{u}_{1,i}-\bfB'_{eq} s^*_{t,a,i}\\
        &g^* = \mathbf{u}_{2,i}-\bfB''_{eq} s^*_{t,a,i},
    \end{align*}
    and $\mathbf{y}_{1,i}^*$, $\mathbf{y}_{2,i}^*$ similarly to the proof of \cref{thm:pi_c}.
    
    We next compute the probability of getting an acceptance for any malicious prover. This can be performed in the following way: First, conditioned on verification passes, as specified in \cref{proto:pi_e_verification}. Suppose for some $j$ in 
    \begin{align*}
        x^*=  \sum_{j=1}^{256} d_{1,j}(\begin{bmatrix}
        \sigma_{-1}(\mathbf{r}_j)^T \tilde{C}_{t,a,i}\>\sigma_{-1}(\mathbf{e}_j)^T
    \end{bmatrix}
    \begin{bmatrix}
        \mathbf{m}^*_{i}\\ \mathbf{y}^*_{3,i}
    \end{bmatrix}  
    - \\( \sigma_{-1}(\mathbf{e}_j)^T\mathbf{z}_{3,i}+ \sigma_{-1}(\mathbf{r}_j)^T \tilde{u}_{t,a,i})),
    \end{align*}
    the constant coefficient is not equal to zero. As $\{d_{1,j}\}$ are chosen uniformly at random by the verifier, the probability that the constant coefficient of $h=0$ holds is at most $1/q$. Secondly, we are checking when 
    $\norm{\mathbf{z}_{3,i}}_{\infty} \leq \sqrt{2k}\mathfrak{s}_3$, but $ \norm{\mathbf{l}^*_{t,a,i}}_\infty> 2*\sqrt{2k}\mathfrak{s}_3$. Here, the probability is at most $2^{-256}$ according to \cref{lemma:infty_bound}.
    
    In conclusion, a malicious participant who does not know the information of $(m_{i},\mathbf{s}_{t,a,i})$ can only convince the verifier by either breaking the commitments by solving \\ $\msis_{\kappa,5\kappa+3\lambda+9,8\omega\sqrt{(\mathfrak{s_1}\sqrt{4(2\kappa+\lambda+3)d})^2+(\mathfrak{s_2}\sqrt{2(\kappa+\lambda+3)d})^2}}$ as shown in Lemma 4.3 in \cite{ALS20}, or by at most responding to specific challenges per commitment such that $\bar{c}$ is non-invertible. That is, Protocol $\pi_{i}^{Eq}$ is sound with knowledge error as
    \begin{align*}
        \mathcal{O}\left(\frac{l}{q^{d/l}}\right)+\frac{1}{q}+\frac{1}{2^{256}}.
    \end{align*}
\end{proof}

\subsection{Proof of Asset}
\label{app:proof_of_asset}

We restate the theorem for convenient.
\begin{theorem}
    Protocol $\pi_i^A$ is complete, sound, and honest-verifier zero-knowledge.
\end{theorem}

\begin{proof}
\textit{Completeness}:
For each transaction $t$ and asset $a$, an honest participant $i$ follows Protocol~\ref{proto:pi_a} to prove they have sufficient funds to proceed with the transaction. Unless the protocol aborts due to rejection sampling, we follow the computation similarly to the proof of Theorem \ref{thm:pi_c}. Conditioned on non-aborting by the participant $i$, we obtain the verification step \ref{vrfy:pi_a_1} passes with probability
\begin{align*}
    \Pr \big(\| \mathbf{z}_i \| \leq \mathfrak{s}\sqrt{2(\kappa+\lambda+3)d}\big) -\frac{2^{-k}}{M} \geq 1-2^{-(\kappa+\lambda+3)d/8} -\frac{2^{-k}}{3}
\end{align*}

The verification step \ref{vrfy:pi_a_2} follows straightforwardly from BDLOP proof of opening as Protocol~\ref{proto:ZKPoO_bdlop_linear}. In the third verification, we have:
\begin{align*}
    &(\mathbf{a}_\text{bin}^T\mathbf{z}_i - c_if_{1,i})( \mathbf{a}_\text{bin}^T\mathbf{z}_i - c_if_{1,i}+c_i)+(\mathbf{a}_\text{bin}')^T\mathbf{z}_i -  c_i \mathbf{u}_{1,i}-u_\text{3,i}\\
    =&(\mathbf{a}_\text{bin}^T\mathbf{y}_i)^2+(c_iv_\text{bin})^2+(\mathbf{a}_\text{bin}')^T\mathbf{y}_i-c_i^2v_\text{bin}-u_\text{3,i}\\
    =&(\mathbf{a}_\text{bin}^T\mathbf{y}_i)^2+(c_iv_\text{bin})^2+(\mathbf{a}_\text{bin}')^T\mathbf{y}_i-c_i^2v_\text{bin}-(\mathbf{a}_\text{bin}^T\mathbf{y}_i)^2-(\mathbf{a}_\text{bin}')^T\mathbf{y}_i\\
    =&c_i^2v_\text{bin}(v_\text{bin}-1).  
\end{align*}
Since $v_\text{bin}$'s NTT is binary with an honest prover, the verification step \ref{vrfy:pi_a_3} will end up as 0. In the verification step \ref{vrfy:pi_a_4}, we obtain the following result by moving everything on the left-hand side:
\begin{align*}
    &c_i(\text{NTT}^{-1}(Q^T\phi_i)f_{1,i}-\text{NTT}^{-1}(\phi_i)\com_{1,t,a,i}'+ \mathbf{u}_{2,i}-h)+u_\text{4,i}-\\ 
    &\mathbf{a}_\text{g}^T\mathbf{z}_i-\mathbf{a}_\text{bin}^T\mathbf{z}_i\text{NTT}^{-1}(Q^T\phi_i)+\pk_{1,i}\mathbf{z}_i\text{NTT}^{-1}(\phi_i)\\
    =&\text{NTT}^{-1}(Q^T\phi_i)(c_if_{1,i}+c_i\mathbf{a}_\text{bin}^T\mathbf{y}_i-c_iv_\text{bin}-\mathbf{a}_\text{bin}^T\mathbf{z}_i)\\
    &-\text{NTT}^{-1}(\phi_i)(c_i\com_{1,t,a,i}'-c_iv_\text{tot}+\pk_{1,i}\mathbf{y}_i-\pk_{1,i}\mathbf{z}_i)+c_i\mathbf{u}_{2,i}\\
    &-c_ig+\mathbf{a}_\text{g}^T\mathbf{y}_i-\mathbf{a}_\text{g}^T\mathbf{z}_i\\
    =&0.
\end{align*}
Finally, as $Q(\text{NTT}(v_\text{bin}))=\text{NTT}(v_{\text{tot}})$ in the honest case, we combine \cref{lemma:NTT_as_original} and rewrite the verification step \ref{vrfy:pi_a_5} as:
\begin{align*}
    l\sum_{i=0}^{d/l-1}h_i=\sum_{i=0}^{d/l-1}g_i+\dotprod{Q(\text{NTT}(v_\text{bin}))-\text{NTT}(v_{tot})}{\phi}.
\end{align*}
Since the first $d/l$ coefficients of $g$ are $0$s by construction, the fifth verification holds and concludes the completeness proof.

\textit{Honest-verifier zero-knowledge:}
We follow the idea on proving \cref{thm:pi_c}'s honest-verifier zero-knowledge, and construct the simulator $\simulator$ such that:
\begin{itemize}
    \item $\simulator$ samples $z_i\xleftarrow[]{\$}\mathcal{D}_{\mathfrak{s}}^{(k+\lambda+3)d}$.
    \item $\simulator$ samples $c_i\xleftarrow[]{\$}\mathcal{C}$ and $\phi_i\xleftarrow[]{\$}\mathbb{Z}_q^d$.
    \item $\simulator$ samples $h\xleftarrow[]{\$}(\mathbf{0}_{d/l}\big|\mathbb{Z}_q^{d/l\times(l-1)})$.
    \item $\simulator$ finally samples $f_i,\mathbf{u}_{1,i},\mathbf{u}_{2,i}$ from $\CalRq^{\kappa+3}$.
\end{itemize}
The simulator $\simulator$ further computes $\mathbf{w}_i, u_{3,i}$ and $u_{4,i}$, by verification equation in an honest transcript. Due to the correctness of the protocol, these messages can be computed by the simulator in such a way that the verification equations hold true, and therefore the protocol is honest-verifier zero-knowledge.

\textit{Soundness:} We construct an extractor $\extractor$ with rewinding black-box access to the prover specified in $\pi_i^A$. With two valid transcripts:
\begin{align*}
    &(f_i,\mathbf{u}_{1,i},\mathbf{u}_{2,i},u_{3,i},\mathbf{w}_i,\phi_i,h,u_{4,i},c_i,\mathbf{z}_{i})\\
    &(f_i,\mathbf{u}_{1,i},\mathbf{u}_{2,i},u_{3,i},\mathbf{w}_i,\phi_i,h,u_{4,i},c_i',\mathbf{z}_{i}'),
\end{align*}
where $\bar{c}=c_i-c_i'$ is non-invertible at most $\mathcal{O}\left(\frac{l}{q^{d/l}}\right)$ from \cref{corollary:c_invertible}. $\extractor$ is capable to perform a weak opening of BDLOP commitment as defined in \cref{def:opening_bdlop}. Let $\bar{\mathbf{z}}_i=\mathbf{z}_{i}-\mathbf{z}_{i}'$, $\extractor$ further extracts:
\begin{align*}
    &\mathbf{r}_{\text{PoA}}^*=\bar{\mathbf{z}}_i/\bar{c}\\
    &v_{tot}^* = \com_{1,t,a,i}'-\pk_1^T\mathbf{r}_{\text{PoA}}^*\\
    &v_{bin}^* = f_{1,i}-\mathbf{a}_\text{bin}^T\mathbf{r}_{\text{PoA}}^*\\
    &v_{bin}'^* = \mathbf{u}_{1,i}-(\mathbf{a}_\text{bin}')^T\mathbf{r}_{\text{PoA}}^*\\
    &g^* = \mathbf{u}_{2,i}-\mathbf{a}_\text{g}^T\mathbf{r}_{\text{PoA}}^*,
\end{align*}
and $\mathbf{y}_{i}^*$ similarly to the proof of \cref{thm:pi_c}.

For proving the extracted $v_{tot}^*$ and $v_{bin}^*$ satisfy \cref{eq:ZKPoA_1st_cond} and \cref{eq:ZKPoA_2nd_cond} with overwhelming probability, we refer to Theorem 3.1 of \cite{LNPS21}, and show that probabilities of \cref{eq:ZKPoA_1st_cond} and \cref{eq:ZKPoA_2nd_cond} do not hold when the verification passes are at most $\frac{2}{q^{d/l}}+\mathcal{O}(\epsilon)$ and $\frac{1}{q^{d/l}}$, respectively. 

In conclusion, a malicious participant who does not know the information of $(v_{bin}, v_{tot},\mathbf{r}_{\text{PoA}})$ can only convince the verifier by either breaking the commitments by solving $\msis_{\kappa,\kappa+\lambda+3,8\omega\mathfrak{s}\sqrt{2(\kappa+\lambda+3)d}}$ as shown in Lemma 4.3 in \cite{ALS20}, or by at most responding to specific challenges per commitment such that $\bar{c}$ is non-invertible. Overall, Protocol $\pi_{i}^A$ is sound with knowledge error as $\mathcal{O}\left(\frac{l+3}{q^{d/l}}\right)+\mathcal{O}(\epsilon)$.
\end{proof}
\section{Deferred Proofs for Security Analysis}
\label{app:sec_proof_analysis}

We restate the theorem for convenient. 

\begin{theorem}[Balanced]
    $\name$ is a balanced $\etl$ transaction scheme in ROM, assuming that NIZK is zero-knowledge and (validity) sound, BDLOP and ABDLOP commitment schemes are binding, and MSIS is hard.
\end{theorem}

\begin{proof}
We will show this theorem via a series of small changes made to the experiment.

\gamedes{0} The original privacy experiment in \cref{definition:exp_balance}.

\gamedes{1} Let ($\zksimulator_1, \zksimulator_2$) be the simulator for the ZK experiment of the underlying proof system.
The experiment instead of computing the proof normally using $\nizkprove_{\lang_{\relation}}(\crs,\stmt,\wit)$, it use the simulator $\zksimulator_2(\crs, \stmt, {\sf st})$ to generate the proof where $(\crs, {\sf st}) \gets \zksimulator_1(\secpar)$. Since the relations still hold, this game is indistinguishable from the previous game by the NIZK zero-knowledge property.

\gamedes{2} The experiment extracts the witness of zero-knowledge proofs produced by the adversary by running the knowledge extractor $\mathcal{E}$ for the corresponding relation and abort in the event that any of the witness extraction failed. From the validity knowledge soundness of NIZK, it follows that abort event happens with negligible probability $\negl$.

\gamedes{3} The experiment aborts in the event that for any of the commitment opening extracted from the proofs, there exists two different extracted openings for the same commitment. 
The abort event happens with negligible probability $\negl$ due to the binding property of the commitment scheme $\bindingadv$.

\gamedes{4} The experiment now sample the "public key" portion of the commitment key $\pk_{1,i}^T,\pk_2^T$ in $\commitkey_i$ uniformly at random that is $(\mathbf{\pk}^T_{1,i},\mathbf{\pk}^T_{2,i}) \sample R_q^{1 \times (\kappa+\lambda+3)} \times R_q^{1 \times (\kappa+\lambda+3)}$ for all not corrupted participant $i \in (\bank \setminus \bank_{C})$ generated during setup. The replacement is indistinguishable by the $\mlwe_{\kappa+\lambda+3,\kappa,\chi}$ assumption. 

\gamedes{5} The experiment aborts in the event for any asset $a \in \assetlist$, for any participant $i$, the extracted value is a negative amount $v^{*}_{\widehat t,a,i}<0$, but the adversary does not query the corruption oracle $\mathcal{O}_{c}$ on  participant index $i$, i.e. $i \notin \pklist_{cor}$. This is a signature forgery.
We claim the event happens with negligible probability, that is the following probability 
\begin{align*}
    |\Pr[\GAME_5(\adv)=1] -\Pr[\GAME_4(\adv)=1]| &\leq |\pklist| \cdot \mathsf{acc} \\ &\leq |\pklist|\cdot(\frac{Q_{eq}}{|\mathcal{C}|}+ \sqrt{Q_{eq}\cdot \mathsf{frk}})
\end{align*}
is negligible.

The claim is proved by constructing a reduction algorithm that find a solution to the MSIS problem $\msis_{\kappa+\lambda+3,\kappa+1,\beta}$ in the Hermite normal form (HNF). Given an instance of MSIS problem encoded as $\begin{bmatrix}
    \mathbf{A}^T||\pk_1
\end{bmatrix} := \mathbf{\widehat A} \sample \mathcal{R}_q^{(\kappa+\lambda+3) \times (\kappa+1) }$, we embed the instance into the scheme by replacing both the commitment key $\mathbf{A}$ and $\pk^T_{1,i}$ for some random participant index $i$ with the problem instance. The view remains the same for the adversary since the $\mathbf{A}$ and public key $\pk$ is uniformly random. Let     $\tilde{\bf A}=\begin{bmatrix}
        \bfA^T,\texttt{Id},0,0\\
        0,0,\bfA^T,\texttt{Id}
    \end{bmatrix}$. Then by invoking the forking lemma from \cref{lemma:forking_lemma} on the algorithm $\mathcal{B}$ that simulates similar to experiment of $\GAME_4$ but fork at $c_i$ and perform the correct bookkeeping, then with probability $\mathsf{frk}$, we obtain two forgeries $(\mathbf{v}_{1,i}^*, \widehat c_i^*, \mathbf{z}_{1,i}^*)$ and $(\mathbf{\widehat v}_{1,i}^*, \widehat c_i^*, \mathbf{\widehat z}_{1,i}^*)$ that verifies the verification defined in \cref{proto:pi_e_verification} where $\tilde{\mathbf{A}}\mathbf{z}_{1,i}^*-c_i^*
\begin{bmatrix}
\pk_{1,i}\\ \pk_{2,i}
\end{bmatrix}=\mathbf{v}_{1,i}^*$ and 
$\tilde{\mathbf{A}}\mathbf{\widehat z_{1,i}}^*-\widehat c_i^*
\begin{bmatrix}
\pk_{1,i}\\ \pk_{2,i}  
\end{bmatrix}=\mathbf{\widehat v}_{1,i}^*$. By the forking lemma, we have that $\mathbf{v}^*_{1,i}=\mathbf{\widehat v}^*_{1,i}$ since until the forking point the adversary view remains the same and that $c_i^* \neq \widehat c_i^*$. 
Let consider the upper half of the term ($\pk_{1,i}$) only and let define the term $(\mathbf{z}^*, \cdot ):=\mathbf{z}_{1,i}^*$ and $(\mathbf{\widehat z}^*, \cdot) := \mathbf{\widehat z}_{1,i}^*$ to denote the upper half of the response term as well.
Therefore, we have that $[\mathbf{A^T|I]}\mathbf{z}^*-c_i^*\pk_{1,i}= [\mathbf{A^T}|\mathbf{I}]\mathbf{\widehat z}^*-\widehat c_i^*\pk_{1,i}    
$
and consequently
$$
[\mathbf{A^T||I||}\pk_{1,i}] \begin{bmatrix}
    \mathbf{z}^*-\mathbf{\widehat z}^*\\ \mathbf{\widehat c}^* - \mathbf{c}^*
\end{bmatrix}
=0.
$$
Now recall that $[\mathbf{A^T}||\pk||\mathbf{I}]$ is an instance of (HNF) MSIS problem instance and that $\mathbf{\widehat c}^* - \mathbf{c}^* \neq 0$. We indeed found a solution that is bounded by $\beta \leq \sqrt{(2 \mathfrak{s}_1\sqrt{4(2\kappa+\lambda+3)d})^2 + 2\omega }$. 

Therefore, the probability of fork is $\mathsf{frk} \leq \advantage{}{\msis_{\kappa+\lambda+3,\kappa+1,\beta}}$.

In $\GAME_5$, we now show that none of the ledger violation conditions 1) 2) and 3) specified in Balance experiment can be fulfilled.
Let $\widehat t$ be the pending transaction.

For condition 1), the ownership bypass condition is aborted in $\GAME_5$. 

For condition 2), we have that due to the additive homomorphic property of BDLOP commitment scheme, $v_{a,i}^{sum}$ is committed under the summed commitment $\com_{a,i}^{sum}$. 
Due to the abort condition in $\GAME_2$ and the extraction, we extracted $\mathbf{e}_{1,i}^*$,$\mathbf{e}_{2,i}^*$, from the $\pi_{PoKW}$ such that  $\norm{(\mathbf{e^*_{\{0,1\},i}})}_{\infty} \leq 2\sqrt{2k}\mathfrak{s}_{3,PoKW}$. We first show that the adversary is bounded to the same extracted $(\mathbf{s_1^*,e_1^*,s_2^*,e_2^*)}$ in both proof of key well-formness and proof of equivalence. Let $\mathbf{A}^T \sample R_q^{(\kappa+\lambda+3) \times  \kappa} $ be a MSIS problem instance $\msis_{(\kappa+\lambda+3), \kappa, \beta}$. Without loss of generality we consider $\mathbf{s_1^*,e_1^*}$ since a similar strategy apply to $\mathbf{s_2^*,e_2^*}$. Let $\mathbf{\tilde s_1^*,\tilde e_1^*}$ and $\mathbf{\widehat s_1^*,\widehat e_1^*}$ (with $\bar c$ opening challenge difference) be the extracted secret key from PoKW and PoE respectively. Assume that $\tilde S=\mathbf{\tilde s_1^*||\tilde e_1^*} \neq \mathbf{\widehat s_1^*||\widehat e_1^*}=\widehat S$. We know that $[\mathbf{A}^T||\mathbf{I}]\tilde S - [\mathbf{A}^T||\mathbf{I}]\widehat S = \pk_1-\pk_1=0$.
Multiply in $\bar c$ and rearranging, we have $[\mathbf{A}^T||\mathbf{I}](\bar c\mathbf{\tilde S} - \bar c \mathbf{\widehat S})=0$.
Since $\norm{\bar c (\mathbf{\widehat s_1}||\mathbf{\widehat e_1})} \leq 2{\mathfrak{s}_{1,PoE}}\sqrt{2(4\kappa+2\lambda+6)d}= \beta_a$ from opening and  $\norm{\bar c (\mathbf{\tilde s}_1||\mathbf{\tilde e}_1)} \leq |\bar c|_1\cdot \sqrt{dm} \cdot \norm{\mathbf{\tilde s}_1||\mathbf{\tilde e}_1}_{\infty} \leq (2\omega)( \sqrt{d(2\kappa+\lambda+3)}2\sqrt{2k}\mathfrak{s}_{3,PoKW}) \allowbreak = \beta_b$. Therefore, we found a solution to MSIS instance with bound $\beta \leq \beta_a+\beta_b$ as desired.

Given the extracted summed commitment random value $\mathbf{r}_{a,i}^{sum} = \sum_{t\in \txlist \cup \widehat t} \mathbf{r}_{t,a,i}^*$ and the re-commitment random value $\mathbf{r}_{\widehat t,a,i}'^*$, the validity soundness of NIZK ensures the extracted value fulfill the relation $\relation_{PoC}$ such that $\norm{\mathbf{r}_{a,i}^{sum}}_{\infty} \leq (|\txlist|+1) 2\sqrt{2k}\mathfrak{s}_{3,PoC}$ and $\norm{\mathbf{r}_{\widehat t,a,i}'^*}_{\infty} \leq 2\sqrt{2k}\mathfrak{s}_{3,PoC}$. Consequently, we have $\norm{\dotprod{\mathbf{e}^*_{\{1,2\},i}}{\mathbf{r}_{a,i}^{sum}-\mathbf{r}_{\widehat t,a,i}'^*}}_{\infty} \leq (\kappa+\lambda+3)\cdot 2d(\beta_e\beta_r)$ where $\beta_e = 2\sqrt{2k}\mathfrak{s}_{3,PoKW}$ and $\beta_r =  (|\txlist|+2) 2\sqrt{2k}\mathfrak{s}_{3,PoC}$. Let the parameter be set such that $2\sqrt{2k}\mathfrak{s}_{3,PoC} < \sqrt q/4$ and $(\kappa+\lambda+3)\cdot 2d(\beta_e\beta_r) < \sqrt q/4$, 
Combined with the $\relation_{PoE}$ which ensures that $\norm{\tilde{C} \mathbf{m} -\tilde{u}}_{\infty} \leq 2\sqrt{2k}\mathfrak{s}_{3,PoE}$, we now apply \cref{lemma:comp_equil} to conclude that $v_{a,i}^{sum}=v'$ where $v'$ is committed in $\com'$. 
Then, by the validity soundness of NIZK for $\relation_{PoA}$, we checked that $v'$ is an integer with the range $[0, N)$ as required.

For condition 3), we need to check for a new transaction $\widehat t$, we have that the sum of transaction value for $\widehat t$ is 0 given any asset. That is for any $a \in \assetlist$,  $v_{\widehat t,a}^{sum}=\sum_{i \in \pklist} v_{\widehat t,a,i}=0$. This means that by appending $\widehat t$ to a ledger $\ledger$, the total asset value does not change.
Due to the additive homomorphic property of BDLOP commitment scheme, $v_{\widehat t,a}^{sum}$ is committed under the summed commitment $\com_{\widehat t, a}^{sum}$.
Due to the abort condition in $\GAME_2$, the proof of balance $\pi_{PoB}$ and the associated relation $\relation_{PoB}$, we have that it must be the case that for the pending transaction $\widehat t$ and any asset $a$, we have that $v_{\widehat t,a}^{sum}=0$ as required or else it is not cannot fulfill the relation of $\relation_{PoB}$.

Therefore, we conclude that the adversary has probability of 0 at winning the game $\GAME_5$.

\end{proof}

\begin{theorem}[Balanced (QROM)]
    $\name$ is a balanced $\etl$ transaction scheme in QROM, assuming that NIZK is zero-knowledge and quantum (validity) sound, BDLOP and ABDLOP commitment schemes are binding, and MSIS is hard.
\end{theorem}

\begin{proof}
    \textit{Sketch.} The proof proceeds similar to that of the ROM version. Instead of invoking proof of knowledge extractor, we invoke quantum proof of knowledge extractor since the protocols are proven to be quantum proof of knowledge in \cref{sec:qpok-pob-collapsing} and \cref{app:quantum_poc_poe_poa_etc}.  For the signature forgery, we invoke a similar "gentle measurement + restart" argument which is used in quantum proof of knowledge extractor for $\pi^{Eq}$ to extract a MSIS solution. Since the random oracle is not used in anywhere else, the scheme is now balanced in the QROM model.
\end{proof}

\begin{theorem}[Privacy]
    $\name$ is a private $\etl$ transaction scheme, assuming that commitment scheme is hiding, NIZK is zero-knowledge and MLWE sample is indistinguishable from uniform sample.
\end{theorem}

\begin{proof}
    We will show this theorem via a series of small changes made to the experiment.

\gamedes{0} The original privacy experiment $$\Pr[\GAME_0(\adv)=1] := \advantage{\expprivacy}{\adv}$$.

\gamedes{1} 
Let ($\zksimulator_1, \zksimulator_2$) be the simulator for the zero-knowledge experiment of the underlying proof system.
The experiment when creating transaction $\tx$ using $\createtx$, instead of computing the proof normally using $\nizkprove_{\lang_{\relation}}(\crs,\stmt,\wit)$, we use the simulator $\zksimulator_2(\crs, \stmt, {\sf st})$ to generate the proof where $(\crs, {\sf st}) \gets \zksimulator_1(\secpar)$. Since the relations still hold, this game is indistinguishable from the previous game by the NIZK zero-knowledge property. Let $|\pi|$ be the number of proofs per asset and per participant, we have:
$$|\Pr[\GAME_1(\adv)=1] -\Pr[\GAME_0(\adv)=1]| \leq |\pi| \sum_{a \in \assetlist} |\valuelist_a| \cdot  \sf{ZK}^{\adv}(\secpar).$$

\gamedes{2} The experiment when creating the key pairs using $\keygen$, sample $\pk$ from the corresponding uniform distribution. Since the adversary do not have access to the secret key, it is indistinguishable by the hardness of the MLWE,

$$|\Pr[\GAME_2(\adv)=1] -\Pr[\GAME_1(\adv)=1]| \leq |\pklist| \cdot {\mlwe_{\kappa+\lambda+3,\kappa,\chi}^\adv(\secpar)}.$$

\gamedes{3} The experiment when creating transaction $\tx$ using $\createtx(\valuelist, \sklist, \ledger)$, instead of committing to $\valuelist$, it commits to a random valid value list $\valuelist^\dagger \in \valueset$. It is indistinguishable by the hiding property of the commitment scheme,
$$|\Pr[\GAME_3(\adv)=1] -\Pr[\GAME_2(\adv)=1]| \leq \sum_{a \in \assetlist} |\valuelist_a| \cdot \hidingadv.$$

At this point, the output of $\createtx$ is independent of the value list $\valuelist$ input, and thus independent of $b$. Therefore, $\adv$ has probability of $1/2$ of outputting $b=b'$ in $\GAME_3$. 

\end{proof}

\section{Quantum Proof of Knowledge}
\label{app:quantum_pok}

\subsection{Preliminaries}
We restate some of the preliminaries here.  We denote $\lambda$ as the security parameter rather than the $\msis$ or $\mlwe$ parameter for the remainder of the paper. Let $\mathcal{R}_q=\mathbb{Z}_q[X]/(X^d+1)$ a cyclotomic ring, and let $C$ be a public-coin challenge space (e.g., short ternary polynomials) sampled by the verifier in a public-coin sigma protocol. We denote by $\negll(\lambda)$ a negligible function and by $\gamma(\lambda)>0$ a non-negligible function. All probabilities are taken over the internal randomness of the parties and the outcomes of quantum measurements.

\medskip

We consider two settings:
\begin{itemize}
  \item \textbf{Commitment-level:} A (possibly interactive) commitment scheme with public parameters $\mathsf{pp}$, and a set of valid openings/witnesses $\mathcal{W}(x,a)$ for an instance $x$ and commitment $a$.
  \item \textbf{Sigma-protocol-level:} A public-coin sigma protocol with first message $a$ (commitment), challenge $c\leftarrow C$, and response $z$. For fixed $(x,a,c)$, define the set of valid responses $\mathcal{W}(x,a,c)$ satisfying the protocol's verification relation.
\end{itemize}

\subsection{Collapsing and Weakly Collapsing Definition}
We formalize collapsing via two indistinguishability games that differ by whether the witness (or its injective image) is measured before the protocol continues.

\begin{definition}[Collapsing Game (sigma-protocol form)]
\label{def:collapsing-game}
Fix public parameters $\mathsf{pp}$, instance $x$, and first message $a$. Let $D$ be a quantum polynomial-time (QPT) distinguisher. Consider two experiments:
\begin{itemize}
\item \textbf{Unmeasured Game} $\mathrm{CollapsingGame}^0$: 
  \begin{enumerate}
    \item The challenger samples $c\leftarrow C$ and gives $(x,a,c)$ to $D$.
    \item $D$ prepares a (normalized) superposition over candidate responses 
    $$
    \sum_{z\in \mathcal{W}(x,a,c)} \alpha_z \,|z\rangle \otimes |\psi_z\rangle,
    $$
    and sends the $|z\rangle$ register to the challenger. The challenger does not measure $|z\rangle$ and returns it to $D$.
    \item $D$ outputs a bit $b_0\in\{0,1\}$.
  \end{enumerate}
\item \textbf{Measured Game} $\mathrm{CollapsingGame}^1$: 
  \begin{enumerate}
    \item The challenger samples $c\leftarrow C$ and gives $(x,a,c)$ to $D$.
    \item $D$ prepares $\sum_{z\in \mathcal{W}(x,a,c)} \alpha_z \,|z\rangle \otimes |\psi_z\rangle$ and sends the $|z\rangle$ register to the challenger. The challenger measures $|z\rangle$ in the computational basis (obtaining some valid $z$), returns the collapsed (post-measurement) state to $D$.
    \item $D$ outputs a bit $b_1\in\{0,1\}$.
  \end{enumerate}
\end{itemize}
\end{definition}

\begin{definition}[Collapsing]
\label{def:collapsing}
A sigma protocol (or commitment scheme) is \emph{collapsing} if for all QPT distinguishers $D$,
\begin{align*}
&\Big|\Pr[\mathrm{CollapsingGame}^0\text{ outputs }1] - \Pr[\mathrm{CollapsingGame}^1\text{ outputs }1]\Big| \\
& \;\le\; \negll(\lambda).
\end{align*}
\end{definition}

\begin{definition}[Weakly Collapsing]
\label{def:weak-collapsing}
A sigma protocol (or commitment scheme) is \emph{weakly collapsing} if there exists a non-negligible $\gamma(\lambda)>0$ such that for all QPT distinguishers $D$,
\begin{align*}
&\Pr[\mathrm{CollapsingGame}^1\text{ outputs }1] \\
&\;\ge\;
\gamma(\lambda)\cdot \Pr[\mathrm{CollapsingGame}^0\text{ outputs }1]
\;-\;
\negll(\lambda).
\end{align*}
We say the protocol is \emph{worst-case} weakly collapsing if the above holds for every fixed $(x,a)$ (and $c\leftarrow C$), and \emph{average-case} weakly collapsing if it holds for random instances/measures defined by a generation algorithm (e.g., $(x,w)\leftarrow \mathrm{Gen}(1^\lambda)$).
\end{definition}

\subsection{Ancilla-Only Measurement Variant}
Direct measurement of $|z\rangle$ may be overly destructive in a quantum setting. A standard collapsing instantiation measures an \emph{injective or bounded-image function} of $z$, coherently computed on an ancilla.

\begin{definition}[Measured Function Variant]
\label{def:clf-pob-rt-a}
A compatible lossy function (CLF) for PoB is a public generator
$$
\mathrm{CLF.Gen}(a,c,\mathrm{mode}) \;\longrightarrow\; f:\,\mathcal{Z}_{a,c}\to \mathcal{R}_q^{\,l},
$$
with parameters $l=\mathrm{poly}(\lambda)$, $p_c(\lambda)$, $\tau(\lambda)$ and the following properties:

\begin{itemize}
\item \emph{Preserving (lossy) mode}: with probability $\tau(\lambda)$, $|\mathrm{Im}(f)|\le p_c(\lambda)$ over $\mathcal{Z}_{a,c}$ (bounded image).
\item \emph{Injective mode}: With probability $1-\mathrm{negl}(\lambda)$, the sampled $f$ is injective on $\mathcal{Z}_{a,c}$.
\item \emph{Indistinguishability}: For any QPT distinguisher $D$ that chooses $(a,c)$ and is given $f$ sampled either from the preserving family or the injective family,
$$
\Big|\Pr[\mathcal{D}(f\leftarrow\mathsf{pres})=1]-\Pr[\mathcal{D}(f\leftarrow\mathsf{inj})=1]\Big| \;\le\; \mathrm{negl}(\lambda).
$$
\end{itemize}
Define a \emph{measured} experiment where the challenger applies a reversible circuit $U_f$ mapping $|z\rangle|0\rangle\mapsto |z\rangle|f(z)\rangle$, measures \emph{only} the ancilla $|f(z)\rangle$ in the computational basis, and uncomputes the ancilla back to $|0\rangle$ via $U_f^\dagger$.
\end{definition}

\begin{definition}[Weakly Collapsing (Ancilla-Only Measured)]
\label{def:weak-collapsing-ancilla}
A sigma protocol (or commitment scheme) is \emph{weakly collapsing} in the ancilla-only measured sense if for some non-negligible $\gamma(\lambda)>0$,
\begin{align*}
&\Pr[\mathrm{accept}\mid \mathrm{ancilla\ measured}] \\
&\;\ge\; 
\gamma(\lambda)\cdot \Pr[\mathrm{accept}\mid \mathrm{unmeasured}] 
\;-\; 
\negll(\lambda),
\end{align*}

where acceptance is the verifier's final decision event in the continuation of the protocol, and the partial measurement acts only on the ancilla as per Definition~\ref{def:clf-pob-rt-a}.
\end{definition}

\subsection*{Generalized Measurement Lemma from Boneh--Zhandry \cite{BZ13}}

Our quantum knowledge soundness proofs rely heavily on the generalized measurement lemma proved by Boneh and Zhandry \cite{BZ13}. This lemma bounds acceptance loss under partial measurements, and we use it for quantum rewinding.

\begin{lemma}[Generalized measurement lemma (Boneh--Zhandry)]
\label{lem:boneh-zhandry-pob-rt-a}
Let $\mathcal{H}$ be a finite-dimensional Hilbert space, and let $\mathsf{Alg}$ be a quantum algorithm modeled as a CPTP map
$\mathcal{E}:\mathsf{D}(\mathcal{H})\to\mathsf{D}(\mathcal{H})$ whose final decision (accept/reject) is determined by a two-outcome POVM $\{F_{\mathrm{acc}},F_{\mathrm{rej}}\}$ with $F_{\mathrm{acc}}+F_{\mathrm{rej}}=\mathbb{I}$. Fix an initial state $\rho\in\mathsf{D}(\mathcal{H})$ and define the baseline acceptance probability
$$
p \;:=\; \Pr[\mathrm{accept}] \;=\; \operatorname{Tr}\!\big(F_{\mathrm{acc}}\,\mathcal{E}(\rho)\big).
$$
Construct a modified algorithm $\mathsf{Alg}'$ by inserting, at any intermediate point in the computation, a partial measurement (POVM) $\{M_y\}_{y=1}^k$ with at most $k$ outcomes, acting on an ancilla register that is computed coherently from the main registers (and whose workspace is subsequently uncomputed). Conditioned on outcome $y$, the algorithm resumes with (possibly $y$-dependent) CPTP map $\mathcal{E}_y$, and then applies the same final decision POVM $\{F_{\mathrm{acc}},F_{\mathrm{rej}}\}$.

Let $p'$ denote the acceptance probability of $\mathsf{Alg}'$ on input $\rho$. Then for every input state $\rho$,
\begin{equation}
\label{eq:BZ-bound}
p' \;\ge\; \frac{1}{k(\lambda)}\, p \;-\; \negll(\lambda),
\end{equation}
where $\lambda$ is the security parameter (accounting for finite-precision or implementation errors), and the negligible term captures the approximation incurred by coherent ancilla computation/uncomputation and any bounded numerical rounding in the reversible circuits.
\end{lemma}

\subsection*{Unruh's Two–Projection Acceptance Lower Bound}
\label{sec:unruh-two-projection}

We state the two–projection (``challenge–only rewind'') acceptance lower bound due to Unruh, which underlies the cubic dependence in our two–challenge extractor. It is presented both in a general abstract form.

\begin{theorem}[Unruh's two–projection bound {\cite{Unr12}}]
\label{thm:unruh-two-projection}
Let $\mathcal{C}$ be a finite challenge set and let $\{| \psi \rangle\}$ be a unit vector in a Hilbert space $\mathcal{H}$. For each $c \in \mathcal{C}$, let $W_c : \mathcal{H} \to \mathcal{H}$ be a unitary operator and let $P_c : \mathcal{H} \to \mathcal{H}$ be an orthogonal projector. Define the single–run acceptance probability
$$
\varepsilon \;:=\; \mathbb{E}_{c \leftarrow \mathcal{C}} \big\| P_c \, W_c \, | \psi \rangle \big\|^2,
$$
and the two–projection (challenge–only rewind) acceptance probability
$$
\delta \;:=\; \mathbb{E}_{c_1,c_2 \leftarrow \mathcal{C}}
\big\| P_{c_2} \, W_{c_2} \, W_{c_1}^{\dagger} \, P_{c_1} \, W_{c_1} \, | \psi \rangle \big\|^2.
$$
Then
$$
\delta \;\ge\; \varepsilon^{\,3}.
$$
\end{theorem}

\subsection*{Remarks}
\begin{remark}[Relation to QPoK]
Weakly collapsing (Definition~\ref{def:weak-collapsing-ancilla}) suffices to enable the ``gentle measurement + restart'' extractor in quantum proofs of knowledge for public-coin sigma protocols: measuring an ancilla that encodes $f(z)$ does not reduce acceptance beyond a non-negligible multiplicative factor. When combined with 2-special soundness, high-probability invertibility of the challenge difference (e.g., $c-c'$ in $\mathcal{R}_q$), shortness of the extracted witness, completeness, and resettable black-box access to the prover's unitary (and its inverse), one obtains a QPoK extractor with success probability bounded below by $\frac{\gamma(\lambda)}{p(\lambda)}\cdot \varepsilon(\lambda)^3 - \negll\lambda$, where $p(\lambda)$ is a polynomial aggregating partial measurement and structural losses.
\end{remark}

\begin{remark}[Boneh--Zhandry Lemma]
The acceptance-preservation bound follows from the generalized measurement lemma of Boneh--Zhandry: a partial intermediate measurement with at most $k(\lambda)$ outcomes reduces acceptance probability by at most a factor $k(\lambda)$ (up to negligible terms). In collapsing instantiations, $k(\lambda)$ is the bounded image size $p_c(\lambda)$ in the preserving mode.
\end{remark}

\section{Quantum Proof of Knowledge for PoB $\pi^B$}
\label{sec:qpok-pob-collapsing}

We establish that the Proof of Balance (PoB) protocol is a quantum proof of knowledge (QPoK) under the weakly collapsing property.

\subsection{Notation and Protocol Recap}
We redefine some of the notations.
Let $\mathcal{R}_q = \mathbb{Z}_q[X]/(X^d+1)$ , 
and let $(X^d+1)$ factors modulo $q$ into $\widehat l$ irreducible polynomials of degree $\widehat d$ where $\widehat d = d/\widehat l$. 
. Let $C$ be the challenge space. For each transaction $t$, asset $a$, and participant $i$, the first message (commitment) is
$$
\alpha_{\mathrm{PoB},t,a,i} = (w_{t,a,i}, u_{t,a,i}),
$$
where $w_{t,a,i} = A y_{t,a,i}$ and $u_{t,a,i} = B^\top y_{t,a,i}$ for random $y_{t,a,i}$. The challenge is $c_{t,a,i} \leftarrow C$, and the response is
$$
z_{t,a,i} = y_{t,a,i} + c_{t,a,i} r_{t,a,i},
$$
where $r_{t,a,i}$ is the short witness. The verifier checks
$$
A z_{t,a,i} = w_{t,a,i} + c_{t,a,i} \mathrm{com0}_{t,a,i}, \qquad
B^\top z_{t,a,i} = u_{t,a,i} + c_{t,a,i} \mathrm{com3}_{t,a,i},
$$
and the rejection-sampling norm bound on $z_{t,a,i}$. For $w\in \mathcal{R}_q$, write $\|w\|_\infty$ for the maximum absolute value of its coefficients taken in $[-(q-1)/2,(q-1)/2]$; extend coefficient wise to vectors in $\mathcal{R}_q^k$ by $\|w\|_\infty:=\max_j\|w_j\|_\infty$. Fix an integer step parameter $m \geq 1$. We define the rounding operator $[\,\cdot\,]_{m} : \mathcal{R}_q \rightarrow \mathcal{R}_q$ as coefficient wise reduction to the nearest multiple of $m$ in the central interval $(-m/2,m/2]$, extended component-wise to vectors and matrices.

\begin{definition}[Negacyclic multiplication matrix]
\label{def:nega}
For $a \in \mathcal{R}_q$, the negacyclic multiplication matrix $\text{Rot}(a) \in \mathbb{Z}_q^{d\times d}$ is defined by $\text{Rot}(a)_{ij} = $ coefficient of $X^i$ in $X^j.a (\mod X^d +1)$.Its columns are the signed cyclic shifts of the coefficient vector of $a$.

\[
\text{Rot}(a)_{k,j} = (a)_{(k-j)\mod d}. \delta_{k,j} \qquad \delta_{k,j} := \begin{cases} +1 ~~k\geq j, \\ -1 ~~ k<j.
\end{cases}
\]
\end{definition}

Since $|\delta_{k,j}| = 1$, and the map $j\mapsto (k-j) \mod d$ is a bijection on $\{0,\ldots , d-1\}$ for every fixed $k$, each row of $\text{Rot}(a)$ is a signed permutation of the coefficient vector of $a$. In particular,
\begin{align}
\label{eq:nega}
&\|\text{Rot}(a)_{k,.}\|^2 = \sum_{j=0}^{d-1}|(a)_{(k-j)\mod d}|^2 = \|a\|_2^2,\\ \nonumber & \text{for every }k\in\{0,\ldots , d-1\}.
\end{align}

\subsection*{Ring-Based Perturbation Bound}
Let $\zb$ be the set of accepting responses $\z_{t,a} := (z^1_{t,a}, \ldots , z^h_{t,a})$ for a fixed challenge $c$. Suppose $\|\z_{t,a}\|_2 \leq B_z :=\mathfrak{s}\sqrt{2hd}$, where $h:=\mathrm{k}+\lambda+3$. Let $\e \leftarrow \mathcal{D}_{\mathfrak{s}}^{h}$ be a random row of the evaluation matrix. The following lemma bounds $\|\e.\z_{t,a}\|$, where $\e.\z_{t,a} = \sum_{i=1}^{h} e_iz^i_{t,a} \in \mathcal{R}_q$ is the inner product in $\mathcal{R}_q$.

\begin{lemma}
\label{ref:perturb}
Let $\z_{t,a} := (z^1_{t,a}, \ldots , z^h_{t,a}) \in R^h_q$ with
$$\|\z_{t,a}\|_2^2 = \sum_{i=1}^h \|z^i_{t,a}\|_2^2 \leq B_z^2,$$ and let $\e = (e_1, \ldots , e_h) \sim \mathcal{D}_{\mathfrak{s}}^{(\mathrm{k}+\lambda+3)}$. Define $w := \e.\z_{t,a} = \sum_{i =1}^h e_iz^i_{t,a} \in \mathcal{R}_q$. For any coefficient index $k \in \{0,\ldots , d-1\}$ and threshold $T > 0$,
\begin{equation}
\Pr[|w_k| > T] \leq 2 \exp{\left(-\frac{T^2}{2\mathfrak{s}^2B_z^2}\right)}.
\end{equation}
Consequently, for $T = \Delta := \mathfrak{s} B_z \sqrt{2\ln(2d/\delta_0)}$,
\begin{equation}
\Pr[\|w\|_{\infty} > \Delta] \leq \delta_0.
\end{equation}
\end{lemma}

\begin{proof}
Fix a coefficient index $k \in \{0, \ldots , d-1\}$. The $k$-th coefficient of $w = \sum_{i=1}^h e_iz^i_{t,a}$ is 

\begin{equation}
w_k = \sum_{i=1}^h \text{Rot}(z_i)_k \cdot \text{Coeff}(e_i) = \sum_{i=1}^h \sum_{j=0}^{d-1} \text{Rot}(z_i)_{k,j} (e_i)_j,
\end{equation}

where $\text{Coeff}(e_i)) \in \mathbb{Z}_q^d$ denotes the coefficient vector of $e_i$, and $\text{Rot}(z_i)_k$ is the $k$-th row of $\text{Rot}(z_i)$. 
Expanding the above equation, $w_k$ is a linear combination of the $dh$ independent random variables $(e_i)_j \sim \mathcal{D}_{\mathbb{Z},\mathfrak{s}}$ with deterministic coefficients $\text{Rot}(z_i)_{k,j} = (z_i)_{(k-j) \mod d}. \delta_{k,j}$.

Since each $(e_i)_j$ has mean $0$ and variance $\mathfrak{s}^2$, we obtain:

\begin{align*}
\text{Var}(w_k) &= \mathfrak{s}^2 \sum_{i=1}^h \sum_{j=0}^{d-1} |\text{Rot}(z_i)_{k,j}|^2\\
&=\mathfrak{s}^2 \sum_{i=1}^h \|\text{Rot}(z_i)_{k,.}\|_2^2\\
&= \mathfrak{s}^2 \sum_{i=1}^h \|z^i_{t,a}\|_2^2 ~~ \text{From Equation \ref{eq:nega}}\\
& \leq \mathfrak{s}^2B_z^2,
\end{align*}

By the Gaussian tail bound, we get:
$$\Pr[|w_k| > T] \leq 2 \exp{\left(-\frac{T^2}{2\mathfrak{s}^2B_z^2}\right)}.$$

Taking $T = \Delta = \mathfrak{s}B_z\sqrt{2 \ln(2d/\delta_0)}$ gives tail probability $\delta_0/d$ per co-ordinate. A union bound over $k \in \{0, \ldots , d-1\}$ gives $\Pr[\|w\|_{\infty} > \Delta] \leq \delta_0$.
\end{proof}

\begin{corollary}
\label{cor:del}
    Fix a failure probability $\delta_0 = 2^{-\lambda}$. For a matrix $E \sim D^h_{\mathfrak{s}}$ with i.i.d. rows, define
    \begin{equation}
    \Delta_l := \mathfrak{s} B_z\sqrt{2\ln(\frac{2ld|\zb|}{\delta_0})}.
        \end{equation}
    Then with probability at least $(1 - \delta_0)$ over $E$,
    \begin{equation}
    \forall \z_{t,a} \in \zb : \|E\z_{t,a}\|_{\infty} \leq \Delta_l.
    \end{equation}
\end{corollary}

\subsection{Weakly Collapsing Property for PoB }
\label{sec:weak-collapse-pob-detailed-rt-a}

We prove that the PoB sigma protocol is \emph{weakly collapsing} for the short witness $r_{t,a}$ (the aggregate randomness for transaction $t$ and asset $a$). 
After the prover measures the first message (commitment) $a$ at the checkpoint, the post–commitment workspace is $|\phi_{\alpha}\rangle$. For each public–coin challenge $c\in\mathcal{C}$, let $P_{a,c}$ the projector onto the accepting subspace for transcripts with first message $a$ and challenge $c$. Define the baseline (unmeasured) acceptance probability conditioned on $a$ by
$$
\varepsilon_{\alpha} \;:=\; \mathbb{E}_{c\leftarrow\mathcal{C}}\!\left[\,\big\|P_{a,c}\,|\phi_{\alpha}\rangle\big\|^2\,\right].
$$
Let $\zb\subseteq \mathcal{R}_q^{h}$ denote the set of valid responses for $(a,c)$ (the dimension $h$ is that of PoB's response register). Assume there exists a bound $B_z=\mathrm{poly}(\lambda)$ such that $\|\z_{t,a}\|_\infty\le B_z$ for all $\z_{t,a}\in\zb$ (this follows from the PoB acceptance checks and the standard rejection–sampling bounds).

\paragraph{Setup and Notation}
Fix $(t,a)$ and the PoB \emph{checkpoint} immediately before the final challenge $c_{t,a}\leftarrow C$ is sampled. The first message (commitment) at this point is
$$
\alpha_{\mathrm{PoB},t,a} \;=\; (w_{t,a}, u_{t,a}),\quad
w_{t,a} \;=\; A\, y_{t,a},\quad
u_{t,a} \;=\; B^\top y_{t,a},
$$
and the prover’s internal response register encodes a superposition over short $\z_{t,a}\in\zb$:
$$
|\Psi\rangle \;=\; \sum_{\z_{t,a}\in\zb} \alpha_z\, |\z_{t,a}\rangle \otimes |\psi_z\rangle,
$$
where $\|\z_{t,a}\|_2\le B_z$, and $|\psi_z\rangle$ is auxiliary. The baseline (unmeasured) acceptance probability from the checkpoint is denoted 

$\Pr[\mathrm{accept}\mid \mathrm{unmeasured}]$. The PoB verification equations are
\begin{align*}
&A \z_{t,a} \;=\; w_{t,a} \;+\; c_{t,a}\, \sum_{i\in\mathrm{PT}} \mathrm{com_0}_{t,a,i}, \\
&B^\top \z_{t,a} \;=\; u_{t,a} \;+\; c_{t,a}\, \sum_{i\in\mathrm{PT}} \mathrm{com_3}_{t,a,i},
\end{align*}
with response $\z_{t,a} = y_{t,a} + c_{t,a}\, r_{t,a}$ and a rejection-sampling norm bound on $\z_{t,a}$.

\begin{definition}[CLF for PoB]
\label{def:clf-concrete_pob}
Set 
\begin{align*}
    \Delta_l := \mathfrak{s} B_z\sqrt{2\ln(2ld|\zb|/\delta_0)}, \qquad \delta_0 \leq 2^{-8}, \qquad m=q/4
\end{align*}
\begin{align*}
    \widehat l \cdot h \cdot\log(4\beta_{z,\infty}+1) + \frac{\lambda}{\widehat d} \leq l \leq m/(2d\Delta_l)
\end{align*}
Let $A$ be the matrix $A$ used in the verification equation of the PoB.
Define two families:
\begin{itemize}
\item \emph{Preserving mode:} Sample $E\leftarrow \mathcal{R}_q^{\,l\times h}$ with i.i.d.\ rows $\e_i\leftarrow \mathcal{D}_{\mathfrak{s}}^{(\mathrm{k}+\lambda+3)}$ (discrete Gaussian, coefficientwise), $C \sample R_q^{l \times \kappa}$, and $\z_0\leftarrow \mathcal{R}_q^{\,l}$ uniformly. Output $f_{\mathrm{pres}}(\z_{t,a}):=[ (CA+ E)\z_{t,a}+\z_0]_m$. 
\item \emph{Injective mode:} Sample $B\leftarrow \mathcal{R}_q^{\,l\times h}$ and $\z_0\leftarrow \mathcal{R}_q^{\,l}$ uniformly. Output $f_{\mathrm{inj}}(\z_{t,a}):=[B \z_{t,a}+\z_0]_{m}$. 
\end{itemize}
\end{definition}

\emph{Indistinguishability:} Distinguishing $f_{\mathrm{pres}}$ from $f_{\mathrm{inj}}$ reduces to distinguishing $l$ parallel noisy linear forms from uniform in $\mathcal{R}_q$ (i.e., $l$–parallel MLWE 
where secrets are uniformly distributed in $\mathcal{R}_q$ and errors are sampled from the discrete gaussian $\mathcal{D}$ distribution
with parameter $\mathfrak{s}$), hence the two families are computationally indistinguishable.

\begin{lemma}
\label{lem:pres}
If $\delta_0 \leq 2^{-8}$, $ 6 \leq ld \leq m/(2\Delta_l)$, and $\Delta_l < m/2$, then for every $A$, $w_{t,a}$, $c_{t,a}$, $\com_{0,t,a,i}$, choosing $f_{\text{pres}}(X)$ as defined in \cref{def:clf-concrete_pob}, $f_{\text{pres}}$ is a constant function for the input space $\zb$ with non-negligible probability of at least $ 1/3$. 
\end{lemma}

\begin{proof}
For every valid $\z_{t,a} \in \zb$, passing the verification equation for any $w_{t,a}$, $A$, $c_{t,a}$, $com_0$ means satisfying $A\z_{t,a}=w_{t,a} \;+\; c_{t,a}\, \sum_{i\in\mathrm{PT}} \mathrm{com_0}_{t,a,i}$.
Therefore, $A\z_{t,a}= \mathfrak{c}_0$ for some constant value $\mathfrak{c}_0$ for all valid $\z_{t,a} \in \zb$ and similarly for $A\z_{t,a} + \z_0 = \mathfrak{c}_0 + \z_0$ where $\z_0 \sample \mathcal{R}_q^{l}$ across all valid $\z_{t,a} \in \zb$. 
The same argument applies for the constant function $C\cdot (A\z_{t,a}) + \z_0 = C\cdot\mathfrak{c}_0 + \z_0$ and similarly $[C\cdot (A\z_{t,a}) + \z_0]_m = [C\cdot\mathfrak{c}_0 + \z_0]_m$.

It can be seen that the preserving mode function is defined as $[ (CA+ E)\z_{t,a}+\z_0]_m= [ C(A\z_{t,a}) + z_0 + E\z_{t,a}]_m$. Therefore, we need to show that the function with an extra additive term $E\z_{t,a}$ introduced is still constant with respect to $\forall \z_{t,a} \in \zb$. 

Let $\mathsf{Safe} = \cup_i (i\cdot m- m/2 + \Delta_l, i\cdot m + m/2 - \Delta_l]$, we have that if $x_{i,j}\in \mathsf{Safe}$ where $\mathbf{x}=(x_1,\cdots,x_l)$ and $x_i=(x_{i,1},\ldots,x_{i,d})\in \mathcal{R}_q$, then for $\norm{E\z_{t,a}}_{\infty} \leq \Delta_l$, $[\mathbf{x}]_m = [\mathbf{x}+E\z_{t,a}]_m$.

For a $C \sample R_q^{l \times \kappa}$, and $\z_0\leftarrow \mathcal{R}_q^{\,l}$, we have that the probability of not in $\mathsf{Safe}$ be at most $2\Delta_l/m$, that is the probability that $CA\z_{t,a}+\z_0$ in $\mathsf{Safe}$ is at least $(1-2\Delta_l/m)^{ld}$.
Let $l\cdot d \leq (m/2\Delta_l)$, it follows that since $ld \geq 6$, $\delta_0 \leq 2^{-8}$ and $\norm{E\z_{t,a}}_{\infty} \leq \Delta_l$ with probability at least $1-\delta_0$, we have that 
$\Pr[constant] \geq (1-\frac{2\Delta_l}{m})^{ld} \cdot (1-\delta_0)  \geq (1-\frac{1}{ld})^{ld} \cdot (1-\delta_0) \geq \frac{1}{3}$.

\end{proof}

\begin{lemma}
Let m= q/4, for every $A$, $w_{t,a}$, $c_{t,a}$, $\com_{0,t,a,i}$, if $l \geq \widehat l \cdot h \cdot\log(4\beta_{z,\infty}+1) + \frac{\lambda}{\widehat d}$, then we have that choosing $f_{\text{inf}}(X)$ from \cref{def:clf-concrete_pob} is injective for the input space $\zb$ with overwhelming probability of at least $1-2^{-\lambda}$.
\end{lemma}

\begin{proof}
Let $x \mapsto (\widehat x_1,\ldots,\widehat x_{\widehat l})$ where $\widehat x_i \in \mathbb{F}_{q^{d/\widehat l}}^h$ represent (incomplete) NTT representation of slot $i$ for the entire vector. We now refer to these as its equivalence CRT representation.
It can be seen that a rounded collision means that $(Bx+z_0)-(Bx'+z_0) = B\bar x \in D_{m}^{dl}$ where $D_{m}^{dl}$ is the difference set applied across coefficient and $D_{m}=\{-(m-1),\ldots,m-1\}$.

Define $T(x) := \{i : \widehat x_i \neq 0\} $.
Let $V_{T(x)}= \{y\in R_q^l: \forall i \notin T(x), \forall j \in [l] , \widehat y_{i,j} =0\}$, by the linearity and bijection of the NTT transform, $V_{T(x)}$ preserves the NTT dimension $l(\widehat d\cdot  |T(x)|)$ in the coefficient domain. Let $\pi_I$ be the projection map that project onto $l(\widehat d\cdot  |T(x)|)$ coordinates such that $V_{T(x)} \rightarrow \mathbb{F}_q^{l(\widehat d\cdot  |T(x)|)}$ is a bijection and $I$ picks a set of $l(\widehat d\cdot  |T(x)|)$ coordinates. 

Since $B$ is uniformly sampled, it follows that its CRT representation $\widehat B_i$ is uniformly distributed as well. For any fixed $x$, $Bx$ is uniform over $V_{T(x)}$. Then, $\Pr[Bx\in D_{m}^{dl}] \leq \Pr[\pi_I(Bx)\in D^{l(\widehat d\cdot  |T(x)|)}_{m}]\leq (\frac{2m-1}{q})^{l(\widehat d\cdot  |T(x)|)} \leq (\frac{2m}{q})^{l(\widehat d\cdot  |T(x)|)}$.

For any short $x \in S$ where $ S := \{\z_{t,a}-\z'_{t,a} \in \mathcal{R}_q^{h} : z\neq z' \wedge \norm{\z_{t,a}}_{\infty} \leq B_{z,\infty} \wedge \norm{\z_{t,a}'}_{\infty} \leq B_{z,\infty} \}$,
we use union bound to obtain 
\begin{align*}
\Pr[\exists x \in S : Bx\in D_{m}^{dl}] &\leq \sum_{x\in S} (\frac{2m}{q})^{l(\widehat d\cdot  |T(x)|)} \\
&\leq |S| \cdot (\frac{2m}{q})^{l\cdot \widehat d} \\
&\leq (4\beta_{z,\infty}+1)^{d \cdot h} \cdot (\frac{2m}{q})^{l\cdot \widehat d}
\end{align*}

Since $m= q/4$, we have
\begin{align*}
    &\Pr[ f(x) =  [B\z_{t,a}+\z_0]_m \text{ is not injective} : B  \sample \mathcal{R}_q^{l \times h}, \z_0 \sample \mathcal{R}^{l}_q] \\
    & \leq (4\beta_{z,\infty}+1)^{d \cdot h} \cdot (\frac{2m}{q})^{l\cdot \widehat d}\\
    & \leq (4\beta_{z,\infty}+1)^{d \cdot h} \cdot (\frac{1}{2})^{l\cdot \widehat d}.
\end{align*}

It immediately follows that if $l \geq \widehat l \cdot h \cdot\log(4\beta_{z,\infty}+1) + \frac{\lambda}{\widehat d}$, then the function $f(x)$ is injective  with at least probability $1-2^{-\lambda}$.
\end{proof}

\begin{lemma}[Weakly collapsing PoB]
\label{lem:weak-collapse-pob-rt-a}
Under Definition~\ref{def:clf-pob-rt-a}, for any QPT adversary,
\begin{align*}
&\Pr[\mathrm{accept}\mid \mathrm{measured}] \\
&\;\ge\; \tau(\lambda) \cdot \Pr[\mathrm{accept}\mid \mathrm{unmeasured}] \;-\; \negll(\lambda),
\end{align*}
where “measured” denotes measuring $f_{t,a}(\z_{t,a})$ computed coherently on an ancilla and uncomputed thereafter. Equivalently, PoB is weakly collapsing with parameter $\gamma(\lambda)\ge \tau(\lambda)$.
\end{lemma}

\begin{proof}
(\emph{Step 1: Coherent computation and ancilla-only measurement.})
At the checkpoint, the extractor applies a reversible circuit $U_f$ that computes $f_{t,a}(\z_{t,a})$ coherently on a fresh ancilla:
$$
U_f:\ |\z_{t,a}\rangle \otimes |0\rangle \longmapsto |\z_{t,a}\rangle \otimes |f_{t,a}(\z_{t,a})\rangle
$$
(with auxiliary workspace folded into the ancilla). It then measures only the ancilla in the computational basis, obtaining outcome $y \in \mathrm{Im}(f_{t,a})$, and applies $U_f^\dagger$ to uncompute the ancilla to $|0\rangle$. The post-measurement projected state is
$$
|\Psi_y\rangle \;\propto\; \sum_{\substack{\z_{t,a}\in \mathcal{B}_{z}\\ f_{t,a}(\z_{t,a})=y}} \alpha_z\, |\z_{t,a}\rangle \otimes |\psi'_{\z_{t,a}}\rangle,
$$
and the classical protocol prefix remains unchanged.

(\emph{Step 2: Two experiments and acceptance event.})
Define the two experiments:
\begin{itemize}
\item Unmeasured: run the protocol from the checkpoint without intermediate measurement; denote acceptance probability $\Pr[\mathrm{accept} \mid \mathrm{unmeasured}]$.
\item Measured: perform the ancilla-only measurement of $f_{t,a}(\z_{t,a})$ (computed coherently and uncomputed thereafter), then run the protocol; denote $\Pr[\mathrm{accept} \mid \mathrm{measured}]$.
\end{itemize}
The event $E$ is “the verifier accepts the transcript.”

(\emph{Step 3: Preserving/lossy mode bound})
From Lemma \ref{lem:pres}, we get: condition on the preserving/lossy event $\varepsilon_{\mathrm{pres}}$
holds with probability at least $\tau(\lambda)$. Moreover, the ancilla measurement has at most one outcome on $\zb$. Applying Lemma~\ref{lem:boneh-zhandry-pob-rt-a} yields
\begin{align*}
&\Pr[\mathrm{accept} \mid \mathrm{measured}, \varepsilon_{\mathrm{pres}}]\\
&\;\ge\;
 \Pr[\mathrm{accept} \mid \mathrm{unmeasured}]
\;-\;
\negll(\lambda).
\end{align*}
This holds because the partial measurement projects onto the fiber $f_{t,a}^{-1}(y)$ and the acceptance event (the verifier’s final measurement) commutes with that projection up to negligible disturbance.

(\emph{Step 4: Injective mode behavior.})
With probability $1-\negll(\lambda)$, $f_{\mathrm{inj},t,a}$ is injective on $\zb$ Measuring $f_{t,a}(\z_{t,a})$ is then equivalent to measuring $\z_{t,a}$ itself (up to an invertible affine relabeling and rounding), so
\begin{align*}
&\Pr[\mathrm{accept} \mid \mathrm{measured}, \mathrm{inj\ mode}]\\
&\;=\;
\Pr[\mathrm{accept} \mid \mathrm{direct\ response\ measurement}],
\end{align*}
i.e., the measured acceptance matches the acceptance in a protocol that measures the response register.

(\emph{Step 5: Indistinguishability mixing.})
If the measured acceptance differed substantially between preserving and injective modes, one could distinguish $f_{\mathrm{pres},t,a}$ from $f_{\mathrm{inj},t,a}$, contradicting MLWE-based indistinguishability. Therefore,
\begin{align*}
&\Pr[\mathrm{accept} \mid \mathrm{measured}]\\
&\;\ge\;
\tau(\lambda)\cdot \Pr[\mathrm{accept} \mid \mathrm{unmeasured}]
\;-\;
\negll(\lambda),
\end{align*}
which is the claimed weakly collapsing bound with $\gamma(\lambda) \ge \tau(\lambda)$.
\end{proof}

\subsection{Quantum PoK for PoB (Main Theorem)}
\label{sec:qpok-pob-main}

\begin{theorem}[Quantum PoK for PoB]
\label{thm:qpok-pob-main}
Let PoB be as above. Assume:
\begin{enumerate}
    \item \textbf{Completeness:} The honest prover is accepted with probability $1-\negll(\lambda)$.
    \item \textbf{Shortness:} The witness $r_{t,a}$ is short (either structurally as an aggregate of short randomness or via an explicit range proof).
    \item \textbf{Invertibility:} For independent $c_{t,a},c'_{t,a}\leftarrow C$, 
    $$\Pr[(c_{t,a}-c'_{t,a})\text{ is non-invertible in }\mathcal{R}_q] \le \negll(\lambda).$$
    \item \textbf{2-special soundness:} From two accepting transcripts 
    
    $(\alpha_{\mathrm{PoB},t,a},c_{t,a},\z_{t,a})$ and $(\alpha_{\mathrm{PoB},t,a},c'_{t,a},\z'_{t,a})$ with $c_{t,a}\neq c'_{t,a}$ and $(c_{t,a}-c'_{t,a})$ invertible, one can extract
    $$
    r_{t,a}^\ast \;=\; (\z_{t,a}-\z'_{t,a})\,(c_{t,a}-c'_{t,a})^{-1},
    $$
    and $(x,r_{t,a}^\ast)$ is a valid witness.
    \item \textbf{Weakly collapsing:} There exists a non-negligible $\gamma(\lambda)>0$ such that
    $$
    \Pr[\mathrm{accept}\mid \mathrm{measured},a] \;\ge\; \gamma(\lambda)\cdot \Pr[\mathrm{accept}\mid \mathrm{unmeasured},a] \;-\; \negll(\lambda),
    $$
    where “measured’’ denotes measuring a compatible lossy function $f_{t,a}(\z_{t,a})$ computed coherently on an ancilla (and uncomputed thereafter) at the PoB checkpoint, and “unmeasured’’ denotes the baseline protocol.
\end{enumerate}
Let $p(\lambda)$ be a polynomial upper-bounding all polynomial losses in the extractor’s analysis, including:
\begin{itemize}
    \item the number of ancilla outcomes in the partial (gentle) measurement,
    \item the number of possible challenges,
    \item and any other polynomially-bounded factors arising from the protocol structure.
\end{itemize}
Define $\varepsilon_{\alpha}(\lambda)=\Pr[\mathrm{accept}\mid \mathrm{unmeasured},a]$ as the baseline acceptance probability, and $\varepsilon := \mathbb E_{\alpha} [\varepsilon_{\alpha}]$. Then there exists a QPT extractor $E^{\mathcal{A}}$ with resettable black-box access to the prover’s unitary $U,U^\dagger$ such that
$$
\Pr\!\left[(x,r_{t,a}^\ast)\in R\right]
\;\ge\;
\frac{\gamma(\lambda)}{p(\lambda)}\cdot \varepsilon^3
\;-\;
\negll(\lambda).
$$
\end{theorem}

\begin{proof}
\textbf{Checkpoint and extractor model.} We set the PoB \emph{checkpoint} immediately before sampling $c_{t,a}$. At the checkpoint, $\alpha_{\mathrm{PoB},t,a}$ and any public coins (if present) are classical and fixed, and the prover’s internal state encodes a superposition over $\z_{t,a}\in\zb$. The extractor has resettable black-box access to the prover’s unitary $U$ and inverse $U^\dagger$.

\textbf{Gentle ancilla-only partial measurement.} Instead of measuring $\z_{t,a}$ directly, the extractor computes a compatible lossy function $f_{t,a}(\z_{t,a})$ coherently on an ancilla with a reversible circuit $U_f$:
$$
U_f:\ |\z_{t,a}\rangle\otimes|0\rangle \longmapsto |\z_{t,a}\rangle\otimes|f_{t,a}(\z_{t,a})\rangle,
$$
measures only the ancilla (obtaining $y\in\mathrm{Im}(f_{t,a})$), and applies $U_f^\dagger$ to uncompute the ancilla to $|0\rangle$. This produces a projected state onto the fiber $f_{t,a}^{-1}(y)$ while leaving the classical prefix unchanged.

\textbf{Single-run acceptance bound.} By the generalized measurement lemma of Boneh--Zhandry (Lemma~\ref{lem:boneh-zhandry-pob-rt-a}), a partial measurement with at most $k(\lambda)$ outcomes degrades acceptance by at most a factor $k(\lambda)$. In our setting, $k(\lambda)\le 1$ (bounded image size in preserving mode). By the weakly collapsing lemma for PoB (Lemma~\ref{lem:weak-collapse-pob-rt-a}), we hence obtain
$$
\Pr[\mathrm{accept}\mid \mathrm{measured},a]
\;\ge\;
\gamma(\lambda)\cdot \varepsilon_{\alpha}
\;-\;
\negll(\lambda).
$$


\textbf{Restart to the checkpoint and two-run acceptance.} After the first accepting run with challenge $c_{t,a}$, the extractor samples an independent $c'_{t,a}\leftarrow C$. Similar to challenge $c_{t,a}$, the extractor performs a CLF measurement and record the classical response. We can also apply the generalized measurement lemma from Boneh-Zhandry \cref{lem:boneh-zhandry-pob-rt-a} for this second round as well. For the second round,

\begin{align*}
&\Pr[\text{two accept}|\text{both measured},a] \\&
\ge \gamma(\lambda) \Pr[\text{two accept}| \text{first round measured},a] - \negll(\lambda)
\end{align*}
By combining the two inequalities from these two rounds, we get
\begin{align}\nonumber
&\Pr[\text{two accept}|\text{both measured},a]\\
&\ge \gamma(\lambda)^2 \Pr[\text{two accept}| \text{unmeasured},a] - \negll(\lambda)
\end{align}

From Unruh's two acceptance lower bound \cref{thm:unruh-two-projection}, we get the following result.
\begin{equation}
\label{eq:two_acc}
\Pr[\text{two accept} | \text{unmeasured},a] \ge (\varepsilon_{\alpha})^3.
\end{equation}

By combining the last two inequalities, we get the following. 

\begin{align}\nonumber
&\Pr[\text{two accept} | \text{both measured},a] \\ \nonumber
&\ge \gamma(\lambda)^2\Pr[\text{two accept} | \text{unmeasured},a] - negl(\lambda)\\ \label{eq:unruh_low}
& \ge \gamma(\lambda)^2 (\varepsilon_{\alpha})^3 - negl(\lambda). ~~~~~~\text{From }\cref{eq:two_acc}
\end{align}

Aggregating the polynomial losses (bounded ancilla outcomes and any structural factors) into $p(\lambda)$, we obtain the two-run acceptance bound
\begin{align*}
\Pr[\text{two accept}|\text{measured},a] \;\ge\; \frac{\gamma(\lambda)^2}{p(\lambda)}\cdot (\varepsilon_{\alpha})^3 \;-\; \negll(\lambda).
\end{align*}

\textbf{Invertibility and extraction.} By the invertibility assumption, $(c_{t,a}-c'_{t,a})$ is invertible in $\mathcal{R}_q$ with probability $1-\negll(\lambda)$. Given two accepting transcripts with invertible difference, 2-special soundness yields
$$
r_{t,a}^\ast \;=\; (z_{t,a}-z'_{t,a})\,(c_{t,a}-c'_{t,a})^{-1}.
$$
Shortness ensures $(x,r_{t,a}^\ast)\in R$. Therefore by Jensen's inequality we get,

\begin{align*}
&\Pr[\text{two accept}|\text{measured}] \\
&\;\ge\; \frac{\gamma(\lambda)^2}{p(\lambda)}\cdot \mathbb{E}_{\alpha}[(\varepsilon_{\alpha})^3] \;-\; \negll(\lambda) \\
& \;\ge\; \frac{\gamma(\lambda)^2}{p(\lambda)}\cdot \mathbb{E}_{\alpha}[(\varepsilon_{\alpha})]^3 \;-\; \negll(\lambda) \\
&\;=\; \frac{\gamma(\lambda)^2}{p(\lambda)}\cdot \varepsilon^3 \;-\; \negll(\lambda).
\end{align*}

By completeness, we get $\varepsilon \ge (1 - \negll(\lambda))$. Therefore, the final acceptance test succeeds with probability $1-\negll(\lambda)$.

\textbf{Conclusion and runtime.} Combining the above, the extractor outputs a valid witness with probability at least $(\gamma(\lambda)/p(\lambda))\cdot \varepsilon(\lambda)^3 - \negll(\lambda)$. The rewinding involves a polynomial number of oracle calls to prover's strategy, CLF measurements and polynomial-time arithmetic in $\mathcal{R}_q$, so the extractor runs in expected quantum polynomial time.
\end{proof}

\section{Quantum Knowledge Soundness Proof Ideas for PoC, PoE, PoA}
\label{app:quantum_poc_poe_poa_etc}
The quantum knowledge soundness proofs for the other protocols like PoC, PoE, PoA are similar. 

All the other protocols are also sigma protocols, with $2$-special soundness property. Similar to the PoB, we can prove the weakly collapsing property of the other protocols as well with suitable 
compatible lossy functions (CLF) with parameters $\tau(\lambda)$ and $p_c(\lambda)$, yielding $\gamma(\lambda)=\tau(\lambda)/p_c(\lambda)>0$.
Ancilla-only acceptance preservation follows Boneh–Zhandry’s generalized measurement lemma \cite{BZ13}.
The extractor has resettable black-box access to the prover’s $U,U^\dagger$ and the verifier’s coin.

\emph{Global proof idea.}
At the checkpoint, we gently measure a bounded-image function of the witness on a coherently computed ancilla, ensuring acceptance only drops by a controlled factor (\cref{lem:boneh-zhandry-pob-rt-a}).
We then reset the computation and obtain two accepting transcripts with independent challenges.
Since challenge differences are invertible in $\mathcal{R}_q$ with high probability, 2-special soundness yields algebraic extraction via $(\z-\z')(c-c')^{-1}$.
Shortness and completeness make the final verification succeed with overwhelming probability.
This template applies to PoC, PoE, PoA and others with protocol-specific linear/quadratic relations.

\subsection*{NIZKPoK from Quantum Proof of Knowledge Protocol}

Interactive quantum proof of knowledge proven in the previous section can be made non-interactive via Fiat-Shamir transformation \cite{fiatshamir}. As proven in \cite{LZ19}, applying Fiat-Shamir transformation on a quantum proof of knowledge sigma protocol, we get NIZKPoK scheme for the associated relation in the QROM model. 

\begin{theorem}[NIZKPoK for $\name$ Protocols in the QROM]
Let $\lambda$ be the security parameter. Let $\pi^B$ (PoB) denotes the proof of balance protocol used in PQ-TaDL, and let $\pi^{\star}_{\FS}$ denote the Fiat--Shamir compilation of $\pi^\star$ using a quantum-accessible random oracle $H$.

Assume that the corresponding interactive protocol satisfy:
(i) $2$-special soundness with efficient extraction, and
(ii) weakly collapsing in the sense of Definition~\ref{def:weak-collapsing}.

Then each compiled protocol $\pi^{\star}_{\FS}$ is a non-interactive zero-knowledge proof of knowledge (NIZKPoK) for its associated relation in the QROM.

In particular, for PoB: let $P^*$ be any $\QPT$ prover making at most $Q_H=Q_H(\lambda)$ quantum queries to $H$ and outputting a proof $\pi=(a,z)$ that verifies under $\pi^B_{\FS}$ with probability $\varepsilon(\lambda)$. Then there exists a $\QPT$ extractor $E^{P^*}$ outputting a witness $w$ such that
\[
\Pr\big[(x,w)\in R_{\mathrm{PoB}} \,:\, w \leftarrow E^{P^*}(x)\big]
\;\ge\;
\frac{\gamma(\lambda)^2}{(Q_H+1)^3}\cdot \varepsilon(\lambda)^3 \;-\; \mathsf{negl}(\lambda),
\]
where $\gamma(\lambda)$ is the weak-collapsing parameter from Lemma~\ref{lem:weak-collapse-pob-rt-a}.
\end{theorem}

\begin{proof}
We prove the claim for $\pi^B$ (PoB). The extension to $\pi^C,\pi^{Eq}, \allowbreak \pi^A,$ and others is given at the end.

\paragraph{Step 1: Verifying Liu--Zhandry conditions (special soundness and weakly collapsing).}
Liu--Zhandry~\cite{LZ19} show that Fiat--Shamir in the QROM yields a proof of knowledge for sigma protocols satisfying (a) efficient special soundness and (b) a (weak) collapsing property capturing that intermediate measurements used by the extractor do not materially change acceptance probability.

\smallskip
\emph{(1a) $2$-special soundness.}
PoB is $2$-special sound with an efficient extractor: given two accepting transcripts
\[
(a,c,z)\quad\text{and}\quad(a,c',z')
\]
with the same commitment $a=(w_{t,a},u_{t,a})$ and distinct challenges $c\neq c'$, the paper’s classical special soundness theorem (Theorem~\ref{thm:g1}) provides an explicit polynomial-time extraction formula
\[
r^*_{t,a} \;=\; (c-c')^{-1}\,(z-z') \quad \text{in } \Rq,
\]
whenever $(c-c')$ is invertible (and the paper’s parameter setting ensures non-invertibility occurs only with negligible probability).

\smallskip
\emph{(1b) Weakly collapsing.}
The paper proves that PoB is weakly collapsing (Definition~\ref{def:weak-collapsing}) with parameter $\gamma(\lambda)$ (Lemma~\ref{lem:weak-collapse-pob-rt-a}). Concretely, measuring the appropriate extractor-side register (as formalized in Appendix~J via the CLF framework) changes the acceptance probability by at most a multiplicative $\gamma(\lambda)$ loss, up to $\negll(\lambda)$, with $\gamma(\lambda)$ explicitly lower-bounded in Lemma~\ref{lem:weak-collapse-pob-rt-a} (via Lemma~\ref{lem:pres}).

\paragraph{Step 2: NIZK extraction in the QROM.}
Let $P^*$ be a $\QPT$ prover against $\pi^B_{\FS}$ making at most $Q_H$ quantum queries to $H$ and outputting a (classical) proof $\pi=(a,z)$ for instance $x$. Let $\varepsilon$ denote the verifier acceptance probability.

\subparagraph{(2a) Decomposition at the critical query.}
In Fiat--Shamir, the challenge is $c := H(a,x)$ where $a$ is the prover’s commitment (first message) encoded in the proof. If $P^*$ produces an accepting proof with non-negligible probability, then except with negligible probability $P^*$ must query $H$ at input $(a,x)$ at least once during execution; otherwise $c$ is information-theoretically uniform from $P^*$’s perspective and the success probability is upper-bounded by the knowledge error (negligible for PoB due to the challenge space size).

Define the \emph{critical query index} $j^\star\in\{0,1,\dots,Q_H\}$ as the index of the first query at input $(a,x)$ along an accepting execution (if multiple such queries occur). The extractor will guess $j^\star$ uniformly, succeeding with probability at least $1/(Q_H+1)$.

\subparagraph{(2b) State at the critical query.}
Following the standard QROM extraction model (per Liu--Zhandry~\cite{LZ19}), view $P^*$ as a sequence of unitaries interleaved with oracle queries, and consider the joint state immediately before the $j^\star$-th oracle query. Conditioning on the extractor’s correct guess, the algorithm’s configuration can be written in the form required by the Liu--Zhandry $A_0/A_1$ abstraction:
\begin{itemize}
\item $A_0$ produces the commitment $a$ (classical in the proof output) and an intermediate quantum state $|\phi\rangle$ capturing the prover’s internal workspace immediately before receiving $c$.
\item $A_1$ takes $(|\phi\rangle,c)$ and continues the computation to produce a classical response $z$ (and hence the proof $\pi=(a,z)$).
\end{itemize}

Note that the joint pre-query state may be a superposition over commitments, e.g.,
\(
\sum_a \alpha_a\,|a\rangle|\phi_a\rangle
\).
Measuring the query input register projects onto a specific commitment $a$ with probability $|\alpha_a|^2$, leaving the corresponding conditional internal state $|\phi_a\rangle$ intact. The extractor proceeds conditioned on this specific $a$, with the final bound obtained by averaging over $a$ via Jensen’s inequality (Step~(2d)).

\subparagraph{(2c) Connection to the interactive QPoK (two-projection extraction on $|\phi\rangle$).}
The core quantum rewinding step is performed on the intermediate state $|\phi\rangle$, not on a classical response. Let $W_c$ denote the unitary part of the prover’s post-challenge computation (i.e., the coherent implementation of $A_1$ prior to measuring the output registers). Let $P_c$ denote the orthogonal projector that checks (coherently) whether the output registers encode an accepting PoB transcript for challenge $c$ and commitment $a$.

By the two-projection bound (Theorem~\ref{thm:unruh-two-projection}, following Unruh~\cite{Unr12}), if
\[
\varepsilon_\phi \;:=\; \E_{c\leftarrow\Cset}\big[\|P_c W_c|\phi\rangle\|^2\big]
\]
denotes the average acceptance probability from $|\phi\rangle$, then
\[
\E_{c_1,c_2\leftarrow\Cset}\Big[\big\|P_{c_2}W_{c_2}W_{c_1}^\dagger P_{c_1}W_{c_1}|\phi\rangle\big\|^2\Big]
\;\ge\;
\varepsilon_\phi^3.
\]
Intuitively, this yields (with probability $\ge \varepsilon_\phi^3$) two accepting transcripts for the same commitment $a$ under two independently chosen challenges $c_1,c_2$.

The weakly collapsing property (Lemma~\ref{lem:weak-collapse-pob-rt-a}, Definition~\ref{def:weak-collapsing}) is then used exactly as required by Liu--Zhandry~\cite{LZ19}: it ensures that the extractor’s intermediate measurements (needed to implement the QROM extraction strategy around the oracle interface) do not reduce acceptance probability by more than a multiplicative $\gamma(\lambda)$ per measurement event, up to negligible terms. In the PoB instantiation, the extractor performs (at most) two such measurement events in the standard two-transcript extraction, yielding a $\gamma(\lambda)^2$ factor.

\subparagraph{(2d) Extraction bound.}
The QROM extractor $E^{P^*}$ proceeds as follows (high-level, per Liu--Zhandry~\cite{LZ19}):

\begin{itemize}
\item Guess $j^\star \leftarrow \{0,\dots,Q_H\}$ uniformly.
\item Run $P^*$ with a simulated QROM until the $j^\star$-th oracle query, obtaining an intermediate state consistent with commitment $a$ and internal state $|\phi\rangle$.
\item Use the two-projection technique (Theorem~\ref{thm:unruh-two-projection}) to obtain two accepting responses for two independent challenges $c_1,c_2$ at the critical point, with the weak-collapsing loss bounded by $\gamma(\lambda)^2$.
\item From the two accepting transcripts, output the PoB witness by the $2$-special soundness extractor from Theorem~\ref{thm:g1}, namely
\[
r^*_{t,a} = (c_1-c_2)^{-1}(z_1-z_2).
\]
\end{itemize}

We now bound the overall success probability of the extractor.

For each fixed guess $j^\star$, let $\varepsilon_{j^\star}$ denote the acceptance probability conditioned on query $j^\star$ being the critical query. Conditioned on this event, measuring the query input register yields a commitment value $a$ with some distribution $\{p_a\}_a$, and the corresponding conditional post-measurement internal state is $|\phi_a\rangle$. Let $\varepsilon_{\phi_a}$ denote the average acceptance probability from $|\phi_a\rangle$ (as in Step~(2c)).

By Jensen's inequality (since $x^3$ is convex on $\mathbb{R}_{\ge 0}$),
\[
\sum_a p_a \cdot \varepsilon_{\phi_a}^3 \;\ge\; \left(\sum_a p_a \cdot \varepsilon_{\phi_a}\right)^3 \;=\; \varepsilon_{j^\star}^3.
\]

Averaging over the random guess $j^\star \leftarrow \{0,\dots,Q_H\}$, we can write the extractor’s two-transcript success probability (prior to the final special-soundness computation) as
\[
\Pr\big[E^{P^*}(x)\text{ reaches two accepting transcripts}\big]
\;=\;\]
\[
\frac{\gamma(\lambda)^2}{Q_H+1}\sum_{j=0}^{Q_H} \varepsilon_j^3 \;-\; \negll(\lambda),
\]
where $\varepsilon_j$ abbreviates $\varepsilon_{j^\star=j}$ and where the factor $\gamma(\lambda)^2$ accounts for the at-most-two weak-collapsing measurement events (Step~(2c)).

Next, by H\"older’s inequality,
\[
\sum_{j=0}^{Q_H} \varepsilon_j^3 \;\ge\; \frac{\left(\sum_{j=0}^{Q_H}\varepsilon_j\right)^3}{(Q_H+1)^2}.
\]
Moreover, $\sum_{j=0}^{Q_H}\varepsilon_j \ge \varepsilon$ since each accepting proof contributes to at least one critical query index. Combining these bounds yields
\[
\Pr\big[E^{P^*}(x)\text{ reaches two accepting transcripts}\big]
\;\ge\;\]
\[
\frac{\gamma(\lambda)^2}{(Q_H+1)^3}\cdot \varepsilon^3 \;-\; \negll(\lambda).
\]

Finally, conditioned on obtaining two accepting transcripts with distinct challenges $c_1\neq c_2$, the extractor computes the witness via the $2$-special soundness equation from Theorem~G.1, provided $(c_1-c_2)$ is invertible in $\Rq$. By Corollary~E.9, non-invertibility occurs with probability $O(l/q^{d/l})=\negl(\lambda)$ under the paper’s parameter setting. Incorporating this negligible loss gives
\[
\Prb\big[(x,w)\in R_{\mathrm{PoB}} \,:\, w \leftarrow E^{P^*}(x)\big]
\;\ge\;
\frac{\gamma(\lambda)^2}{(Q_H+1)^3}\cdot \varepsilon^3 \;-\; \negll(\lambda),
\]
which matches the theorem statement and is non-negligible for any non-negligible $\varepsilon$ since $Q_H=\text{poly}(\lambda)$.

\paragraph{Remark: Interface between interactive and NIZK extraction.}
The interactive extractor in Appendix~J is formulated in a \emph{resettable access} model and may interact with the prover’s coherent computation before any output measurement, which enables response-register techniques (e.g., CLF-based intermediate measurements) to establish weak collapsing (Lemma~\ref{lem:weak-collapse-pob-rt-a}) and thereby interactive QPoK (Theorem~\ref{thm:qpok-pob-main}). In contrast, the NIZK/QROM extractor operates at the level of the prover’s \emph{critical oracle query} and works with the intermediate state $|\phi\rangle$ in the Liu--Zhandry $A_0/A_1$ abstraction. The two settings are reconciled by the fact that Appendix~J establishes the specific weak-collapsing condition (Definition~\ref{def:weak-collapsing}) required by Liu--Zhandry~\cite{LZ19} to lift the sigma-protocol proof of knowledge to a Fiat--Shamir NIZKPoK in the QROM.

\paragraph{Extension to $\pi^C$, $\pi^{Eq}$, $\pi^A$, and others.}
The same argument applies to PoC, PoE, PoA and others, once the paper’s corresponding classical special soundness/extraction property and weakly collapsing property are established for each protocol (in the sense of Definition~\ref{def:weak-collapsing}), with the extractor using the protocol-specific special soundness equation in place of Theorem~\ref{thm:g1}’s PoB extractor. The QROM lifting step is protocol-agnostic: it depends only on (i) public-coin sigma structure, (ii) $2$-special soundness with efficient extraction, (iii) weakly collapsing as required by Liu--Zhandry~\cite{LZ19}, and (iv) the two-projection bound (Theorem~I.2), and it yields the same $(Q_H+1)^{-3}$ query-loss scaling (up to negligible terms) for these protocols as well.
\end{proof}
\newpage


\end{document}